\documentclass{article}
\RequirePackage{lmodern}
\RequirePackage{anyfontsize}

\renewcommand\normalsize{%
  \fontsize{10.4pt}{12.8pt}\selectfont}
\normalsize

\usepackage{amsmath,amsthm,amssymb}
\usepackage{algpseudocode, algorithm}
\usepackage[top=1.9cm, bottom=2.2cm, left=2.9cm, right=2.9cm]{geometry}
\usepackage{url}
\usepackage{booktabs}
\usepackage[colorlinks=true,bookmarks=false,citecolor=blue,urlcolor=blue]{hyperref} 
\usepackage{comment}
\usepackage{caption}
\usepackage{float}
\usepackage{subcaption}
\usepackage{graphicx}
\usepackage{enumitem}
\usepackage[backend=biber, maxbibnames=99]{biblatex}
\usepackage{mathtools}
\usepackage{authblk}
\usepackage{bbm}
\usepackage{etoolbox}
\usepackage{hyperref}
\usepackage{multirow}
\usepackage{fancyhdr}
\usepackage{ulem}
\usepackage{multicol}
\usepackage{setspace}
\title{\textbf{Optimal Execution under Liquidity Uncertainty}}

\author{
Etienne CHEVALIER\textsuperscript{*}
 $\,$
 Yadh HAFSI\textsuperscript{*}
 $\,$ 
 Vathana LY VATH\textsuperscript{\textdagger}
 $\,$ 
 Sergio PULIDO\textsuperscript{\textdagger\textrm{$\ddagger$}}
}
\newcommand{\midbar}[1]{\overline{\hspace{-0.12em}#1\hspace{-0.12em}}}

\newcommand{\nunder}[2][5]{\mathrlap{\mkern\the\numexpr#1/2mu\relax\underline{\phantom{\mathrm{#2}\mkern-#1mu}}}#2}

\newcommand{\cA}{\mathcal{A}}

\newcommand{\cC}{\mathcal{C}}
\newcommand{\cD}{\mathcal{D}}
\newcommand{\cE}{\mathcal{E}}
\newcommand{\cF}{\mathcal{F}}

\newcommand{\cI}{\mathcal{I}}
\newcommand{\cJ}{\mathcal{J}}
\newcommand{\cK}{\mathcal{K}}
\newcommand{\cL}{\mathcal{L}}

\newcommand{\cN}{\mathcal{N}}
\newcommand{\cO}{\mathcal{O}}

\newcommand{\cS}{\mathcal{S}}

\newcommand{\EE}{\mathbb{E}}

\newcommand{\NN}{\mathbb{N}}

\newcommand{\PP}{\mathbb{P}}

\newcommand{\RR}{\mathbb{R}}

\renewcommand{\epsilon}{\varepsilon}

\renewcommand{\P}{\PP}

\newcommand{\E}{\EE}

\def \os \hat{o}

\newcommand{\R}{\mathbb R}

\newcommand{\prob}{\mathbb{P}}

\newcommand{\expec}{\mathbb{E}}

\newcommand{\G}{\mathcal{G}}

\newcommand{\ud}{\mathrm d}

\def \I{\mathbb{I}}

\def \R{\mathbb{R}}

\def \E{\mathbb{E}}

\def \P{\mathbb{P}}

\def \Ac{{\cal A}}

\def \Lc{{\cal L}}

\def \eps{\varepsilon}

\def\ref#1{{\rm(\ref{#1})}}

\def\1{\mbox{1\hspace{-0.25em}l}}

\def \Ac{{\cal A}}

\def \Lc{{\cal L}}

\begin{document}
\maketitle
\pagenumbering{arabic}
\newcounter{axiom}
\newtheorem{axiome}[axiom]{Axiome}
\newtheorem{defi}{Definition}[section]
\newtheorem{theo}{Theorem}[section]
\newtheorem{prop}{Proposition}[section]
\newtheorem{corro}{Corollary}[section]
\newtheorem{rque}{Remark}[section]
\newtheorem{example}{Example}[section]
\newtheorem{nota}{Notation}[section]
\newtheorem{demo}{Demonstration}
\newtheorem{propt}{Propriété}
\newtheorem{lemma}{Lemma}
\newtheorem{assump}{Assumptions}[section]
\newcounter{cases}
\newcounter{subcases}[cases]
\newenvironment{mycase}
{
    \setcounter{cases}{0}
    \setcounter{subcases}{0}
    \newcommand{\case}
    {
        \par\indent\stepcounter{cases}\textbf{Case \thecases.}
    }
    \newcommand{\subcase}
    {
        \par\indent\stepcounter{subcases}\textit{Subcase (\thesubcases):}
    }
}
{
    \par
}
\renewcommand*\thecases{\arabic{cases}}
\renewcommand*\thesubcases{\roman{subcases}}

\begingroup
\renewcommand{\thefootnote}{\fnsymbol{footnote}}
\footnotetext[1]{Laboratoire de Math\'ematiques et Mod\'elisation d'Evry, Universit\'e Paris-Saclay, UEVE, UMR 8071 CNRS, France; email: etienne.chevalier@univ-evry.fr, yadh.hafsi@universite-paris-saclay.fr.}
\footnotetext[2]{Laboratoire de Math\'ematiques et Mod\'elisation d'Evry, Universit\'e Paris-Saclay, ENSIIE, UEVE, UMR 8071 CNRS, France; email: vathana.lyvath@ensiie.fr, sergio.pulidonino@ensiie.fr.}
\footnotetext[3]{The work of Sergio Pulido benefited from the financial support of the research program Chaire Deep learning in finance and Statistics (Fondation de l’École polytechnique, École
polytechnique and Qube R\&T).}
\endgroup

\begin{abstract}
We study an optimal execution strategy for purchasing a large block of shares over a fixed time horizon. The execution problem is subject to a general price impact that gradually dissipates due to market resilience. We allow for general limit order book shapes to characterize instantaneous market impact. To model the resilience dynamics, we introduce a stochastic process that governs the rate at which the deviation between the impacted and unaffected prices decays. This \textit{volume-effect} process reflects fluctuations in market activity that drive the pace of liquidity replenishment. Additionally, we incorporate stochastic liquidity variations through a regime-switching Markov chain to capture abrupt shifts in market conditions. We study this singular control problem, where the trader optimally determines the timing and rate of purchases to minimize execution costs. The associated value function to this optimization problem is shown to satisfy a system of variational Hamilton–Jacobi–Bellman inequalities. Moreover, we establish that it is the unique viscosity solution to this HJB system and study the analytical properties of the free boundary separating the execution and continuation regions. To illustrate our results, we present numerical examples under different limit-order book configurations, highlighting the interplay between price impact, resilience dynamics, and stochastic liquidity regimes in shaping the optimal execution strategy.
\end{abstract}\hspace{10pt}
\\
 \noindent\textbf{Keywords:} Optimal Execution, Singular Control, Viscosity Solutions, Free Boundary Problem, Regime Switching.\\\\
\noindent
\textbf{Mathematical subject classifications: }93E20, 49L25, 91B70.

\section{Introduction}
The problem of optimal execution in financial markets has been extensively studied over the past decades, with early research focusing on deterministic and linear price impact models. Classical approaches assumed that market participants could execute large trades with predictable and stable costs, often modeling price impact as a linear function of trading volume as in \textcite{BERTSIMAS19981} and \textcite{Almgren2012}. These models provided foundational insights into the trade-off between execution cost and market risk, establishing a framework in which traders optimize their execution trajectories over a fixed time horizon. However, empirical studies such as \textcite{BOUCHAUD200957} and \textcite{Zhou_impact}, have shown that the assumption of a simple linear price impact does not accurately reflect real-world market dynamics. Market impact is often nonlinear, transient, and influenced by complex interactions within the limit order book (LOB).

As a result, recent research has moved beyond simplistic assumptions, incorporating more advanced models that better capture this non-linearity. \textcite{OBIZHAEVA20131} introduced a framework where market impact depends explicitly on the LOB shape, showing that execution strategies should adapt to liquidity fluctuations.
\textcite{general_shape} extended this idea by considering a continuous LOB model where price impact depends on both the execution rate and liquidity replenishment. These advancements laid the foundation for further research on execution strategies, notably \textcite{pemy2007,alfonsi2016dynamic, moreau_JMK, ekren2019, muhle_karbe2020}. A significant contribution in this direction was made by \textcite{predoiu2011}, who analyzed optimal execution in a one-sided limit order book with nonlinear price impact and resilience. Their model, which accommodates an arbitrary order book shape, derives the optimal strategy as a sequence of lump-sum trades interspersed with continuous trading at a rate matching order book resiliency. These results formalized the interaction between execution strategies and market recovery, providing a framework for modeling price impact in order-driven markets. However, the assumption of deterministic resilience limits its applicability, as stochastic fluctuations in liquidity have been well documented; see for instance in \textcite{taranto_impact}.

Building on this, \textcite{Becherer2018} studied a two-dimensional singular control problem of finite fuel type over an infinite time horizon, motivated by the optimal liquidation of an asset position in a financial market with multiplicative and transient price impact. This framework incorporates stochastic liquidity, where the \textit{volume effect} process is driven by its own random noise, and they derive an explicit closed-form solution under the assumption of multiplicative price impact. \textcite{Schoneborn_fruth2017} introduced a stochastic order book depth, addressing the limitations of models that assume constant or deterministically varying liquidity. They establish the existence and uniqueness of optimal execution strategies when order book depth follows a general diffusion process. Unlike prior work assuming deterministic liquidity variations, this framework captures stochastic liquidity fluctuations. More recently, \textcite{Ackermann} explored optimal execution in a market where order book depth and resilience evolve stochastically. The authors formulated a discrete-time trading model and derive a recursive equation for the minimal expected execution cost. The authors analyzed conditions under which traders should execute immediately or gradually. \textcite{Horst2019} studied the optimal liquidation of multi-asset portfolios in markets where resilience is stochastic. The authors developed a framework that accounts for both immediate and persistent price impacts, formulating the problem using backward stochastic Riccati differential equations (BSRDEs) to characterize the value function. To handle the singular terminal conditions in these equations, they introduced an approximation method via penalization, solving a sequence of unconstrained problems that enforce full liquidation. Their results provide explicit optimal trading strategies. Finally, \textcite{fouque_jaimungal} examined optimal trading in markets where price impact is stochastic and often exhibits fast mean reversion. Using singular perturbation methods, the authors developed approximations to the resulting optimal control problem and established their accuracy by constructing sub- and super-solutions. The analysis highlights how stochastic trading frictions influence execution strategies. \textcite{muhlekarbe2022optimal} demonstrate that in concave price impact models a tractable linear approximation emerges in the high-frequency limit. The authors show that a stochastic liquidity parameter effectively captures the model’s nonlinearity, yielding a diffusion limit that preserves key features and performs well on limit order book data.

In this work, we introduce a stochastic volume effect governed by a jump diffusion process, allowing the model to capture both continuous liquidity variation and sudden market shocks. The price impact function is allowed to be discontinuous, accommodating abrupt changes in market conditions. Instead of specifying a particular model for the fundamental price, we require only a martingale property, making the framework applicable under general market dynamics. Liquidity regimes are incorporated through a finite-state Markov chain, reflecting structural shifts in market depth. We formulate the resulting optimization as a singular stochastic control problem. Unlike most singular stochastic control problems, which are typically posed in infinite horizon settings (see, e.g., \textcite{shreve1994}, \textcite{bank2001}), our problem is inherently finite-horizon due to the nature of execution tasks. In the execution literature, problems are often formulated either in discrete time or as impulse controls over a finite horizon, highlighting the distinctive structure of our continuous-time formulation. Moreover, our model avoids explicit cost functions often used in classical formulations, as our setting focuses on liquidity-induced frictions rather than inventory or running penalties. To the best of our knowledge, the closest related work is that of \textcite{predoiu2011}, which addresses a similar execution setting but in a static framework, whereas our approach is fully dynamic.

We incorporate different market regimes to model liquidity variations more realistically. Empirical evidence, such as \textcite{chevalier2023uncovering}, shows that market liquidity distribution varies significantly throughout the trading day. Studies like \textcite{pemy2008,BayraktarLudoExecution,COLANERI20201913,SIU201917,Dammann2023,Cartea2023,chevalier2024} have addressed this gap by modeling liquidity as a stochastic process. To capture this, we model liquidity regimes as a finite state Markov chain. In our framework, liquidity regime variations affect only the shape of the limit order book, while the underlying \textit{volume effect} dynamics remain unchanged. This ensures that price evolution is not exogenously driven by regime shifts but instead results from structural changes in supply and demand at different liquidity levels. By allowing the LOB’s form to vary across regimes, we account for regime-dependent price impact, where trades in low-liquidity states induce greater price responses than in high-liquidity conditions. This approach captures transient liquidity shocks, aligns with empirical evidence, and provides a more precise characterization of execution risk.

Finally, we address an optimal execution problem formulated as a singular control problem. Our approach is guided by techniques from \textcite{eliott}, \textcite{Liu18052016}, and \textcite{shiao_tian}, which connect regime-switching models with the first component of the Hamilton-Jacobi-Bellman Quasi-Variational Inequalities (HJBQVI). These works reformulate the execution problem as a stochastic control problem for a Markov jump linear system, deriving the value function and optimal feedback strategy through coupled differential Riccati equations. We also formulate our study as a free boundary problem in the spirit of \textcite{Caffarelli1977}, \textcite{pham1997optimal} and \textcite{FigalliSerra2019}, focusing on obstacle-type formulations and regularity properties. Due to the multi-dimensionality of the state space, classical verification techniques are not applicable, as the value function and optimal control cannot be characterized in closed form. To address this, we adopt a viscosity solution framework and characterize our value function as the unique solution to the HJB system. This allows us to describe the structure of the continuation and intervention regions, and to show the existence of a connected free boundary separating them. Due to the lack of sufficient regularity properties inherent in our setting, we rely on approximation methods to construct the solution. We illustrate our findings with numerical examples across various limit order book configurations. The viscosity solution is approximated using an implicit-explicit finite difference scheme for the PIDE, following \textcite{finite_diff}. The simulations show that higher market activity reduces market impact. In particular, increased volatility and resilience in the volume effect lead to greater liquidity and lower exercise boundaries, while the opposite holds in less active markets. We also examine how different limit order book shapes influence the free boundary.

The article is organized as follows. Section \ref{sec:control_problem} introduces the model describing the price and liquidity dynamics of a single asset, with a LOB featuring stochastic resilience and accounting for the stochastic market liquidity. This section also formulates the singular control problem associated with the execution strategy. Section \ref{Characterization of the Value Function} examines the analytical properties of the value function and establishes its continuity over the domain. In Section \ref{viscosity_section}, we demonstrate that the value function uniquely solves the system of variational inequalities as a viscosity solution and prove the connectedness of the free boundary. We present the numerical scheme to approximate the viscosity solution in Section \ref{numerical_scheme}. Finally, Section \ref{approximation_numerical} presents numerical illustrations, analyzing the influence of state variables and price impact on the value function and the free boundary.
\begin{nota}
    For $x \in \mathbb{R}^k$, where $k$ is determined by the context, $\|x\|$ denotes its Euclidean norm, and $B_r(x)$ represents the open ball centered at $x$ with radius $r > 0$. The scalar product is denoted by $\langle \cdot, \cdot \rangle$, and for a vector $x \in \mathbb{R}^k$, its transpose is denoted by $x^T$. For a set $A \subset \mathbb{R}^k$, $\partial A$ denotes its boundary, while $\mathrm{int}(A)$ or equivalently $\mathring{A}$ denote its interior. The symbol $\otimes$ represents the Kronecker product of two matrices. For a function $\varphi : \mathbb{R}_+ \times \mathbb{R}^d \times \mathbb{R}^k \to \mathbb{R}$, the gradient and Hessian matrix are denoted by $D\varphi$ and $D^2\varphi$, respectively, whenever they are well-defined.
\end{nota}
\section{Optimal Execution Problem} \label{sec:control_problem}
Consider a complete filtered probability space $(\Omega ,\mathbb{F}, \mathcal{F} = \{\mathcal{F}_t\}_{t \geq 0},\mathbb{P})$, where $\{\mathcal{F}_t\}_{t \geq 0}$ is a right-continuous filtration. We assume that the filtration $\{\mathcal{F}_t\}_{t \geq 0}$ supports a standard $\mathbb{P}$-Brownian motion $W$, a continuous Markov chain $I$, and a random Poisson measure $M$ on $\RR_+ \times \RR_+$ with compensator $\lambda_t \nu(\ud z)\mathrm{d}t$, where 
    $\lambda : [0, T] \to [0, \midbar{\lambda}]$, and $\nu$ is a measure on $\RR$. Throughout this paper, unless explicitly stated otherwise, any equality between random variables is understood to hold $\mathbb{P}$-almost surely when it is not explicitly stated.
\subsection{Model Setup}
\label{sec:model}
 We consider $\{I_t\}_{t \geq 0}$  to be a continuous-time Markov chain with values in a finite state space $\I_m := \{1, \ldots, m\} \subset \mathbb{N}$ and with càdlàg sample paths. We denote $\mathcal{F}^I = \{\mathcal{F}_t^I\}_{t \geq 0}$ the natural filtration generated by $I$, augmented by the $\mathbb{P}$-null sets. The process $I$ represents the liquidity state within the limit order book. Its infinitesimal generator $Q(t) = (Q_{ij}(t))_{1 \leq i,j \leq m}$ captures the intensities of transitioning between different liquidity states. The off-diagonal elements $Q_{ij}(t)$ for $i \neq j$ correspond to the continuous and bounded instantaneous rate of transition from state $i$ to state $j$. We assume that $I$ is stable and conservative, i.e., for all $(j,k)\in \I_m^2$ and $t\geq 0$,
\begin{equation}
\label{assump_mc}
Q_{jj}(t) = -\sum_{k \neq j} Q_{jk}(t),~~\text{and}~~ Q_{jj}(t) < +\infty.\end{equation}
\paragraph{Volume Effect.} Let $i\in \I_m$ be a liquidity regime. We define $A := (A_t)_{t\geq 0}$ as the reference price of the assets, which we
assume to be a continuous $(\prob,\mathcal{F})$-martingale. In our model,
we assume that, in the absence of trading, the number of available
shares at time $t$ in the price interval $[A_t, A_t+x[$ is $F_{I_t}(x)$. Here, $F_i$ is a
non-decreasing and left-continuous function associated to an infinite measure $\mu_i: \mathbb{R}_+\rightarrow \mathbb{R}_+$ in the following way:
\begin{equation*}
F_i(x):=\mu_i([0,x[),\quad\forall x\in\RR_+.
\end{equation*}
The shape of the \textit{shadow limit order book} (LOB) dynamically evolves as the liquidity regime changes. This allows the model to adapt to varying market conditions.
\begin{rque}
    This model is general and accommodates various limit order book structures and market impacts. Examples include block, modified block, and discrete order books (see \textcite{predoiu2011}). It also supports other types, such as power-law order books (see \textcite{power_impact} and \textcite{Avellaneda_stoikoiv}).
\end{rque}
 Let $Y$ be an $\mathcal{F}$-adapted nonnegative process representing the \textit{volume effect} process. Economically, $Y$ measures how much liquidity has been temporarily removed from the ask side of the LOB. This includes both the liquidity consumed by the trader and the liquidity taken by other market participants. A high value of $Y$ means  that the book is thinner and that buying becomes more expensive. The process $Y$ evolves according to
\begin{equation}\label{volume_effect_process}
\left\{\begin{array}{ll}
  \mathrm{d}Y_u=-h(Y_{u^-})\mathrm{d}u+\sigma(Y_{u^-})\mathrm{d}W_u + \int_{\R} 
  q(Y_{u^-}, z) M (\ud u, \ud z),\\Y_{t^-}=y,
\end{array} \right.
\end{equation}
with $t\geq 0$, $u\in[t,T]$, $y\geq 0$, $h:\RR_+\rightarrow\RR$, $\sigma:\RR_+\rightarrow\RR$ and $q:\RR_+\times\RR\rightarrow\RR$. We denote $$\check{Y}_{u^-} := Y_{u^-} + \Delta_{M} Y_u,$$
where $\Delta_{M} Y_u$ represents the jump contribution of the measure $M$ to $Y$ at time $u$. 

Each term in \eqref{volume_effect_process} has the following interpretation: The drift term $-h(Y_{u-})$ models the rate at which
the LOB refills after liquidity is consumed. The diffusion term $\sigma(Y_{u-})\,\ud W_u$ captures small fluctuations in the volume displaced from the shadow LOB $F_{I_u}$. These fluctuations are caused by the continuous flow of limit orders and cancellations. A larger value of
$\sigma$ corresponds to a higher level of market activity, in the sense that the
available volume fluctuates more rapidly during periods of intense trading. The jump term $\int_{\R} 
  q(Y_{u^-}, z) M (\ud u, \ud z)$ models abrupt liquidity shocks generated by block trades or large market orders submitted by other participants.
\begin{assump}
\label{no_mutual_cov}
 We assume that the processes $I$ and $M$ do not have mutual covariation. 
 \end{assump}
 \begin{rque}
     Assumption \ref{no_mutual_cov} simplifies the analysis without minimal loss of generality, as any dependence can typically be absorbed into the dynamics of $M$ or handled via a change of measure.
 \end{rque}
  The ask price is no longer the reference price but is
given by $P := A + D$. Here, the process $D$ defines the price deviation process, such that $$D_u :=\psi_{I_u}(Y_u), \quad \forall u\geq 0$$ In plain terms, $D$ measures how far the impacted ask price is from the reference/equilibrium price $A$. A larger value of $Y$ leads to a larger price deviation, since more liquidity has been taken out of the book. The function $\psi_i$ is left-continuous and verifies \begin{equation} 
\label{def_psi}
 \psi_i(y) := \sup \{ a \geq 0: F_i(a) < y \},
\quad\forall y\geq 0, \end{equation}with $\psi_i(0):=0$. It follows from the monotonicity and the
left-continuity of $F_i$ that $\psi_i$ is non-decreasing on $\RR_+$ and that $F_i(\psi_i(y))=y$, for all $y\in \RR_+$.
\begin{figure}[H]
    \centering
    \begin{subfigure}[b]{0.44\textwidth}
        \centering
        \includegraphics[width=\textwidth]{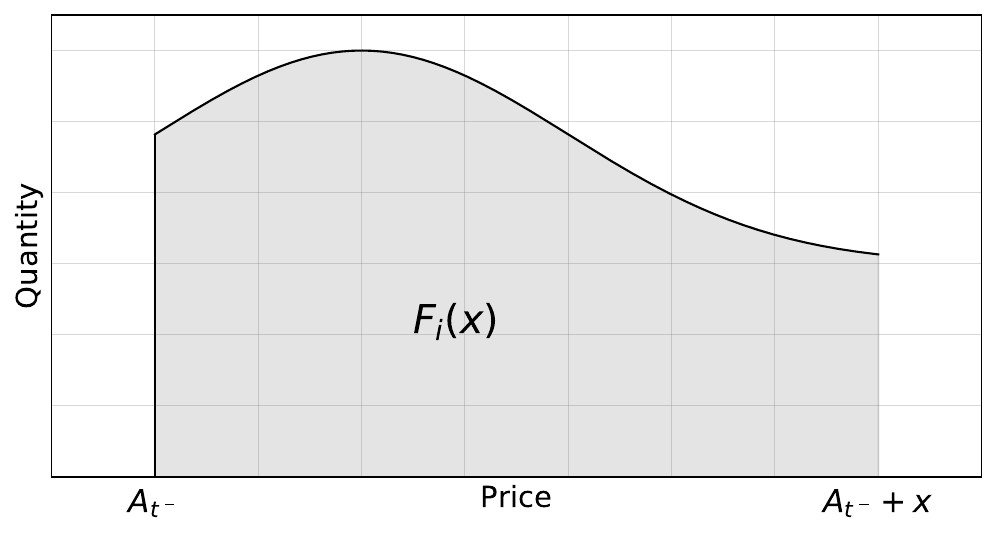}
    \end{subfigure}
    $\quad$
    \begin{subfigure}[b]{0.44\textwidth}
        \centering
        \includegraphics[width=\textwidth]{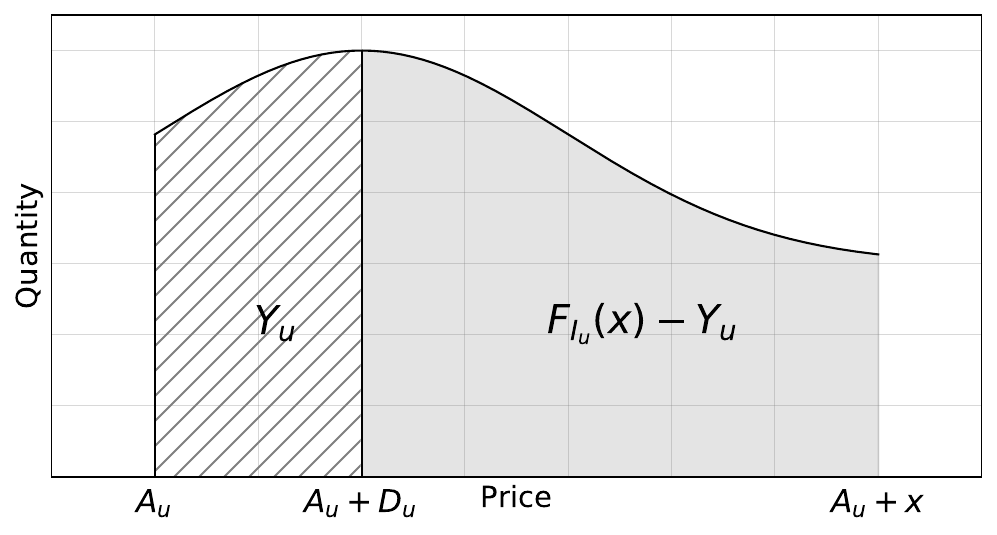}
    \end{subfigure}
       \caption{Representation of the limit order book at times $t^-$ and $u$. The left panel illustrates the initial state with $Y_{t^-} = y = 0$, while the right panel shows the LOB after transactions occurring between $t^-$ and $u$. The shaded area $Y_u$ indicates the volume of consumed shares.}
       \label{fig:lob}
\end{figure}
  The volume consumption affecting the LOB also accounts for the influence of other market participants, modeled through the Brownian motion $W$ and the jump term $M$, with compensator $\lambda_u \nu(\ud z)\mathrm{d}u$. The jumps here capture the impact of block trades and large orders placed by external market participants.
  \begin{rque}
The framework naturally extends to a multi-asset setting (see, e.g., \textcite{pemy2023}). 
For each asset $k \in \{1,\dots,d\}$, let $Y^k$ denote the corresponding volume effect. 
Cross-impact effects, as discussed in \textcite{hey2025}, might be incorporated by allowing 
the deviation process of asset $k$ to depend on the liquidity states of all assets, for instance through 
$D^k_u = \Psi^k_{I_u}(Y^1_u,\dots,Y^d_u)$,
where $\Psi^k$ captures both self- and cross-liquidity interactions. A formulation of the corresponding multidimensional model is provided in Appendix \ref{sec:multi_asset_extension}.
\end{rque}
We list the following assumptions to support our study of the control problem later.
\begin{assump} Assume that:
\begin{enumerate}[label=(A\arabic*), ref=(A\arabic*)]
\item\label{carnetinfini} There exists $b>0$, $\beta>0$ and $a> 0$ such that
\begin{equation*}
F_i(x)\geq b x^\beta,\quad\forall (x,i)\in [a,+\infty[\times\I_m.
\end{equation*}
\item\label{q_moments_monotonicity} The measure $\nu(\ud z)$ satisfies $$\int_{\mathbb{R}} (1 + z^2)\nu(\ud z) < +\infty.$$
\item \label{linear_growth}There exists a constant $C > 0$ such that, for all $y \geq 0$,
$$|h(y)| + |\sigma(y)| + \int_{\mathbb{R}} |q(y, z)| \nu(\ud z) \leq C(1 + |y|).$$
\item\label{lipschitz} There exists a constant $L > 0$ such that, for all $y, y' \geq 0$,
$$|h(y) - h(y')| + |\sigma(y) - \sigma(y')| + \int_{\mathbb{R}} |q(y, z) - q(y', z)| \nu(\ud z) \leq L |y - y'|.$$
\end{enumerate}
\end{assump} 
Assumption \ref{carnetinfini} prevents the placement of buy orders with finite size at infinitely high prices. This condition imposes a power–law lower bound on the cumulative LOB shape
$F_i$. This condition is well supported empirically: numerous studies report that
market impact grows as a concave power law of traded volume, see \textcite{BOUCHAUD200957} and \textcite{gabaix2016}. It provides an upper bound for the resilience function $\psi_i$, which is crucial for establishing the joint continuity of the value function introduced in the following sections. It also plays a key role in proving the Comparison Principle (see Theorem \ref{comparaison}) for the viscosity solutions of the value function. Assumptions \ref{q_moments_monotonicity}, \ref{linear_growth}, and \ref{lipschitz} guarantee the existence and uniqueness of a strong solution to the SDE \eqref{volume_effect_process} via Proposition \ref{existence_uniqueness_SDE}, as well as a comparison property for the \textit{volume effect} process through Lemma \ref{comparison_prop_Y}. These assumptions hold for the rest of this study.
\subsection{Problem Formulation}
\label{problem_formulation}
Our goal is to address an optimal purchase problem for a single asset over a finite time horizon $T$. As this study is conducted in a continuous trading context rather than high-frequency trading, we consider singular controls as admissible strategies, where jumps correspond to block trades.

A financial agent wants to buy $\midbar{X}$ shares of an illiquid asset over the time interval $[0,T]$. The agent's holdings $X$ form an $\mathcal{F}$-progressively measurable process.
\begin{defi}[Admissible strategies] \label{admissible_strategies}Define $\Ac_t(x)$ as the set of admissible controls for an agent
with a position $x\in \RR_+$ at time $t\in[0,T]$. An admissible purchase strategy consists of a non-decreasing $\cF$-adapted right-continuous process $X = (X_u)_{t\leq u\leq T}$ such that $X_{t^-}=x$ and $X_T=\midbar{X}$. We denote by $\Delta X_u=X_u-X_{u^-}$ the jump at time
$u$. We also represent $X^c$ as the continuous component of $X$.
\end{defi}
 We define the controlled dynamics of the \textit{volume effect} process $Y^{t,y,X}$ as
\begin{equation}\label{controlled_volume_effect_process}
\left\{\begin{array}{ll}
  \mathrm{d}Y^{t,y,X}_u=\mathrm{d}X_u-h(Y^{t,y,X}_{u^-})\mathrm{d}u+\sigma(Y^{t,y,X}_{u^-})\mathrm{d}W_u + \int_{\R} 
  q(Y^{t,y,X}_{u^{-}}, z) M (\mathrm{d}u, \mathrm{d}z),\\Y^{t,y,X}_{t^-}=y, 
\end{array} \right.
\end{equation}
with $u\in[t,T]$ and $y\geq 0$. Using classical theory (see \textcite{protter2005stochastic}), we can directly conclude that the given system of stochastic differential equations has a unique strong solution. 
\begin{prop}
\label{existence_uniqueness_SDE}
For any $\cF_t$-measurable random variable $\xi$ valued in $\RR_+$ such that $\EE(|\xi|^p)<+\infty$, for some $p>1$, the SDE \eqref{controlled_volume_effect_process} admits, for all $t\in [0,T]$, a unique strong solution $Y^{t,\xi,X}$, with $Y_{t}^{t,\xi,X}=\xi$. Moreover, 
\begin{equation*}
   \E\big[\underset{0\leq u\leq T}{\sup}|Y^{t,\xi,X}_u\mid^p\big]\leq C_T\big(1+\EE(|\xi|^p)\big),
\end{equation*}
where the constant $C_T\in \RR_+$ only depends on $T$, $\midbar X$ and the Lipschitz coefficients of $h$, $\sigma$ and $q$.
\end{prop}
We state a comparison result for the sample paths of $Y$, which will be used in proving the monotonicity of $v$ and later in the proof of continuity.
\begin{lemma}[Comparison property]
    \label{comparison_prop_Y}
     Let $t\in[0,T]$, $x\in[0,\midbar X]$ and $X\in\cA_t(x)$. For any two controlled solutions $Y^{t,Y^1_0,X}$ and $Y^{t,Y^2_0,X}$ of the SDE \eqref{controlled_volume_effect_process} satisfying $Y^1_0\leq Y^2_0$, it holds that $$\PP\big(Y^{t,Y^1_0,X}_u\leq Y^{t,Y^2_0,X}_u,~~\forall u\in[t,T]\big)=1.$$
\end{lemma}
\begin{proof}
    Refer to Theorem $2.3$ in \textcite{dawson_li}.
\end{proof}
Note that the \textit{volume effect} process $Y^{t,y,X}$ incorporates the impact of the investor's trading $X$ within its dynamics. Equation \eqref{controlled_volume_effect_process} specifies how the order book is affected by the transactions of the trader. If a large transaction happens at time $t$ with the investor buying $x$ shares, then $Y^{t,y,X}_t= \check{Y}^{t,y,X}_{t} + x
$. The ask price jumps from
$A_{t^-}+\psi_{I_{t^-}}(Y^{t,y,X}_{t^-})$ to $A_t+\psi_{I_{t}}(\check{Y}^{t,y,X}_{t} + x)$ after the transaction. Based on the limit order book structure described in Section \ref{sec:model}, the purchase cost for a large trade of size $y$, in excess of $A_t$, at time $t$ is expressed, for every regime $i\in \I_m$, as 
\begin{equation}
\label{definitionPhi}
\begin{split} \Phi_i(y)&:=\int_0^{\psi_i(y)}\xi
\mathrm{d}F_i(\xi)\\&=\int_0 ^y
\psi_i(\zeta)\, \mathrm{d}\zeta, \: \quad \forall y \geq 0.
    \end{split}
\end{equation}
In other words, the cost of placing a buy order of size $y$ at time $t$ at price $P_t$ is given by
\begin{equation*}
\pi^i_t(y)=\int_0^{y}\big[P_t+\psi_i(\zeta)\big] \ud \zeta=\underbrace{P_t y}_{\text {Cost at the current price }}+\underbrace{\Phi_i(y)}_{\text {Impact cost }} .
\end{equation*}
 We assume that the Markov chain $I$ starts in state $i$ at time $t$. Since $X$ is nonnegative and bounded by $\midbar{X}$, with $X_{t^-}=x$, we express the cost of the purchase strategy $X$ over $[t,T]$, under suitable conditions on the price $A$, as described in \textcite{predoiu2011}. In this case, the cost is equal to
\begin{align*}
& \int_t^T (A_u+ \check{D}_{u^-})\mathrm{d}X^c_u+\sum_{t\leq u\leq T}A_u\Delta
X_u+\Phi_{I_u}(Y^{t,y,X}_u)-\Phi_{I_{u}}(\check{Y}^{t,y,X}_{u^-})\\
&= \int_t^T\psi_{I_{u}}(\check{Y}^{t,y,X}_{u^-})\mathrm{d}X^c_u+A_T \bar X - A_t x - \int_t^T X_u \mathrm{d}A_u+\sum_{t\leq u\leq
T}\Phi_{I_u}(Y^{t,y,X}_u)-\Phi_{I_{u}}(\check{Y}^{t,y,X}_{u^-}),
\end{align*}
where $\check{D}_u
:= \psi_{I_u}(\check{Y}^{t,y,X}_u)$ and $u\in[t,T]$. Consequently, assuming the common condition that the ask price $A$ is a martingale, minimizing slippage is equivalent to minimizing the total cost of purchases. This leads to the following optimization problem, where we aim to minimize the expected cost in excess of $A_t(\midbar X - x)$ over the set of admissible strategies, defined by the value function $v := (v_i)_{i\in\I_m}$, such that
\begin{equation}\label{value_function}
v_i(t,x,y):=\inf_{X\in\mathcal{A}_t(x)}\EE\Big[\int_t^T\psi_{I_u}(\check{Y}^{t,y,X}_{u^-})\mathrm{d}X^c_u+\sum_{t\leq
u\leq T} \Phi_{I_u}(Y^{t,y,X}_u)-\Phi_{I_{u}}(\check{Y}^{t,y,X}_{u^-})\Big],
\end{equation}
with $(t,x,y)\in\cS:=[0,T[\times[0,\midbar{X}[\times
\RR^*_+$ and $i\in \I_m$. We denote $\midbar \cS$ as the closure of $\cS$.\\The boundary conditions are defined by the immediate purchase of the remaining shares as
\begin{equation}\label{boundarycondition}
\begin{split}
  v_i(T,x,y)&=\Phi_i(y+\midbar{X}-x)-\Phi_i(y),\\
  v_i(t,\midbar{X},y)&=0.
  \end{split}
\end{equation}
Note that $v_i$ satisfies the following growth condition (see Proposition \ref{finite_val_fun}) 
\begin{equation} 
\label{growthcondition}
 0 \leq v_i(t,x,y)\leq
\Phi_i(y+\midbar{X}-x)-\Phi_i(y).\end{equation}

\subsection{Examples}
\label{examples}
We begin by presenting examples that inspired our study. For clarity, these illustrations focus on the single-regime case ($m = 1$), although we consider the general setting throughout the paper.

Let $(t,x,y)\in \midbar \cS$, and $X\in\mathcal{A}_t(x)$. Assume that the \textit{volume effect} is deterministic, meaning $\sigma\equiv q\equiv 0$. Further, suppose that process $Y^{t,y,X}$ has finite variation and satisfies $Y^{t,y,X}_{t^-}=y$. 
Under these conditions, the setup corresponds to Model $1$ in \textcite{general_shape} within a discrete-time framework and, more specifically, to the one detailed in \textcite{predoiu2011}. Applying Itô's formula, we obtain
\begin{equation*}
  \Phi(Y^{t,y,X}_T)-\Phi(y)=\int_t^T\psi(Y^{t,y,X}_u)\big(\ud X^c_u-h(Y^{t,y,X}_u)\ud u\big)+\sum_{t\leq
  u\leq T}\Phi(Y^{t,y,X}_u)-\Phi(Y^{t,y,X}_{u^-}).
\end{equation*}
Hence,
\begin{equation*}
v(t,x,y)=\inf_{X\in\mathcal{A}_t(x)}\expec\Big[\Phi(Y^{t,y,X}_T)+\int_t^Tg(h(Y^{t,y,X}_u))\ud u\Big]-\Phi(y),
\end{equation*}
where $$g(y)=y\psi(h^{-1}(y)).$$ If one further assumes that $g$ is a convex function, Jensen's inequality implies that
\begin{equation*}
\int_t^Tg(h(Y^{t,y,X}_u))\ud u\geq(T-t)g\Big(\frac{1}{T-t}\int_t^T h(Y_u^{t,y,X})\,\ud u\Big).
\end{equation*}
Knowing that $Y_T^{t,y,X}=y+\midbar{X}-x-\int_t^T h(Y_u^{t,y,X})\,\ud u$,
we get 
\begin{equation*}
v(t,x,y)\geq
\inf_{X\in\mathcal{A}_t(x)}\expec\Big[G_{t,x,y}(Y_T^{t,y,X})\Big]-\Phi(y),
\end{equation*}
where $G_{t,x,y}(e)=\Phi(e)+(T-t)g\Big(\frac{\midbar{X}-x+y-e}{T-t}\Big)$. We define a type $A$ strategy as a strategy that has an initial jump at time $t$, satisfying $\ud X_u=\ud X^c_u=h(Y_u)\ud u$ on $]t,T[$, with the agent unwinding the remaining position at time $T$. Assume that $e^*_{t,x,y}$ minimizes $G_{t,x,y}$, and
\begin{equation*}
  e^*_{t,x,y}\in\big\{Y_T^{t,y,X}:X\in\mathcal{A}_t(x)\text{ is of type $A$}\big\}.
\end{equation*}
 Then, 
\begin{equation}
\label{linear_vf_exercise}
  v(t,x,y)=G(e^*_{t,x,y})-\Phi(y).
\end{equation}
Theorem $4.2$ of \textcite{predoiu2011} establishes that the type $A$ strategy associated to $e^*_{t,x,y}$ is the unique solution to the optimization purchase problem when $g$ is strictly convex. Furthermore, the authors demonstrate that a solution to the optimal execution problem exists even if the convexity of $g$ no longer holds (see Theorem $4.5$). Under this alternative situation, the trader initially purchases a lump sum of shares, $X_0 = y$, at time $0$. Subsequently, shares are bought at a constant rate, $\ud X_t = h(y) \, \ud t$, over the interval $t \in [0, t_0[$, maintaining $Y_t = y$. At time $t_0$, another lump sum purchase is made, after which the trader continues to buy shares at a constant rate, $\ud X_t = h(Y_{t_0}) \, \ud t$, during $t \in [t_0, T[$, keeping $Y_t = Y_{t_0}$. Finally, the remaining shares are purchased at time $T$. For example, consider a block-shaped order book as described in \textcite{OBIZHAEVA20131},
with a density $q_0=1$ and a resilience rate $\rho = 1$, resulting in $F(x)=\psi(x)=h(x)=x$, for all $x\geq 0$. In this case, $\Phi(x)=\frac{x^2}2$ and
$g(x)=x^2$, for all $x\geq 0$. Additionally, if $X\in\mathcal{A}_t(x)$ is of type
$A$, the following holds 
\begin{equation*}
e^*_{t,x,y} = \frac{2(\midbar X-x+y)}{T-t+2}.
\end{equation*}
In the framework proposed by \textcite{predoiu2011}, the volume effect is assumed to be non-dynamic, effectively treating it as deterministic throughout the execution horizon. As a result, the authors implicitly assume that the initial state lies within the execution region when starting from $x = 0$ and $y=0$. However, this assumption does not necessarily hold when the volume effect is dynamic. In general, $e^*_{t,x,y}$ is attainable by a type $A$ strategy if and only if
\begin{equation*}
  \overline{X}-x-(1+(T-t))y\geq 0.
\end{equation*}
Furthermore, the value function in this region is equal to 
\begin{equation*}
v(t, x, y) = \frac{(\bar X - x + y)^2}{2+T-t}-\frac{y^2}{2}.\end{equation*}
The representation of the value function in the next section enables its characterization beyond this domain. Observe that this case falls within the framework of \textcite{alfonsi2016dynamic} if we assume that the strategic trader's proportion of permanent price impact is identical to that of other traders and equal to zero. Suppose further that the measure $m$ satisfies $m\equiv 0$, the jump size of $Y$ satisfies $q\equiv 1$, and that the process
$$\int_0^t \sigma^2(Y_u)\,\ud W_u$$
is $\cF$-martingale, where $Y$ is the solution of the uncontrolled SDE \eqref{volume_effect_process}. The argument in the proof of Theorem $2.1$ of \textcite{alfonsi2016dynamic} shows that there are no Price Manipulation Strategies (PMS) in this setting and that the control problem is equivalent to the execution problem of \textcite{OBIZHAEVA20131}. 
This example motivates our study of the general case, in which the volume effect evolves dynamically as a jump diffusion process. In this setting, we aim to characterize the properties of the value function within each region and along the free boundary.
 \begin{rque}
    The framework introduced by \textcite{alfonsi2016dynamic} was extended by \textcite{chevalier2024} to incorporate impulse controls. Although this extension enhances the modeling of price impact, the limit order book remains static, meaning that its shape does not respond to incoming orders, and exhibits a deterministic volume effect with fixed resilience. Additionally, the dynamics are solely driven by point processes representing order flow. As a result, while these models are well suited to high-frequency trading environments, they are less appropriate for longer time scales. This motivates the use of Brownian noise in \eqref{volume_effect_process} to capture stochastic resilience effects.
\end{rque}
\section{Analytical Properties of the Value Function}
\label{Characterization of the Value Function}
In this section, we will denote, for any $(i,t,x,y)\in \I_m\times\midbar \cS$, and any admissible strategy $X\in\mathcal{A}_t(x)$,
    \begin{equation*}
C_i(t,x,y,X):=\int_t^T \psi_{I_{u}}(\check{Y}_{u^-}^{t,y,X})\ud X^c_u+\sum_{t\leq u\leq
T;\Delta X_u>0}\Phi_{I_{u}}(Y_u^{t,y,X})-\Phi_{I_{u}}(\check{Y}_{u^-}^{t,y,X}).
    \end{equation*}
We start by stating standard results on the finiteness and monotonicity of the value function.
 \begin{prop}
 \label{finite_val_fun}
     The value function $v$ described in \eqref{value_function} is finite.
 \end{prop}
\begin{proof}
\label{proof_finitness}
Let
$(t,x,y)\in\midbar \cS$
be the state variable at time $t$ and $i\in \I_m$. If the investor buys $\midbar{X}-x$ immediately, the associated cost would be equal to
$\Phi_i(y+\midbar{X}-x)-\Phi_i(y)$. Therefore, 
$$0\leq v_i(t,x,y)\leq
\Phi_i(y+\midbar{X}-x)-\Phi_i(y)$$on $\midbar \cS$. This concludes the proof.
\end{proof}

\begin{prop}[Monotonicity]
\label{lemma_decreasing}
For any $(i,t,x,y)\in\I_m\times\midbar\cS$, the following results hold:
\begin{enumerate}
    \item $t_0\mapsto v_i(t_0,x,y)$ is non-decreasing on $[0,T]$.
    \item $x_0\mapsto v_i(t,x_0,y)$ is non-increasing on $[0,\midbar{X}]$.
    \item $y_0\mapsto v_i(t,x,y_0)$ is non-decreasing on $\RR_+$.
\end{enumerate}
\end{prop}
\begin{proof}
\textbf{1.} Let $i\in \I_m$, $x\in[0,\midbar{X}]$ and $y\geq 0$. Suppose that $0\leq t\leq t'\leq T$ and consider $X' = (X'_u)_{t'\leq u\leq T}\in\mathcal{A}_{t'}(x)$ such that $X'_{{t'}^-}=x$ and $X'_T=\midbar{X}$. Using the strong Markov property, we get that
\begin{equation*}
\resizebox{\textwidth}{!}{$
\begin{aligned}\EE \big[C_i(t',x,y,X')|\cF_{t'}\big] &= \EE \Big[\int_{t}^{T-t'+t} \psi_{I_{u}}(\check{Y}_{u^-}^{t,y,{\widetilde X'}})\ud {\widetilde X}^c_{u}+\sum_{\substack{t\leq u\leq
T-t'+t\\\Delta {\widetilde X}_{u}>0}}\Phi_{I_{u}}(Y_{u}^{t,y,{\widetilde X'}})-\Phi_{I_{u}}(\check{Y}_{u^-}^{t,y,{\widetilde X'}})|\cF_{t}\Big],
\end{aligned}
$}
\end{equation*}
with $\widetilde X' = (X'_{u-t+t'})_{t\leq u\leq T-t'+t}$. Define the strategy $X\in\mathcal{A}_t(x)$ as 
\begin{equation*}
\begin{split}
        X_u = \left\{\begin{array}{ll}
\widetilde X_u, &\quad \text { if }u\in[t,T-t'+t], \\ 
0,&\quad\text { else}. 
\end{array}\right.
\end{split}
    \end{equation*}
Note that $\EE \big[C_i(t',x,y,X')|\cF_{t'}\big] = \EE \big[C_i(t,x,y,X)|\cF_{t}\big]$. Since $X\in\cA_t(x)$, then $$v_i(t,x,y)\leq \EE \big[C_i(t',x,y,X')|\cF_{t'}\big].$$   
As $X'$ is arbitrary, it follows that $v_i(t,x,y)\leq v_i(t',x,y)$, completing the proof.\\
\textbf{2.} Let $i\in \I_m$, $t\in[0,T]$ and $y\geq 0$. Suppose that $0\leq x\leq x'\leq \midbar{X}$. Let
$\epsilon>0$ be arbitrary and assume that $X\in\mathcal{A}_t(x)$ is an $\varepsilon$-optimal strategy that satisfies
    \begin{equation*}
             v_i(t,x,y)+\epsilon\geq\mathbb{E}[C_i(t,x,y,X)].
        \end{equation*}
Let $\tau$ be an $\cF$-stopping time, such that
    \begin{equation*}
        \tau:=\inf\big\{u\geq t: X_u-x\geq \midbar{X}-x'\big\}.
    \end{equation*}
We define the strategy $X'\in\mathcal{A}_t(x')$ as 
\begin{equation*}
\begin{split}
        X'_u = \left\{\begin{array}{ll}
X_u+x'-x, &\quad \text { if }u<\tau, \\ 
\midbar{X},&\quad\text { else}. 
\end{array}\right.
\end{split}
    \end{equation*}
     Using the comparison property established in Lemma \ref{comparison_prop_Y}, we have that
    \begin{equation*}
        C_i(t,x',y,X')\leq C_i(t,x,y,X).
    \end{equation*}
Taking expectations on both sides, we deduce that
    \begin{equation*}
v_i(t,x',y)\leq \mathbb{E}[C_i(t,x',y,X')]\leq \mathbb{E}[C_i(t,x,y,X)]\leq v_i(t,x,y)+\epsilon.  \end{equation*}
Since $\epsilon$ is arbitrary, the conclusion follows.\\
\textbf{3.} This result follows directly from the comparison property established for the process $Y$ and the fact that $\Phi_i$ and $\psi_i$ are non-decreasing on $\RR_+$ for all $i\in\I_m$. 
\end{proof}
To prove the continuity of the value function, we need the following lemmata. The first result concerns the continuity of the flow of $Y$. The second is related to a continuity modulus. 
\begin{lemma}
\label{cont-flow}
For any non-decreasing, right-continuous, adapted process $X$ and
any initial time $t\in[0,T]$, the mapping $y \mapsto Y_{.}^{t,y,X}$, representing the flow of the solution to SDE \eqref{controlled_volume_effect_process}, is continuous on $\RR_+$.
\end{lemma}
\begin{proof}
 See Theorems $V.37$ and $V.38$ in \textcite{protter2005stochastic}.
\end{proof}
\begin{lemma}
\label{unifcont_phi}
Let $0\leq t\leq s\leq T$, $x\in[0,\midbar X] $, $y\geq 0$ and
$X\in\mathcal{A}_t(x)$. For any nonnegative random variable $\xi$ with finite moments, it holds that
\begin{equation*}
\E\Big[\Phi_{I_s}\big(Y^{t,y,X}_s+\xi\big)-\Phi_{I_s}\big(Y^{t,y,X}_s\big)\Big]\leq \rho_y(\xi),
\end{equation*}
where the mapping $(y,\zeta)\mapsto\rho_y(\zeta)$ is continuous in $y$ and non-decreasing in $\zeta$ on $\RR_+$, satisfying $\lim_{\zeta \to 0} \rho_y(\zeta) = 0$. Additionally, $\rho_y(\xi)$ converges in probability to $0$ as $\xi$ tends to $0$ in probability, if its moments are uniformly bounded. 
\end{lemma}
\begin{proof}
    Refer to Appendix \ref{appendix_caracterization}.
\end{proof}
We conclude this section by presenting a key result on the continuity of the value function.
\begin{theo}
\label{continuous_v}
The value function $v$ defined in \eqref{value_function}
is a continuous function on $\midbar{\cS}$. \end{theo}
\begin{proof} We will study separately the continuity of the value function in $t$, $x$ and $y$, for each $i\in \I_m$.

\paragraph{1.} We first prove the continuity of $v_i$ in $x$, uniformly with respect to $t$. We fix $t\in[0,T] $, $y\geq 0$, $0\leq x'<x\leq \midbar{X}$ and $\varepsilon
>0$. There exists $X\in\mathcal{A}_t(x)$ which satisfies
        \begin{equation*}
             v_i(t,x,y)+\varepsilon\geq\mathbb{E}\big[C_i(t,x,y,X)\big].
        \end{equation*}
By adding a jump of size $x-x'$ at time $T$ to $X$ we obtain $\widehat X\in\mathcal{A}(t,x')$. We have
that
        \begin{equation*}
            C_i(t,x',y,\widehat X)=C_i(t,x,y,X)+\Phi_{I_T}(Y_T^{t,y,X}+x-x')-\Phi_{I_T}(Y_T^{t,y,X}).
        \end{equation*}
    Taking expectations on both sides of this equation we deduce that
        \begin{equation*}
\mathbb{E}\big[C_i(t,x',y,\widehat X)\big]\leq
v_i(t,x,y)+\varepsilon+\mathbb{E}\Big[\Phi_{I_T}(Y_T^{t,y,X}+x-x')-\Phi_{I_T}(Y_T^{t,y,X})\Big].
        \end{equation*}
On the other hand since, by Proposition \ref{lemma_decreasing}, the value function is decreasing in $x$, we conclude that
        \begin{equation*}
v_i(t,x,y)\leq v_i(t,x',y)\leq
v_i(t,x,y)+\varepsilon+\mathbb{E}\Big[\Phi_{I_T}(Y_T^{t,y,X}+x-x')-\Phi_{I_T}(Y_T^{t,y,X})\Big].
        \end{equation*}
 We deduce from Lemma \ref{unifcont_phi} that
 \begin{equation*}
 0\leq   v_i(t,x',y)-  v_i(t,x,y)\leq   \varepsilon+\rho_{y}(x-x^\prime).
 \end{equation*}
Therefore, we have obtained the
continuity of the value function $v$ in $x$ uniformly in $t$.

\paragraph{2.} Next, we prove the continuity of $v_i$ in $y$, uniformly with respect to $t$ and $x$.
Fix $y\geq 0$, $y^\prime\geq0$, $\epsilon>0$, $t\in[0,T]$ and $x\in[0,\midbar{X}]$. There exists $X\in\mathcal{A}_t(x)$
which satisfies
            \begin{equation*}
                v_i(t,x,y)+\epsilon\geq\mathbb{E}\big[C_i(t,x,y,X)\big].
            \end{equation*}
Our goal here will be to construct an admissible strategy $\widetilde X$ that bounds \begin{equation*}
\begin{split}
    v_i(t,x,y^{\prime}) - v_i(t,x,y) &\leq  \eps + \E \big[C_i(t,x,y^{\prime},\widetilde X) -C_i(t,x,y,X)\big].
\end{split}\end{equation*} 
To do so, we compare the sample paths of $Y^{t,y,X}$ and $Y^{t,y^\prime,\widetilde X}$, and aim to bound their difference using terms that depend only on the initial parameters $y$, $y^{\prime}$ and their difference $|y^{\prime}-y|$, as illustrated in Figure \ref{fig:controlled_volume_effect}. This, combined with the continuity modulus presented in Lemma \ref{unifcont_phi}, should allow us to achieve the uniform continuity in $y$. Consider $\widehat X\in\mathcal{A}_t(x)$ such that  
\begin{equation*}
\ud \widehat X_u^c=0,~~\textrm{and}~~ \Delta\widehat X_u=\Delta X_u,\quad
\forall u\in[0, T[ .
\end{equation*} 
Let $\tau$ be the following $\cF$-stopping time,
\begin{equation*}
\tau:=\inf\big\{u\geq t:\ Y^{t,y^\prime,\widehat X}_u\leq Y^{t,y,X}_u\big\}\wedge T,
\end{equation*}
$\xi$ define the random variable
\begin{equation*}
\xi:= \int_t^\tau \ud  X^c_u - \big(Y^{t,y,X}_{\tau}-Y^{t,y^\prime,\widehat X}_\tau\big),
\end{equation*}
and $\theta$ be an $\cF$-stopping time, such that
$$\theta:=\inf\big\{u\geq \tau:\ \widehat X_\tau+X_u-X_\tau\geq \midbar X\big\}.$$
Observe that when $y^\prime\leq y$, $\tau=t$ and $\xi = y^\prime - y\leq 0$. We construct a strategy $\widetilde X$ belonging to $\mathcal{A}_t(x)$ satisfying
\begin{equation*}
\begin{split}
\ud \widetilde X_u:=\left\lbrace\begin{array}{ll}
\ud \widehat X_u, & \quad\textrm{if }t\leq u<\tau,\\
\big(\Delta X_u+Y^{t,y,X}_u-Y^{t,y^\prime,\widehat X}_u\big)\wedge\big(\midbar X-\widehat X_{\tau^-}\big), & \quad\textrm{if } u=\tau,\\
\ud Z_u, & \quad\textrm{if }\tau< u\leq T.
\end{array}\right.
\end{split}
\end{equation*}
where $Z$ is a non-decreasing $\cF$-adapted process with $Z_\tau=0$. If $\xi\geq 0$,
$$
\ud Z_u:=\left\lbrace\begin{array}{ll}
\ud X_u, & \quad\textrm{if }\tau< u< T,\\
\Delta X_u+\xi, & \quad\textrm{if }u=T.
\end{array}\right.$$
If $\xi< 0$,
$$
\ud Z_u=\left\lbrace\begin{array}{ll}
\ud X_u, & \quad\textrm{if }\tau< u< \theta,\\
\midbar X-\big(\widehat X_{\tau}+X_{\theta^-}-X_{\tau}\big),\, & \quad\textrm{if }u=\theta\textrm{ and }\tau<\theta\leq T,\\
0, & \quad\textrm{else }.
\end{array}\right.$$ 
The goal of this construction is to align the paths of $u\mapsto Y^{t,y,X}_u$ and $u\mapsto Y^{t,y^\prime,\widehat X}_u$ at a chosen time $\tau$, allowing for their comparison and the derivation of bounds for inequality \eqref{ineq_cont_y}.

\begin{figure}[H]
        \centering
        \includegraphics[width=0.6\textwidth]{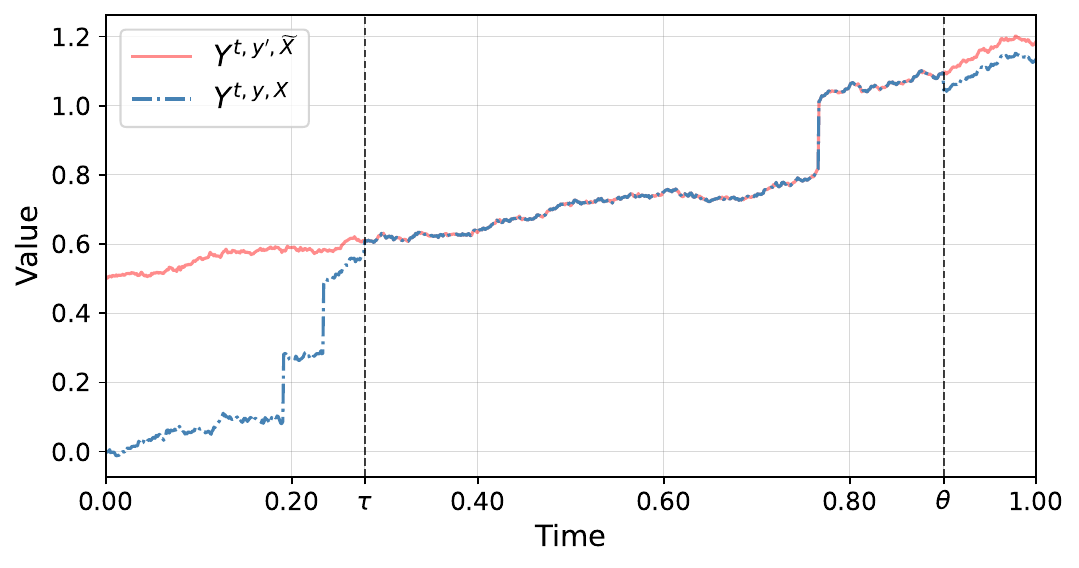}
       \caption{Illustration of $Y^{t,y,X}$ and $Y^{t,y^\prime,\widetilde X}$ sample paths in the scenario where $y<y^{\prime}$, $\tau<\theta<T$ and $\Delta X_u+Y^{t,y,X}_u-Y^{t,y^\prime,\widehat X}_u<\midbar X-\widehat X_{\tau^-}$.}
       \label{fig:controlled_volume_effect}
\end{figure}

Therefore, we obtain the following inequality
\begin{equation}
\label{ineq_cont_y}
\begin{split}
    v_i(t,x,y^{\prime}) - v_i(t,x,y) \leq \eps+R_1+R_2+R_3+R_4+R_5+R_6,
\end{split}\end{equation} 
 where 
\begin{equation*}
\resizebox{\textwidth}{!}{$
\begin{aligned}R_1 &= \E \Big[\int_t^T \psi_{I_{u}}(\check Y_{u^-}^{t,y^{\prime},\widetilde X}) \ud \widetilde X_u^c -\int_t^T
\psi_{I_{u}}( \check
Y_{u^-}^{t,y,X}) \ud X_u^c\Big] \leq 0,\\
R_2 &= \E \Big[\sum_{t\leq u<
\tau;\Delta X_u>0}\Phi_{I_{u}}(Y_u^{t,y^{\prime},\widetilde X})-\Phi_{I_{u}}(\check{Y}_{u^-}^{t,y^{\prime},\widetilde X})\Big]-\E \Big[\sum_{t\leq u<
\tau;\Delta X_u>0}\Phi_{I_{u}}(Y_u^{t,y,X})-\Phi_{I_{u}}(\check{Y}_{u^-}^{t,y,X})\Big],\\
R_3 & = \E \Big[\Phi_{I_{\tau}}(Y_\tau^{t,y^{\prime},\widetilde X})-\Phi_{I_{\tau}}(\check{Y}_{\tau^-}^{t,y^{\prime},\widetilde X})\Big]-\E\Big[\Phi_{I_{\tau}}(Y_\tau^{t,y,X})-\Phi_{I_{\tau}}(\check{Y}_{\tau^-}^{t,y,X})\Big],\\
 R_4 & = \E \Big[\1_{\{\xi\geq 0\}} (\Phi_{I_{T}}(Y_{T^-}^{t,y^{\prime},\widetilde X}+\xi)-\Phi_{I_{T}}(Y_{T^-}^{t,y,X}))\Big],\\
 R_5 & \leq \E \Big[\1_{\{\xi< 0\}} (\Phi_{I_{\theta}}(\check{Y}_{\theta^-}^{t,y^{\prime},\widetilde X}+\Delta \widetilde X_\theta)-\Phi_{I_{\theta}}(\check{Y}_{\theta^-}^{t,y,X}+\Delta X_\theta))\Big]\leq 0.
\end{aligned}
$}
\end{equation*}

\vspace{0.2cm}

We begin with finding an upper bound for $R_2$. Denote by $u_k\in[t,\tau[$ the jump times of $X$. Note that $y\mapsto \Phi_{I_{u}}(y)$ is non-decreasing on $\RR$. An application of the comparison property of $Y$ (see Lemma \ref{comparison_prop_Y}) yields $\check{Y}_{u^-}^{t,y,X}\leq \check{Y}_{u^-}^{t,y^{\prime},\widetilde X}\leq \check{Y}_{u^-}^{t,y^{\prime}, X}$, for all $u<\tau$. Therefore,

\vspace{0.1cm}

\begin{equation}
\label{ineq_R2}
\begin{split}
R_2 &\leq \E \Big[\sum_{t\leq u_k<
\tau}\Phi_{I_{u_k}}(Y_{u_k}^{t,y^{\prime}, X})-\Phi_{I_{u_k}}(Y_{u_k}^{t,y, X})\Big]\\&\leq\E \Big[\sum_{t\leq u_k<
\tau}\Phi_{I_{u_k}}(Y_{u_k}^{t,y, X}+\gamma_{u_k})-\Phi_{I_{u_k}}(Y_{u_k}^{t,y, X})\Big],
\end{split}
\end{equation}
with $\gamma_u:=\check{Y}_{u}^{t,y^{\prime}, X}-\check{Y}_{u}^{t,y, X}\geq 0$, for all $u\in[ t, \tau[$. It follows from the dynamics of $Y$ that
\begin{equation*}
    \resizebox{\textwidth}{!}{$
\begin{aligned}
        \mathrm{d}(\check{Y}_u^{t,y', X} &- \check{Y}_u^{t,y, X})^c 
        \\= c &\big(\check{Y}_{u^-}^{t,y', X} - \check{Y}_{u^-}^{t,y, X}\big)^{c - 1} 
        \left[ 
            -\big(h(\check{Y}_{u^-}^{t,y', X}) - h(\check{Y}_{u^-}^{t,y, X})\big) \, \mathrm{d}u 
            + \big(\sigma(\check{Y}_{u^-}^{t,y', X}) - \sigma(\check{Y}_{u^-}^{t,y, X})\big) \, \mathrm{d}W_u 
        \right] \\
        &\quad+ \frac{c(c - 1)}{2} (\check{Y}_{u^-}^{t,y', X} - \check{Y}_{u^-}^{t,y, X})^{c - 2} 
        \left| \sigma(\check{Y}_{u^-}^{t,y', X}) - \sigma(\check{Y}_{u^-}^{t,y, X}) \right|^2 \, \mathrm{d}u \\
        &\quad + \int_{\mathbb{R}} 
        \left[
            \big(\check{Y}_{u^-}^{t,y', X} - \check{Y}_{u^-}^{t,y, X} 
            + q(\check{Y}_{u^-}^{t,y', X}, z) - q(\check{Y}_{u^-}^{t,y, X}, z)\big)^c 
            - \big(\check{Y}_{u^-}^{t,y', X} - \check{Y}_{u^-}^{t,y, X}\big)^c
        \right] 
        M(\mathrm{d}u, \mathrm{d}z) \\
        &\quad + \int_{\mathbb{R}} 
        c \big(\check{Y}_{u^-}^{t,y', X} - \check{Y}_{u^-}^{t,y, X}\big)^{c - 1} 
        \big(q(\check{Y}_{u^-}^{t,y', X}, z) - q(\check{Y}_{u^-}^{t,y, X}, z)\big) 
        \lambda_u \nu(\mathrm{d}z) \, \mathrm{d}u,
    \end{aligned}
$}
\end{equation*}
for all $u\in[t, T]$ and $c>0$. Using Assumption \ref{lipschitz} that for all jump times $u_k\in[t,\tau[$, 
\begin{equation*}
\begin{split}
\E[|\gamma_{u_k}|^c] \leq |y^{\prime}-y|+L'\E\Big[\int_0^{u_k} |{Y}_{u^-}^{t,y, X}-{Y}_{u^-}^{t,y^{\prime}, X}|^c\, \ud u\Big].
\end{split}
\end{equation*}
with $L'>0$. Thanks to the Grönwall inequality, we can then assert that
\begin{equation*}
    \E[|\gamma_{u_k}|^c] \leq |y^{\prime}-y| e^{L T}.
\end{equation*}
Applying the inequality $(1 + x)^\gamma - 1 \leq \gamma x + C_\gamma x^\gamma$ with $\gamma = \frac{\beta+1}{\beta} > 1$ and $C_\gamma := \sup_{x > 0} \frac{(1 + x)^\gamma - 1 - \gamma x}{x^\gamma}\geq 0$ for $x > -1$ (see \textcite[Chapter 3]{HardyInequalities}), we obtain
\begin{equation}
\label{eq:CBI-bound}
\rho_y(\xi)=\sum_{i=1}^m C^y_i \, \mathbb{E}\Big[\Big(\big(1 + \frac{\xi}{F_i(a)}\big)^{\gamma} - 1\Big)^2\Big]^{\frac{1}{2}}
\leq \sum_{i=1}^m C^y_i \Big( \frac{2\gamma^2 \, \mathbb{E}[\xi^2]}{F_i(a)^2}
+ \frac{2C_\gamma^2 \, \mathbb{E}[|\xi|^{2\gamma}]}{F_i(a)^{2\gamma}} \Big)^{\frac{1}{2}}.
\end{equation}
Lemma \ref{unifcont_phi} combined with \eqref{ineq_R2} and iterated expectations yields
\begin{equation*}
\begin{split}
R_2&\leq \E \Big[\sum_{t\leq u_k<
\tau}\rho_y(\gamma_{u_k})\Big]\leq T \rho'_y(|y^{\prime}-y|),
\end{split}
\end{equation*}
where $\rho'_y:\xi\mapsto  \sum_{i=1}^m C^y_i \Big( \frac{2\gamma^2 \, \mathbb{E}[\xi^2]}{F_i(a)^2}
+ \frac{2C_\gamma^2 \, \mathbb{E}[|\xi|^{2\gamma}]}{F_i(a)^{2\gamma}} \Big)^{\frac{1}{2}}$. Similarly, we have that $$R_3\leq \rho'_y(|y^{\prime}-y|)$$ since $y\mapsto \Phi_{I_{u}}(y)$ is non-decreasing on $\RR$, and $\check{Y}_{\tau^-}^{t,y,X}\leq \check{Y}_{\tau^-}^{t,y^{\prime},\widetilde X}\leq \check{Y}_{\tau^-}^{t,y^{\prime}, X}$, for $u<\tau$. The second inequality follows directly from Lemma \ref{comparison_prop_Y}. We turn to finding an upper bound for $R_4$. Notice that we have $R_5=0$ when $y^\prime\leq y$. If $y^\prime> y$, then $\E[\xi]\leq y^\prime-y$. In this case, 
\begin{equation*}
R_4 \leq\E \Big[\Phi_{I_{T}}(Y_{T}^{t,y, X}+\gamma_{T}+\xi)-\Phi_{I_{T}}(Y_{T}^{t,y, X})\Big].
 \end{equation*}
 Observe that $\xi$ satisfies the conditions of Lemma \ref{unifcont_phi} due to the strong solution result in Proposition \ref{existence_uniqueness_SDE}, and that the admissible controls have bounded variations. Therefore,
 \begin{equation*}\begin{split}
R_4 & \leq\E \Big[\rho_{y}(\gamma_{\tau}+\gamma_{T})\Big]\leq \rho'_{y}(2|y^{\prime}-y|).\end{split}
 \end{equation*}
We conclude this proof by recalling that $\zeta\to\rho_y(\zeta)$ is uniformly continuous on $\RR_+$, for all $y\geq 0$.
\paragraph{3.} Finally, we demonstrate the continuity of $v$ in $t$. Let $x\in [0,\midbar{X}]$,
$y\geq 0$ and $0\leq u<t$. Now, consider a strategy $X\in\mathcal{A}_u(x)$ such that $$v_i(u,x,y)+\eps \geq C_i(u,x,y,X).$$
In the following, we will construct a strategy $\widehat X\in\mathcal{A}_s(x)$ that consists of not doing anything on $[s,t[$ and then following the strategy $X$ similarly to the previous proof. We introduce the $\cF$-stopping time $\tau$, defined by
$$\tau:=\inf\big\{u\geq t:\ \check{Y}^{u,y,\widehat X}_{u^-}\leq y\big\}\wedge T.$$
Let $\widehat X$ define an admissible strategy in $\mathcal{A}_t(x)$, such that
$$
\ud \widehat X_u =\left\lbrace\begin{array}{lll}
\big(\Delta X_\tau+y-\check{Y}^{u,y,\widehat X}_{\tau^-}\big)\wedge \big(\midbar X-x\big), & \textrm{ if } u=\tau,\\
\ud  X_u, & \textrm{ if } \tau< u<\theta,\\
\midbar X-\widehat X_{u^-}, & \textrm{ if } u=\theta,\\
0, & \textrm{ if } \theta< u\leq T,
\end{array}\right.
$$
where 
$$\theta:=\inf\big\{u\geq \tau:\ \widehat X_{u^-}+\Delta X_u\geq \midbar X\big\}\wedge T.$$
Considering an $\eps$-optimal strategy $X\in\mathcal{A}_t(x)$, we can show that
\begin{equation*}
\begin{split}
v_i(u,x,y)-v_i(t,x,y) & \leq C_i(u,x,y,\widehat X)-C_i(t,x,y,X)+\eps\\
 & \leq \E\Big[C_i(t,x,{Y}^{u,y,\widehat X}_{t^-},\widehat X)-C_i(t,x,y,X)\Big]+\eps\\
 & \leq f(\eps),
 \end{split}
 \end{equation*}
 with $\lim_{\varepsilon\rightarrow 0} f(\varepsilon)=0$.
 
 We can combine the uniform continuity of $v_i$ in $x$ with respect to $t$, in $y$ uniformly with respect to $t$ and $x$, and in $t$ to get joint continuity on $\midbar \cS$.
\end{proof}
\begin{prop}
\label{differentiable_vf}
    The value function $v$ defined in \eqref{value_function}
is differentiable almost everywhere on $\midbar{\cS}$ and has finite one-derivatives everywhere.
\end{prop}
\begin{proof}
    The value function $v_i$ is continuous and monotonic on $\midbar\cS$, and thus, by Lebesgue’s Theorem, it is differentiable almost everywhere for all $i\in\I_m$. Therefore, there exists a set of Lebesgue measure zero where $v_i$ may not be differentiable. We denote the right and left partial derivatives of $v_i$ with respect to $x$ and $y$ as $\partial^{\pm}_x v_i$ and $\partial^{\pm}_y v_i$. Based on the proof of Theorem \ref{continuous_v}, we have that
 \begin{equation*}
 0\leq   \frac{v_i(t,x+h,y)-  v_i(t,x,y)}{h}\leq  \frac{\rho'_{y}(h)}{h}, \quad \forall (t,x,y,h) \in \midbar\cS\times \RR_+.
 \end{equation*}
 Moreover, it follows that $$\lim_{h\to 0^+}\frac{\rho'_{y}(h)}{h} = \sum_{i=1}^m C^y_i \frac{\sqrt{2}\gamma }{F_i(a)}<+\infty,$$
which implies that $\partial_x^+ v_i$ is finite on $\midbar S$. Similarly, $\partial_x^- v_i$ is finite. Hence, $v_i$ is differentiable in $x$ on $\midbar\cS$. The same argument applies to $t$ and $y$, establishing differentiability in all variables on 
$\midbar\cS$.
\end{proof}
\section{Viscosity Characterization of the Value Function}
\label{viscosity_section}
This section aims to provide a PDE characterization of the value functions $v$. 
\begin{lemma}[Dynamic programming principle]
\label{prog_dyn}
    For any stopping time $\tau$ in $[t,T]$, we have
\begin{equation*}\begin{split}
v_i(t,x,y)=&\\\inf_{X\in\mathcal{A}_t(x)}\expec\Big[&\int_t^\tau\psi_{I_{u}}(\check{Y}^{t,y,X}_{u^-})\mathrm{d}X^c_u+\sum_{t\leq
u\leq\tau} \big(\Phi_{I_{u}}(Y^{t,y,X}_u)-\Phi_{I_{u}}(\check{Y}^{t,y,X}_{u^-})\big)+v(\tau,I_{\tau},X_{\tau}, {Y}_{\tau}^{t,y,X})\Big],
\end{split}
\end{equation*}
where $(t,x,y)\in\cS$ and $i\in \I_m$.
\end{lemma}
\begin{proof}
    This is an adaptation of the results from \textcite{bouchard_touzi}, leveraging the continuity of the value functions $v_i$ established in Theorem \ref{continuous_v}.
\end{proof}
Building on the problem formulation \ref{problem_formulation}, we consider the following Hamilton-Jacobi-Bellman Quasi-Variational Inequalities for $v$, for all $i\in\I_m$,
\begin{equation}
\label{hjb}
\max \big(-\frac{\partial v_i}{\partial t}-\mathcal{L} v_i-\sum_{j\neq i}(v_j-v_i)Q_{ij},-\frac{\partial v_i}{\partial x}-\frac{\partial v_i}{\partial y}-\psi_i\big)=0,~\text{on}~\cS.
\end{equation}
The partial integro-differential operator $\mathcal{L}$ is given by 
\begin{equation*}
    \begin{split}
    \mathcal{L}\varphi&:= \frac{1} {2}\sigma^2(y)\frac{\partial^2 \varphi}{\partial y^2}-h(y)\frac{\partial \varphi}{\partial y} +
\lambda_t \int_{\R}
\left( \varphi(t,x, y+ q(y,z)) - \varphi \right) \nu(\ud z).\end{split}
\end{equation*}
The boundary/terminal conditions here are specified in \eqref{boundarycondition}. 
\begin{rque}
Vectors are assumed to be $m$-dimensional unless stated otherwise, and all operations are performed component-wise. Deviations from these notations must be explicitly specified.
\end{rque}
\subsection{Viscosity Solution Property}
 As the value function $v$ may lack smoothness, it is formulated as a discontinuous viscosity solution of a quasi-variational inequalities \eqref{hjb} with appropriate boundary conditions, following \textcite{users_guide}.

\begin{defi}[Viscosity solution]\label{viscosity_def}We define a viscosity solution of \eqref{hjb} as follows~:
   \begin{enumerate}
       \item $v$ is a continuous viscosity supersolution (resp. subsolution) of \eqref{hjb} on $\I_m\times \cS$ if it satisfies the growth conditions \eqref{growthcondition}, and if
\begin{equation}  \label{hjbthm}
  \max\Big(-\big(\frac{\partial \varphi}{\partial t}+\mathcal{L} \varphi+\sum_{j\neq i}(v_j-v_i)Q_{ij}\big)(t, x, y),\ -\  \big(\frac{\partial \varphi}{\partial x}+\frac{\partial \varphi}{\partial y}+\psi_i\big)(t, x, y)\Big)\geq 0\,(\text{resp.} \leq),
\end{equation}
for any $(i, t, x, y) \in \I_m\times\cS$ and any smooth test function $\varphi \in C^{1,2}(\cS)$ such that $(v_i - \varphi)$ attains a local minimum (resp. maximum) at $(t, x, y)$ over the set $[t, t+\delta[ \times [x, x+\delta[ \times B_\delta(y) \subset \mathcal{S}$ for some $\delta > 0$, with $(v_i - \varphi)(t, x, y) = 0$.
\item $v$ is a continuous viscosity solution on $\I_m\times \cS$ if it is both a viscosity supersolution and subsolution of \eqref{hjb}.
\end{enumerate}
\end{defi}
Next, we prove that the value function is a viscosity solution as defined in Definition \ref{viscosity_def}.
\begin{theo}\label{viscosity_sub}
The value function $v$ defined in \eqref{value_function} is a viscosity subsolution of \eqref{hjb}.
\end{theo}
\begin{proof}
   Refer to Appendix \ref{proof_subsolution}.
\end{proof}
\begin{theo}\label{viscosity_super}
The value function $v$ defined in \eqref{value_function} is a viscosity supersolution of \eqref{hjb}.
\end{theo}
\begin{proof}
    Refer to Appendix \ref{proof_supersolution}.
\end{proof}
\subsection{Uniqueness Result}
The continuity of $v$ established in Theorem \ref{continuous_v}, together with the results of Theorems \ref{viscosity_sub} and \ref{viscosity_super}, and the at most polynomial growth of the power function, show that $v$ is a viscosity solution of \eqref{hjb}. Although the techniques used to prove the comparison theorem are standard, the regime switching and jump terms render it nontrivial, motivating us to present the full argument here.
\begin{theo}
[Strong comparison principle]
\label{comparaison} If $v_i$ is a continuous viscosity subsolution of \eqref{hjb} and
$w_i$
is a continuous viscosity supersolution of \eqref{hjb}, such that 
$$v_i(t,\midbar{X},y)\leq w_i(t,\midbar{X},y),~~\text{and}~~ v_i(T, x,y)\leq
w_i(T, x,y),$$
for all $(i,t,x,y)\in \I_m\times\midbar \cS$, then $v_i \leq w_i$ on $\cS$.
\end{theo}
\begin{proof}
Let $\beta>0$ such that
\begin{equation}
    \label{alpha_comparaison}
\lim_{y\to+\infty} \max_{i\in I_m}\frac{\Phi_i(y+\midbar{X}-x)-\Phi_i(y)}{y^\beta}=0,\quad \forall (x, y) \in [0,\midbar{X}]\times
\RR_+.\end{equation}
The existence of such a constant $\beta$ follows from Assumption \ref{carnetinfini}. Let $i\in\I_m$, and define $\varphi_i:\midbar\cS\rightarrow\RR$ such that
$$\varphi_i(t,x,y):=-e^{-c t}\big((-a_1x+a_2)y^\beta-b_1x+b_2\big),\quad \forall (t, x, y) \in\midbar\cS,$$
where $a_1$, $a_2$, $b_1$, $b_2$ and $c$ are positive constants satisfying
$a_2>a_1\midbar{X}$ and $b_2>b_1\midbar{X}$. Additionally, we set $a(x):=-a_1x+a_2>0$, for all $x\in[0,\midbar{X}] $. Fix an arbitrary point $z := (t,x,y)\in\midbar \cS$. Hence, we obtain the following inequality
\begin{equation*}
\begin{split}
\big(\frac{\partial \varphi_i}{\partial
t}+\mathcal{L}\varphi_i\big)(z) & = e^{-c
t}\big(c(a(x)y^\beta-b_1x+b_2)-\frac{\beta(\beta-1)}{2}a(x)\sigma(y)^2y^{\beta-2}\big)\\
&\quad +e^{-c t}\big(-\beta h(y)a(x)y^{\beta-1}+\lambda_t
a(x)\int_{\mathbb{R}_+}((y+q(y,z))^\beta-y^\beta)\, \nu(\ud z)\big)\\
& \geq a(x)e^{-c
t}y^\beta\big(c-\beta\frac{h(y)}{y}-\frac{\beta(\beta-1)}{2}\frac{\sigma(y)^2}{y^{2}}-\bar\lambda\int_{\RR_+}\nu(\ud z)\big)+c
e^{-cT} \big(-b_1\midbar{X}+b_2\big)\\
    & \geq c e^{-cT} \big(-b_1\midbar{X}+b_2\big).
\end{split}
\end{equation*} 
The final passage follows from the linear growth assumption \ref{linear_growth} on $h$ and $\sigma$ and is true for $c$ big enough. Moreover, knowing that $\psi_i\geq 0$, we have
$$\big(\frac{\partial \varphi_i}{\partial x}+\frac{\partial
\varphi_i}{\partial y}+\psi_i\big)(z)=e^{-ct}\big(a_1 y^\beta+b_1-\beta a(x)y^{\beta-1}\big)+\psi_i(y)\geq
\varepsilon,$$
for $\eps>0$ and $b_1$ big enough compared to $a_1$ and $a_2$. Define $v_{i,m}:\midbar\cS\rightarrow\RR$ such that
$$v_{i,m}:= v_i + \frac1m \varphi_i,$$
with $m\in\NN^*$. Combining \eqref{alpha_comparaison} with the continuity of $\varphi_i$ on
$\midbar \cS$, we get that
$$\lim_{y\to +\infty}\frac{1}{m}\varphi_i(t,x,y)+\Phi_i(y+\midbar{X}-x)-\Phi_i(y) =
-\infty,\quad\forall (t,x)\in[ 0,T] \times[ 0,\midbar{X}].$$
Note that $\sum_{j\neq i}(\varphi_j-\varphi_i)Q_{ij} = 0$. Therefore, $v_{i,m}$ is a strict subsolution of equation \eqref{hjbthm} in the sense that it satisfies
\begin{equation}
    \label{ineq_strict_sub}
  \max\Big(-\frac{\partial v_{i,m}}{\partial
  t}-\mathcal{L}v_{i,m}-\sum_{j\neq i}(\varphi_j-\varphi_i)Q_{ij},\  -\frac{\partial v_{i,m}}{\partial x}-\frac{\partial
  v_{i,m}}{\partial y}-\psi_i\Big)\leq-\frac\varepsilon m<0,\text{ on }\cS.
\end{equation}
\paragraph{Toward Ishii's lemma.}
Our goal is to show that $\varrho := \sup_{z\in S} v_{i,m}(z)-w_i(z) \leq 0$. Suppose on the contrary that
$\varrho > 0$. Based on the growth condition \eqref{growthcondition}, we know that $$v_i(z)\leq \Phi_i(y+\midbar{X}-x)-\Phi_i(y),\quad \forall z\in\cS.$$ Therefore, 
$$\lim_{x
\to \midbar{X}} v_{i,m} - w_i\leq 0.$$
Additionally, using the previous results, we get, 
\begin{equation*} \lim_{y
\to +\infty} v_{i,m} - w_i = -\infty \mbox{, }\mbox{ and } \lim_{ t \to T} v_{i,m} -
w_i \leq 0.\end{equation*} This means that the supremum is attained an interior point $z_0:=(t_0,x_0,y_0) \in \midbar\cO\subset\cS$, with $\cO:=]0,T[ \times
]0,\midbar{X}[\times ]0,\midbar Y[$. In other words, we have $\varrho
= v_{i,m}(z_0)-w_i(z_0)$. Let $k\geq 1$ and define, for all $(t,x,y,y')\in\cS\times\RR_+$, $$H_k(t,x,y,y') := v_{i,m}(t,x,y) - w_i(t,x,y')
-d_k(y,y'),$$ where $d_k(y,y'):= \frac{k}{2}|y-y'|^2$. Define $\varrho_k :=
\underset{\midbar\cO\times [0,\midbar Y]}{\sup} H_k(t,x,y,y')$. Since $H_k$ is continuous and coercive, its supremum is
attained at some point $(\hat{t}_k,\hat{x}_k,\hat{y}_k, \hat{y}'_k)
\in \midbar\cO\times [0,\midbar Y]$. By means of the Bolzano–Weierstrass theorem, there exists a subsequence $(\hat{t}_{n_k},\hat{x}_{n_k},\hat{y}_{n_k}, \hat{y}'_{n_k})$ that converges to a point
$(\hat{t}_0,\hat{x}_0,\hat{y}_0, \hat{y}'_0)\in\midbar\cO\times [0,\midbar Y]$ as $k \to +\infty$. In the following, we will continue using $k$ as an index instead of $n_k$ to avoid the proliferation of indices. For $k$ large
enough, we can then assume that $\hat{t}_k < T,$ and $\hat{x}_k > 0$. Consider the following inequality
$$ H_k(\hat{t}_0,\hat{x}_0,\hat{y}_0,\hat{y}_0) \leq
H_k(\hat{t}_k,\hat{x}_k,\hat{y}_k, \hat{y}'_k).$$
In particular, we have
\begin{equation*}
\frac{k}2 |\hat{y}_k - \hat{y}'_k|^2 \leq - v_{i,m}(\hat{t}_0,\hat{x}_0,\hat{y}_0)
+ w_i(\hat{t}_0,\hat{x}_0,\hat{y}_0) + v_{i,m}( \hat{t}_k,\hat{x}_k,\hat{y}_k) -
w_i(\hat{t}_k,\hat{x}_k,\hat{y}'_k).
\end{equation*} As $v_{i,m}$ and $w_i$ are continuous on the compact set $\midbar \cO$, there exists $C>0$ such that
\begin{equation}
\label{comparaison_diff_limit}
|\hat{y}_k - \hat{y}'_k|^2  \leq  \frac{C}k. 
\end{equation}
Letting $k$ go to $+\infty$, we find $\hat{y}_0 = \hat{y}'_0$. Finally, we show that $\varrho_k$ tends to $\varrho$ when $k$ goes to $+\infty$. Note that $$\varrho=v_{i,m}(z_0) - w_i(z_0)= 
H_k(t_0, x_0, y_0, y_0)\leq H_k(\hat{t}_k,\hat{x}_k,\hat{y}_k,\hat{y}'_k).$$ Therefore, $\varrho
\leq \varrho_k$. Moreover, we
have
\begin{equation*}
\varrho_k = v_{i,m}(\hat{t}_k,\hat{x}_k,\hat{y}_k) - w_i(\hat{t}_k,\hat{x}_k,\hat{y}'_k) -
\frac{k}{2} |\hat{y}_k - \hat{y}'_k|^2 \leq v_{i,m}(\hat{t}_k,\hat{x}_k,\hat{y}_k) -
w_i(\hat{t}_k,\hat{x}_k,\hat{y}'_k).
\end{equation*} Since $v_{i,m}$ and $w_i$ are continuous on $\cS$, we get that $$\lim_{k\rightarrow+\infty}v_{i,m}(\hat{t}_k,\hat{x}_k,\hat{y}_k) -
w_i(\hat{t}_k,\hat{x}_k,\hat{y}'_k) = v_{i,m}(\hat{t}_0,\hat{x}_0,\hat{y}_0) -
w_i(\hat{t}_0,\hat{x}_0,\hat{y}_0) \leq \varrho.$$ We conclude that $\lim_{k\rightarrow+\infty}\varrho_k = \varrho$ and
$\lim_{k\rightarrow+\infty} |\hat{y}_k - \hat{y}'_k|^2 = 0$. Moreover, we
have $$v_{i,m}(\hat{t}_0,\hat{x}_0,\hat{y}_0) - w_i(\hat{t}_0,\hat{x}_0,\hat{y}_0) = \varrho.$$
\paragraph{Ishii's lemma.} We now apply Theorem $3.2$ from \textcite{users_guide} at the point 
$(\hat{t}_k, \hat{x}_k, \hat{y}_k, \hat{y}'_k)$, ensuring the existence of 
$M, M' \in \mathbb{R}$ such that
 \begin{equation*} \begin{split} 
 \left(
   \begin{array}{cc}
      - {k} - \|A\| & 0 \\
     0 &  - {k} - \|A\| \\
   \end{array}
 \right)
 \leq \left(\begin{array}{cc}
        M & 0 \\
        0 & -M' \\
\end{array}
\right) \leq A + \frac{1}{k} A^2, \end{split} \end{equation*} with $A  = D^2 d_k (\hat{y}_k,\hat{y}'_k) = \left(
\begin{array}{cc}                              k & -k \\
    -k & k \\
\end{array}
\right)$. Let $\cK$ and $\cI$ be the operators defined by
\begin{equation}
\label{operators}
    \cK[p, M ] (y) := \frac{\sigma^2(y)}{2} M -h(y)p,~~\text{and}~~\cI[g](t,x,y) := \int_{\mathbb{R}}(g(t,x,y+ q(y,\zeta))-g(t,x,y))\nu(\ud\zeta).\end{equation}
Using the relationship between superjets and Definition \ref{viscosity_def} of viscosity supersolutions, along with inequality \eqref{ineq_strict_sub}, we deduce from Ishii's Lemma that
\begin{equation}
\label{inegalitesishii1}
\varepsilon/m \leq  \min\Big( \mathcal{K}[  \frac{\partial d_k}{\partial
y}(\hat{y}_k,\hat{y}'_k),M ]  (\hat{y}_k)+ \lambda_{\hat{t}_k}
\cI[v_{i,m}](\hat{t}_k,\hat{x}_k,\hat{y}_k) ,\ \frac{\partial d_k}{\partial y}
(\hat{y}_k,\hat{y}'_k)+\psi_i(\hat{y}_k)\Big),
\end{equation}
\begin{equation}
\label{inegalitesishii2}
0  \geq  \min\Big(-\mathcal{K}[\frac{\partial d_k}{\partial
y'}(\hat{y}_k,\hat{y}'_k),-M' ]  (\hat{y}'_k)+ \lambda_{\hat{t}_k} \cI[w_i](\hat{t}_k,
\hat{x}_k,\hat{y}'_k),\ -\frac{\partial d_k}{\partial y'}
(\hat{y}_k,\hat{y}'_k)+\psi_i(\hat{y}'_k)\Big).
\end{equation}
From the second inequality, we have two cases:
\begin{enumerate}
\item $-\frac{\partial d_k}{\partial y'} (\hat{y}_k,\hat{y}'_k)+\psi_i(\hat{y}'_k)\leq 0.$
\item $-\mathcal{K}[\frac{\partial d_k}{\partial y'},-M' ] 
(\hat{y}'_k)+ \lambda_{\hat{t}_k} \cI[w_i](\hat{t}_k, \hat{x}_k,\hat{y}'_k)\leq 0.$
\end{enumerate}
In the first case, we have $\psi_i(\hat{y}'_k)\leq \frac{\partial d_k}{\partial y'} (\hat{y}_k,\hat{y}'_k)$. Applying inequality \eqref{inegalitesishii1}, it follows that 
\begin{equation*}
\begin{split}
\frac  \varepsilon m &\leq 
\frac{\partial d_k}{\partial y} (\hat{y}_k,\hat{y}'_k)+\psi_i(\hat{y}_k)-\psi_i(\hat{y}'_k)+\psi_i(\hat{y}'_k)\\
& \leq 
\frac{\partial d_k}{\partial y} (\hat{y}_k,\hat{y}'_k)+\frac{\partial d_k}{\partial y'} (\hat{y}_k,\hat{y}'_k)+\psi_i(\hat{y}_k)-\psi_i(\hat{y}'_k)\\
& \leq \psi_i(\hat{y}_k)-\psi_i(\hat{y}'_k).
\end{split}
\end{equation*}
Since $\psi_i$ is continuous at $\hat{y}_0$, it follows that $\lim_{k\rightarrow+\infty}\psi_i(\hat{y}_k)-\psi_i(\hat{y}'_k) = 0$, which leads to a contradiction. In the second case, applying inequality \eqref{inegalitesishii1} once more yields
\begin{equation*}
\begin{split}
&\frac{1}{2}(M\sigma(\hat{y}_k)^2-M^\prime\sigma(\hat{y}'_k)^2)-k(h(\hat{y}_k)-h(\hat{y}'_k))(\hat{y}_k-\hat{y}'_k) +\lambda_{\hat{t}_k} \left(\cI[v_{i,m}](\hat{t}_k,
\hat{x}_k,\hat{y}_k)- \cI[w_i](\hat{t}_k,
\hat{x}_k,\hat{y}'_k)\right)
 \\
&~ = \mathcal{K}[  \frac{\partial d_k}{\partial
y}(\hat{y}_k,\hat{y}'_k),M ] (\hat{y}_k) + \lambda_{\hat{t}_k}
\cI[v_{i,m}](\hat{t}_k,\hat{x}_k,\hat{y}_k) +\mathcal{K}[\frac{\partial d_k}{\partial y'}(\hat{y}_k,\hat{y}'_k),-M'](\hat{y}'_k)- \lambda_{\hat{t}_k} \cI[w_i](\hat{t}_k,
\hat{x}_k,\hat{y}'_k)
\\&~\geq \frac \varepsilon m .
\end{split}
\end{equation*}
Using the continuity of $h$, $\lambda$, $v_{i,m}$ and $w_i$, we deduce that
\begin{equation*}
\begin{split}
\frac \varepsilon m & \leq
\lambda_{\hat{t}_0} \left(\cI[v_{i,m}]- \cI[w_i]\right)(\hat{t}_0,
\hat{x}_0,\hat{y}'_0)+\lim_{k\to+\infty}\frac{1}{2}\Big\langle\left(\begin{array}{cc}M& 0\\ 0 &
-M^\prime\end{array}\right)(\sigma(\hat{y}_k),\sigma(\hat{y}'_k))^T,(\sigma(\hat{y}_k),\sigma(\hat{y}'_k))\Big\rangle \\
& \leq\lambda_{\hat{t}_0} \cI[H_0](\hat{t}_0,\hat{x}_0,\hat{y}_0)+\lim_{k\to+\infty}\frac{1}{2}\Big\langle\big(A+\frac{1}{k}
A^2\big)(\sigma(\hat{y}_k),\sigma(\hat{y}'_k))^T,(\sigma(\hat{y}_k),\sigma(\hat{y}'_k))\Big\rangle ,
\end{split}
\end{equation*}
with $H_0(t,x,y):= H_0(t,x,y,y)$, for all $(t,x,y)\in\cS$. By definition of $\rho = H_0(\hat{t}_0,\hat{x}_0,\hat{y}_0)$, we get $$H_0(\hat{t}_0,\hat{x}_0,\hat{y}_0+ q(y_0,\zeta))-H_0(\hat{t}_0,\hat{x}_0,\hat{y}_0)\leq 0,\quad \forall \zeta\in\RR.$$ Consequently, it follows that $\cI[H_0](\hat{t}_0,\hat{x}_0,\hat{y}_0)\leq 0$ and
\begin{equation*}
\begin{split}
\frac \varepsilon m&\leq\lim_{k\to+\infty} \frac{1}{2}\Big\langle\big(A+\frac{1}{k}
A^2\big)(\sigma(\hat{y}_k),\sigma(\hat{y}'_k))^T,(\sigma(\hat{y}_k),\sigma(\hat{y}'_k))\Big\rangle = \lim_{k\to+\infty}\frac{3}{2}k(\sigma(\hat{y}_k)-\sigma(\hat{y}'_k))^2\leq \frac{3}{2}C.
\end{split}
\end{equation*}
 The last equality follows from the Lipschitz continuity of $\sigma$ and inequality \eqref{comparaison_diff_limit}. Since $\eps$ is arbitrary, this results in a contradiction, thereby concluding the proof.
\end{proof}
\subsection{The Free Boundary Problem}
The value function $v_i$ is differentiable almost everywhere by Proposition \ref{differentiable_vf}. Let $\cN_i$ denote the set of Lebesgue measure zero where $v_i$ may not be differentiable. We now define the associated free boundary.
\begin{defi}
\label{regions_def}
    We define the \textit{exercise region} $\cE_i := \midbar{\textrm{int}(\cE^{diff}_i)}$ as the closure of the interior of the set $\cE^{diff}_i$, with
 \begin{equation*}
        \cE^{diff}_i := \big\{(t,x,y)\in \cS\,\backslash\,\cN_i:  -\frac{\partial v_i}{\partial x}-\frac{\partial v_i}{\partial y}-\psi_i=0\big\},
 \end{equation*}
and we define the \textit{continuation region} $\cC_i := \midbar\cS\,\backslash\, \cE_i$ as its complement. \end{defi}
We now study the geometry of the free boundary $\partial \cE_i:=\midbar \cC_i\cap \cE_i$ separating the continuation region $\cC_i$ and the exercise region $\cE_i$. 
\begin{lemma}[Partial smooth-fit principle]
\label{monotonie_limit_boundary}
For all $(t_0, x_0, y_0) \in \overline{\cS}\,\backslash\,\cN_i$, the mappings
$$
y \mapsto \partial_x v_i(t_0, x_0, y)
~~\text{and}~~
x \mapsto \partial_y v_i(t_0, x, 0)
$$
are continuous almost everywhere on $\RR_+$ and on the interval $[0,\midbar X]$, respectively.
\end{lemma}
\begin{proof}
     We know that $v_i$ is jointly continuous on $\midbar\cS$ and uniformly continuous in the variable $y$. Additionally, the map $x \mapsto v_i(t, x, y)$ is differentiable for almost each $x \in [0,\bar X]$. Define, for $h \ne 0$, 
$$
D_h(t, x, y) := \frac{v_i(t, x+ h, y) - v_i(t, x, y)}{h},\quad \forall z\in\midbar\cS\,\backslash\,\cN_i.
$$
Based on the proof of Theorem \ref{continuous_v}, we have that
$$|D_h(t, x, y) - D_h(t, x, y')|\leq \frac{w(|y-y'|)}{|h|},$$
with $w$ a continuous function on $\RR_+$. Hence, the family 
$\{D_h\}_h$ is equicontinuous in $y$ on compact subsets of $\midbar\cS$. Given the pointwise convergence $D_h \underset{h\to 0}{\to} \partial_x v_i$ almost everywhere, the equicontinuity and uniform boundedness in $y$, the Arzelà-Ascoli Theorem ensures uniform convergence in $y$ over compact subsets of $\midbar\cS$, for almost every $x \in [0, \midbar{X}]$. Hence, $\partial_x v_i$ is the uniform limit of continuous functions in $y$ and $y \mapsto \partial_x v_i(t, x, y)$ is continuous almost everywhere on $\RR_+$. 

Similarly, the map $x \mapsto \partial_y v_i(t, x, 0)$ is continuous almost everywhere on $[0,\midbar X]$. This follows from a combination of Rademacher’s Theorem and the fact that the difference quotients for $(t,x)\mapsto\partial_y v_i(t,x,0)$ converge uniformly in $x$ on compact sets due to the joint continuity of $v_i$ and uniform bound in $x$ with respect to $t$. This concludes the proof.
    \end{proof}
\begin{prop}[Connectedness]
\label{connected}
Assume that $\psi_i$ is continuous, that the interiors of $\cC_i$ and $\cE_i$ are non-empty, and that $\sigma + \lambda > 0$. Then, for each $i \in \I_m$, the free boundary $\partial \cE_i$ is non-empty and path-connected. Moreover, if $(t, x, y) \in \cE_i$, then
$$
(t, x, y') \in \cE_i \quad \text{for all } 0 \leq y' \leq y, \quad \text{and} \quad (t, x', 0) \in \cE_i \quad \text{for all } 0 \leq x' \leq x.
$$
\end{prop}
\begin{proof}
As $\mathring{\cC_i}$ and $\mathring{\cE_i}$ are non-empty, we can consider $z_0:=(t_0,x_0,y_0)\in\mathring{\cE_i}$ be an interior point of $\cE_i$. By contradiction, we suppose that $\{0<y<y_0: (t_0,x_0,y)\in\cC_i\}$ is non-empty. Being interior points, each admits a neighborhood fully contained in $\mathcal{C}_i$ or $\mathcal{E}_i$, respectively. Define the boundary point $z_1 := (t_0,x_0,y_1)$ with $y_1:=y_0-\delta$. Let $\delta:=\sup\{0<y<y_0: (t_0,x_0,y)\in\cC_i\}$ and $\delta'>0$ such that $\{(t_0,x_0,y): y_0-\delta-\delta'<y<y_0-\delta\}\subset\cC_i$. Define $\eps:=\min\{\delta,\delta'\}$. The following holds, 
    \begin{equation*}
\resizebox{\textwidth}{!}{$
\begin{aligned}\cO_c := \{(t_0,x_0,y): y_0-\delta-\eps<y<y_0-\delta\}\subset\cC_i,~~\textrm{and}~~\cO_e := \{(t_0,x_0,y): y_0-\delta<y<y_0-\delta+\eps\}\subset\cE_i.\end{aligned}
$}
\end{equation*}
Figure \ref{free_bound} provides a graphical representation of the previously defined sets. 
\begin{figure}[H]
        \centering
        \includegraphics[width=0.5\textwidth]{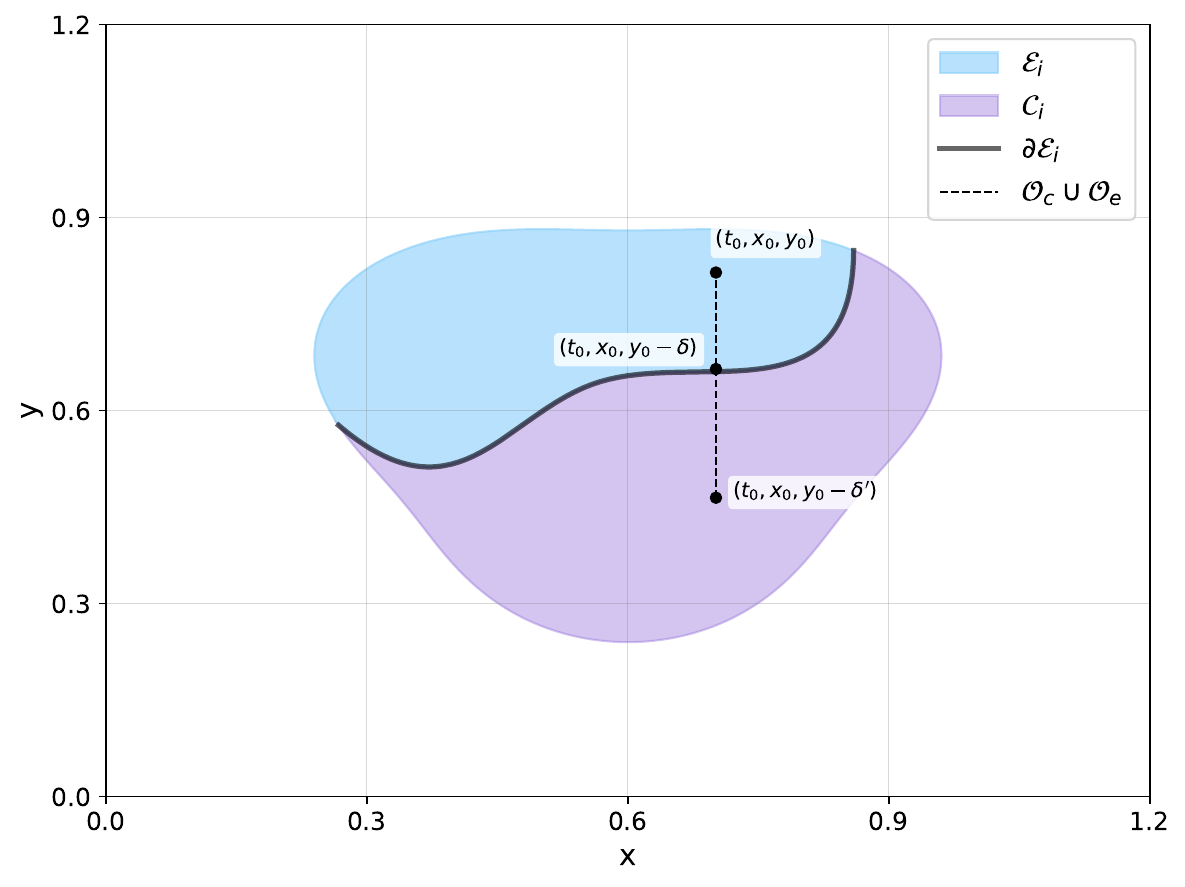}
    \caption{Illustration of $\cO_e$ and $\cO_c$.}
    \label{free_bound}
\end{figure}
In this setting, define the function $f_0:\cO_c\cup\cO_e\cup\{(t_0,x_0,y_1)\} \to \mathbb{R}$ satisfying
\begin{equation}
\label{e_0}
 f_0(t_0,x_0,y) := \begin{cases}
     
 -\partial_x^+ v_i(t_0,x_0,y_1) - \partial^-_y v_i(t_0,x_0,y) - \psi_i(y) ,&\textrm{if} ~~y< y_1. \\-\partial_x^+ v_i(t_0,x_0,y_1) - \partial^+_y v_i(t_0,x_0,y) - \psi_i(y), &\textrm{if} ~~y\geq y_1.
 \end{cases}
\end{equation}
Given that $f_0(t_0, x_0, y_1) = \liminf_{y \to {y_1}^+} f_0(t_0, x_0, y) = 0$, two scenarios may arise: $$\liminf_{y \to {y_1}^-} f_0(t_0, x_0, y) < \eta_0<0,~~\textrm{or}~~\liminf_{y \to {y_1}^-} f_0(t_0, x_0, y) = 0.$$
\textbf{Case 1.} $\liminf_{y\to {y_1}^-}f_0(t_0,x_0,y)<\eta_0<0$.\\
  Define the function $\varphi_{\eps}:\cS\rightarrow \RR$, for $\epsilon>0$ and $z\in\cS$, as $$\varphi_{\epsilon}(z):=v_i(z_1) + \partial^+_t v_i(z_1)(t- t_0) + \partial_x^+v_i(z_1)(x- x_0) -(\partial^+_x v_i(z_1) + \psi_i(y_1) + \frac{\eta_0}{2})(y- y_1) -\frac{1}{2\epsilon} (y-y_1)^2.$$
Since $v_i$ is continuous and non-decreasing function in $y$, it holds, as established in \textcite{rudin1976principles}, that 
$$\limsup_{y \to y_1^-} \partial_y^-v_i(t_0,x_0,y) \leq \partial_y^-v_i(t_0,x_0,y_1) \quad \text{and} \quad \partial_y^+v_i(t_0,x_0,y_1)  \leq \liminf_{y \to y_1^+} \partial_y^+v_i(t_0,x_0,y).$$
Using Taylor's expansion, we have that for $t\geq t_0$, $x\geq x_0$ and $\|z- z_1\|$ small enough,
$$v_i(z)-\varphi_{\eps}(z) \leq g_{i}(z)(y- y_1) + \operatorname{o}(\|z- z_1\|),$$
where $g_{i} : \mathcal{S} \to \mathbb{R}$ are non-positive and defined, for all $z\in\cS$, as
$$g_{i}(z) := \begin{cases}
-\liminf_{y\to {y_1}^-}f_0(t_0,x_0,y) + \frac{\eta_0}{2} & \text{if } y < y_1, \\ \frac{\eta_0}{2} & \text{else} .
\end{cases}$$
Since $g_{i}(z)(y-y_1)<\frac{\eta_0}{2}|y-y_1|<0$, then $\varphi_{\epsilon}$ is a test function in $C^2(\mathcal{S})$ such that $z_1$ achieves a local maximum of $v_i-\varphi_{\epsilon}$ and $v_i(z_1) = \varphi_{\epsilon}(z_1)$. Using Theorem \ref{viscosity_super}, we get that $v_i$ is a viscosity subsolution of \eqref{hjb}. Therefore, 
 $$\max\Big(-\big(\frac{\partial
  \varphi_{\epsilon}}{\partial t}+\mathcal{L}\varphi_{\epsilon}+\sum_{j\neq i}Q_{ij}(v_j-\varphi_{\epsilon})\big)(z_1) ,\ -\big(\frac{\partial
  \varphi_{\epsilon}}{\partial x}+\frac{\partial \varphi_{\epsilon}}{\partial y}+\psi_i\big)(z_1)\Big) \leq 0.$$
We know that $$\mathcal{L}\varphi_{\epsilon}(z_1) = -\frac{1}{2\eps}\sigma^2(y_1)-\big(\partial^+_x v_i(z_1) + \psi_i(y_1) + \frac{\eta_0}{2}\big)\big(h(y_1)+\lambda_t\int_{\RR}q(y_1, z)\nu(\ud z)\big) -\frac{1}{2\epsilon} \lambda_t\int_{\RR}q^2(y_1, z)\nu(\ud z).$$ 
Hence, we have $\lim_{\eps\rightarrow 0}-\cL\varphi_{\eps}(z_1) = +\infty$, while $ -\big(\mathcal{L}\varphi_{\epsilon}+\sum_{j\neq i}Q_{ij}(v_j-v_i)\big)(z_1) \leq 0$ which leads to a contradiction and the initial assumption does not hold true.\\
\textbf{Case 2.} $\liminf_{y\to {y_1}^-}f_0(t_0, x_0,y)=0$.\\
Since $\liminf_{y \to {y_1}^-} f_0(t_0, x_0, y) \leq \limsup_{y \to {y_1}^-} f_0(t_0, x_0, y)\leq 0$, then $\lim_{y \to {y_1}^-} f_0(t_0, x_0, y) = f_0(t_0, x_0, y_1)$. Hence, $y\mapsto f_0(t_0, x_0,y)$ is left-continuous in $y_1$. The fact that $f_0(t_0, x_0, y) < 0$ for almost every $y < y_1$, combined with $f_0(t_0, x_0, y_1) = 0$, and the left-continuity of $f_0$, implies the existence of $\delta_0 > 0$ such that $y\mapsto f_0(t_0, x_0, y)$ is increasing on $[y_1 - \delta_0, y_1]$. Moreover, since $\psi_i$ is continuous on $\RR_+$, then $y\mapsto\partial^-_y v_i(t_0, x_0, y)$ is also continuous on the interval $[y_1 - \delta_0, y_1]$. Given that $\psi_i$ is non-decreasing on $\mathbb{R}_+$ and positive, and that
$$
\partial^-_y v_i(t_0, x_0, y) \geq 0, \quad \forall y > 0,
$$
we conclude that for $f_0$ to vanish at $z_1$, the function $y\mapsto\partial^-_y v_i(t_0, x_0, y)$ must be decreasing on $[y_1 - \delta_0, y_1],$ otherwise $f_0$ would remain negative in $z_1$. Therefore, the map $y \mapsto v_i(t_0, x_0, y)$ is strictly concave on $[y_1 - \delta_0, y_1]$, since the existence of a decreasing left derivative implies strict concavity (see \textcite[Theorem 24.1]{rockafellar1970convex}). This means that for $y_2\in]y_1 - \delta_0, y_1[$, there exists $c>0$ such that, for all $y\in]y_1 - \delta_0, y_1[$,
$$v_i(t_0,x_0,y)<v_i(t_0,x_0,y_2) + c(y-y_2).$$
 For any $\varepsilon > 0$, define the function $\widetilde{\varphi}_\varepsilon : \mathcal{S} \to \mathbb{R}$ by
$$
\widetilde{\varphi}_\varepsilon(z) := v_i(z_2) + \partial_t^+ v_i(z_2)(t - t_0) + \partial^+_x v_i(z_2)(x - x_0) + c (y - y_2) - \frac{1}{2\varepsilon}(y - y_2)^2.
$$
It follows that $v_i - \widetilde{\varphi}_\varepsilon $ is negative with respect to $y$ in a neighborhood of $z_2$, excluding $z_2$ itself. In particular, $\widetilde{\varphi}_\varepsilon\in C^2(\mathcal{S})$ strictly dominates $v_i$ locally around $z_2$. Repeating the same reasoning as in the previous case thus leads to a contradiction. 

Similarly, we have that
$$
(t, x, 0) \in \cE_i\quad\Rightarrow \quad (t, x', 0) \in \cE_i \quad \text{for all } 0 \leq x' \leq x.
$$
As a result, the exercise region always lies below the continuation region with respect to the volume variable. That is, for any fixed $(t_0, x_0) \in [0,T] \times [0,\midbar X]$, it is not possible to have $0 \leq y_1 < y_2 < y_3$ such that $\{(t_0, x_0, y) : y_2 < y < y_3\} \subseteq \cE_i$ while simultaneously $\{(t_0, x_0, y) : y_1 < y < y_2\} \subseteq \cC_i$. A similar argument establishes that the exercise region is always to the left of the continuation region $\cC_i$ in terms of the inventory dimension when $y=0$. Consequently, the free boundary must be connected.
\end{proof}
  \begin{rque}
      The non-emptiness assumption in Proposition \ref{connected} is not restrictive, as it merely excludes trivial cases where $\mathcal{S} = \mathcal{E}_i$ or $\mathcal{S} = \mathcal{C}_i$.
  \end{rque}
\section{Approximation Scheme}
\label{numerical_scheme}
We introduce a numerical scheme to solve the PIDE \eqref{hjb}, employing a discretization approach that ensures stability and accuracy in approximating the integral-differential operator. We apply a standard finite difference scheme to discretize the differential terms while handling the integral term using a suitable quadrature rule, following the approach of \textcite{finite_diff} in the finite activity case.
\paragraph{Discretization.} The time domain $[0,T]$ is discretized into $N_T$ steps with $\Delta t := T/N_T$, such that $t_n := T - n\Delta t$, with $n\in\{0,\dots,N_T\}$. To ensure numerical feasibility, we localize $y$, which is normally in $\RR_+$, by restricting it to a bounded domain of size $\midbar Y$. The state space is discretized into $N_X$ and $N_Y$ steps, respectively. The spatial domain is discretized into $(N_X+1) \times (N_Y+1)$ grid points. The spatial step sizes are defined as $\Delta x := \midbar X /N_X$ and $\Delta y := \midbar Y/N_Y$, leading to the spatial grid points
$$(x_k, y_l) := ((k-1)\Delta x, (l-1)\Delta y), \quad \forall (k,l) \in \I_{N_X+1}\times\I_{N_Y+1}.$$ A function evaluated at a grid point $t_n$ is denoted with the superscript $n$. The matrix  
$v^{n} := (v^{n}_i)_{i\in \I_m}$ has entries representing $v^{n}_i := \big(v^n_i(x_k,y_l)\big)_{(k,l)\in \I_{N_X+1}\times\I_{N_Y+1}}$.
\paragraph{Finite Difference Scheme.} The scheme follows a fully implicit-explicit method to ensure stability. The HJBQVI is discretized, for all $i\in\I_m$ and $n\in\{1,\dots,N_T-1\}$, as
\begin{equation*}
\max \Big(-\frac{v_i^{n+1} - v_i^n}{\Delta t} - \cJ v_i^{n} -\cI v^{n+1}_i- \sum_{j\neq i}(v_j^{n+1} - v_i^{n+1}) Q^{n+1}_{ij},\ -\frac{\partial v^{n+1}_i}{\partial x} - \frac{\partial v^{n+1}_i}{\partial y} - \psi_i \Big) = 0,
\end{equation*}
with $\cJ(g) := \cK[\frac{\partial g}{\partial y},\frac{\partial^2 g}{\partial y^2}] = \frac{\sigma^2(y)}{2} \frac{\partial^2 g}{\partial y^2} -h(y)\frac{\partial g}{\partial y}$, and the operators $\cK$ and $\cI$ are defined in \eqref{operators}. The first-order derivatives are approximated as
\begin{equation*}
\frac{\partial v^n_i}{\partial x}(x_k,y_l) \approx \frac{v_i^n(x_{k+1}, y_l) - v_i^n(x_k,y_l)}{\Delta x},~~\text{and}~~
\frac{\partial v^n_i}{\partial y}(x_k,y_l) \approx \frac{v_i^n(x_k,y_{l+1}) - v_i^n(x_k,y_l)}{\Delta y},
\end{equation*}
for all $(k,l)\in\{1, \dots, N_X\}\times\{1, \dots, N_Y\}$. The second-order derivative is discretized using a centered finite difference scheme for interior points, while first-order \textit{Neumann conditions} are imposed at the boundaries. In other words, for all $k\in\I_{N_X+1}$,
$$
\frac{\partial^2 v^n_i}{\partial y^2}(x_k,y_l) \approx
\begin{cases}
\frac{v^n_i(x_k,y_3) - 2 v^n_i(x_k,y_2) + v^n_i(x_k,y_1)}{\Delta y^2}, & \text{if }l = 1, \\
\frac{v^n_i(x_k,y_{l+1}) - 2 v^n_i(x_k,y_l) + v^n_i(x_k,y_{l-1})}{\Delta y^2}, & \text{if }2\leq l \leq N_Y , \\
\frac{v^n_i(x_k,y_{N_Y+1}) - 2 v^n_i(x_k,y_{N_Y}) + v^n_i(x_k,y_{N_Y-1})}{\Delta y^2}, & \text{if }l = N_Y+1.
\end{cases}$$
The integral term in $\cI$ is discretized using a trapezoidal quadrature scheme, such that for all $(k,l)\in\I_{N_X+1}\times\I_{N_Y+1}$, we have
\begin{equation*}
\cI v^n_i(x_k,y_l) \approx \lambda_{t_n} \sum_{p=1}^{N_Z} w_p \big( v_i^n(x_k, y_l+ q(y_l, z_p)) - v_i^n(x_k, y_l) \big),
\end{equation*}
where $z_p$ are quadrature points, and $w_p$ are corresponding quadrature weights such that $\sum_{p=1}^{N_Z} w_p = \int_{\RR}\nu(\ud z)$. Since $q(y_l, z_p)$ does not necessarily coincide with a grid point, interpolation is necessary. In such cases, the shifted value $y_l + q(y_l, z_p)$ is approximated using linear interpolation between the closest grid points $y_j$ and $y_{j+1}$, satisfying $y_l + q(y_l, z_p) \in [y_{j(p)}, y_{j(p)+1}]$. The interpolation formula is given by $v^{n}_i(x_k,y_l + q(y_l, z_p)) \approx v^{n,p}_i(x_k,y_l)$, with
$$
v^{n,p}_i(x_k,y_l):= (1 - \alpha_p) v^n_i(x_k, y_{j(p)}) + \alpha_p v^n_i(x_k, y_{j(p)+1}),
$$
where $\alpha_p = \frac{y_l + q(y_l, z_p) - y_{j(p)}}{y_{j(p)+1} - y_{j(p)}}$. 
Substituting this into the discretized integral, we obtain the formulation
$$
\mathcal{I} v^n_i(x_k,y_l) \approx \lambda_{t_n} \sum_{p=1}^{N_Z} w_p \big( v^{n,p}_i(x_k,y_l) - v^n_i(x_k,y_l)\big),\quad \forall (k,l)\in\I_{N_X+1}\times\I_{N_Y+1}.
$$ 
Note that we handle the integral part using an implicit time-stepping approach to circumvent the inversion of the dense matrix $\cI$. Define $F_i^n$ and $G_i^n$ such that
\begin{equation*}
\resizebox{\textwidth}{!}{$
\begin{aligned}
        F_i^{n}(x_k,y_l) &:= \Big(-\frac{v_i^{n} - v_i^{n-1}}{\Delta t} - \cJ v_i^{n-1} -\cI v^{n}_i- \sum_{j\neq i}(v_j^{n} - v_i^{n}) Q^{n}_{ij}\Big)(x_k,y_l),~\forall (k,l)\in\I_{N_X+1}\times\I_{N_Y+1}, 
        \\G_i^{n}(x_k,y_l) &:= -\frac{v_i^{n}(x_{k+1}, y_k) - v_i^{n}(x_k,y_l)}{\Delta x} - \frac{v_i^{n}(x_k, y_{l+1}) - v_i^{n}(x_k,y_l)}{\Delta y} - \psi_i(y_l),~ \forall (k,l)\in\I_{N_X}\times\I_{N_Y}.\end{aligned}
$}
\end{equation*}
Since $\max(F^n_i, G^n_i) = 0$, we define the \textit{active region} as $$\cE^n_i = \{(x_k,y_l)\in[0,\midbar X]\times[0,\midbar Y]: F^{n}_i(x_k,y_l) < G^{n}_i(x_k,y_l)\},$$ and the \textit{passive region} as $$\cC^n_i = \{(x_n,y_n)\in[0,\midbar X]\times[0,\midbar Y]: F^{n}_i(x_n,y_n) \geq G^{n}_i(x_n,y_n)\}.$$ 
 Rearranging for all grid points in $\cE^{n+1}_i$ results in a linear system formulated as
\begin{equation}
\label{matrix_system_scheme}
     \big(\big(I_{m} + \Delta t Q^{n+1})\big)\otimes I_m\big) v^{n+1} + \Delta t\lambda_{t_{n+1}}\sum_{p= 1}^{N_Z}w_p (I_m\otimes I_m)(v^{n+1,p}-v^{n+1}) = \big( (I_{m} -\Delta t A)\otimes I_{m} \big)v^n,\end{equation}
where $I_{m}$ is the identity matrix of size $m\times m$. The matrix $A$ is defined as a tridiagonal matrix of size $(N_Y+1)\times (N_Y+1)$ with the structure

\vspace{0.2cm}
\begin{equation*}
A :=
\begin{bmatrix}
\mathbb{B}_1  & \mathbb{A}_1 & 0       & 0      & \dots  & 0  \\
\mathbb{C}_2 & \mathbb{B}_2  & \mathbb{A}_2 & 0      & \dots  & 0  \\
0       & \mathbb{C}_3 & \mathbb{B}_3  & \mathbb{A}_3 & \dots  & 0  \\
\vdots  & \vdots  & \vdots  & \vdots  & \ddots & \vdots \\
0       & 0       & 0       & \mathbb{C}_{N_Y} & \mathbb{B}_{N_Y}  & \mathbb{A}_{N_Y} \\
0       & 0       & 0       & 0      & \mathbb{C}_{N_Y+1} & \mathbb{B}_{N_Y+1}
\end{bmatrix}.
\end{equation*}
The coefficients $\mathbb{A}_l, \mathbb{B}_l$ and $\mathbb{C}_l$ are defined, for $l=1$ and $l=N_y+1$, as
 \begin{align*}
 \mathbb{A}_1 := \frac{\sigma^2(y_1)}{2\Delta y^2},~~\mathbb{B}_1 := -\frac{\sigma^2(y_1)}{\Delta y^2},~~\mathbb{B}_{N_Y+1} := -\frac{\sigma^2(y_{N_Y+1})}{\Delta y^2},~~\text{and}~~ \mathbb{C}_{N_Y+1} := \frac{\sigma^2(y_{N_Y+1})}{\Delta y^2}.
 \end{align*}
 For $l \in \{2, \dots, N_Y\}$,
\begin{align*}
\mathbb{A}_l := \frac{\sigma^2(y_l)}{2\Delta y^2} - \frac{h(y_l)}{2\Delta y},~~\mathbb{B}_l := -\frac{\sigma^2(y_l)}{\Delta y^2} ~~\text{and}~~ \mathbb{C}_l := \frac{\sigma^2(y_l)}{2\Delta y^2} + \frac{h(y_l)}{2\Delta y}.
\end{align*}
 On the hand, if we are in $\cC^{n+1}_i$, we have, for all $(k,l)\in\I_{N_X}\times\I_{N_Y}$,
$$-\frac{v_i^{n+1}(x_{k+1}, y_l) - v_i^{n+1}(x_k,y_l)}{\Delta x} - \frac{v_i^{n+1}(x_k, y_{l+1}) - v_i^{n+1}(x_k,y_l)}{\Delta y} - \psi^{n+1}_i(y_l) = 0.$$ The boundary and terminal conditions are enforced explicitly, ensuring that
\begin{equation*}
v^{N_T}_i(x_k, y_l) = \Phi^n_i(y_l+\midbar{X}-x_k)-\Phi^n_i(y_l),~~\text{and}~~
v^n_i(x_{N_X+1}, y_l) = 0,\quad \forall (k,l)\in\I_{N_X+1}\times\I_{N_Y+1}.
\end{equation*}
To ensure numerical stability, we impose a Courant-Friedrichs-Lewy (CFL) condition
\begin{equation*}
\Delta t \leq C \min\{(\Delta x)^2,(\Delta y)^2\},
\end{equation*}
where $C$ is a positive constant.

We select the time step based on a reference spatial scale proportional to the domain’s upper bound. This follows \textcite{morton2005numerical} and ensures stability in regions with steep gradients. The following algorithm summarizes the time-stepping scheme.
\begin{algorithm}[H]
\caption{Time-Stepping Scheme for PIDE Solution}
\begin{algorithmic}[1]
\State \textbf{Initialize:} Set terminal condition $v^{N_T}_i(x,y) = \Phi^{N_T}_i(y+\midbar{X}-x)-\Phi^{N_T}_i(y)$.
\For{$n = N_T$ to $1$}
    \State Find $v^{n-1}$ that solves the linear system \eqref{matrix_system_scheme}: $F^{n} = 0$.
    \For{$i = m$ to $1$}        
        \State Update $v^{n-1}_i$: $v^{n-1}_i(x,y) \xleftarrow{} \inf_{a\in\{0,\dots,~\midbar X-x\}} v^{n-1}_i(x+a,y+a) + \Phi_i(y+a)-\Phi_i(y)$.
        \State Find the regions $\cE_i^{n-1}$ and $\cC_i^{n-1}$.
    \EndFor
\EndFor
\State \textbf{Output:} Approximate solution $v^n_i(x,y)$.
\end{algorithmic}
\end{algorithm}
\section{Numerical Results}
\label{approximation_numerical}
In this section, we analyze the optimal trading strategy using the numerical scheme detailed in Section \ref{numerical_scheme}. We examine various limit order book shapes and illustrate the structure of the optimal exercise and continuation regions for each case. In this numerical study, we define the coefficients in a multiplicative form as
$$
h(y) = c y, \quad\sigma(y) = d y, \quad\text{and}\quad
q(y, z) = e yz, \quad \forall y,z \in \mathbb{R}_+,
$$
where $c$, $d$ and $e$ are non-negative constants. The controlled \textit{volume effect} process \eqref{controlled_volume_effect_process} in this example satisfies both Lipschitz continuity and linear growth assumptions. Here, the drift term is proportional to $y$, ensuring a Geometric Brownian Motion (GBM) structure, the diffusion coefficient scales with the process to maintain a normal-type distribution, and the jump term remains proportional to the current state. We assume that the jumps, which correspond to block trades by external agents, follow an exponential distribution, meaning that the Lévy measure $\nu$ takes the form
$$
\nu(\ud z) = \eta e^{-\eta z} \mathbbm{1}_{\{z > 0\}} \ud z,\quad \forall z\in\RR,
$$
where $\eta > 0$ controls the decay rate of large jumps. This choice is in line with the existing literature on the subject, notably \textcite{Pomponio2011}. We propose in this case a price impact model in which the trading size influences the price in a concave manner based on the work of \textcite{BOUCHAUD200957}.
 The shape functions are given by the power law family with $$\ud F_i(x)=\frac{\kappa}{(x+1)^{\gamma_i}}\ud x, \quad\forall (i,x)\in\I_m\times\RR_+.$$ In other words, for all $(i,x,y)\in\I_m\times\RR_+^2$, the functions $F_i$ and $\psi_i$ are given by
\begin{equation*}
\resizebox{\textwidth}{!}{$
\begin{aligned}
F_i(x)=\left\{\begin{array}{ll}
\kappa x, & \text { if } \gamma_i=0, \\\kappa \log (x+1), & \text { if } \gamma_i=1, \\
\frac{\kappa}{1-\gamma_i}\big[(x+1)^{1-\gamma_i}-1\big], & \text { else. }
\end{array} ~\text{and}~~ \psi_i(y)= \begin{cases}
\frac{y}{\kappa}, & \text { if } \gamma_i=0, \\e^{\frac{y}{\kappa}}-1, & \text { if } \gamma_i=1, \\
{\big[1+(1-\gamma_i)\frac{y}{\kappa}\big]^{\frac{1}{1-\gamma_i}}-1}, & \text { else.}\end{cases}\right.
\end{aligned}
$}
\end{equation*}
 We use the following default parameters for all numerical results unless stated otherwise. These parameters reflect a square root impact.
\begin{table}[htbp]
\small
    \centering
    \begin{tabular}{lcccccccccc}
        \toprule
        \textbf{Parameter} & \(c\) & \(d\) & \(e\)  & \(\eta\) & \(\kappa\) & \(\gamma_0\)& \(\lambda\) & \(\midbar X\)& \(\midbar Y\) & \(T\)\\
        \midrule
        \textbf{Value} & 0.5 & 0.1 & 0.2 & 1.0 & 0.8 & -1 & 0.5 & 4 & 5
 & 4.0 \\
        \bottomrule
    \end{tabular}
    \caption{Default parameter values.}
    \label{tab:parameters}
\end{table}
\subsection{Single Regime Case}
To isolate the effects of stochastic resilience and liquidity, we first examine the single-regime case, where the Markov chain $I$ remains fixed in a single state. This is equivalent to setting the transition rate matrix $Q$ to zero, making each state absorbing. By doing so, we can analyze the impact of stochastic \textit{volume effect} independently from the influence of changing liquidity regimes. 

First, we examine how the value function responds to changes in various model parameters. The plots in Figure \ref{fig:value_function_variation} illustrate the impact of drift $h$, jump intensity $\lambda$, and volatility $\sigma$ on execution costs, where a lower volume effect corresponds to higher liquidity and improved execution. In particular, a higher volatility indicate greater market activity. The left hand side of Figure \ref{fig:value_function_variation} shows that a higher jump intensity leads to less favorable execution opportunities, thereby raising the overall execution cost. Likewise, increased volatility in the volume process reflects a dynamic market, helping to reduce execution costs. Figure \ref{fig:value_function_variation} also shows that higher resilience, which reflects a higher mean reversion speed, improves execution by accelerating liquidity replenishment. 
\begin{figure}[H]
    \centering
    \begin{subfigure}[b]{0.31\textwidth}
        \centering
        \includegraphics[width=\textwidth]{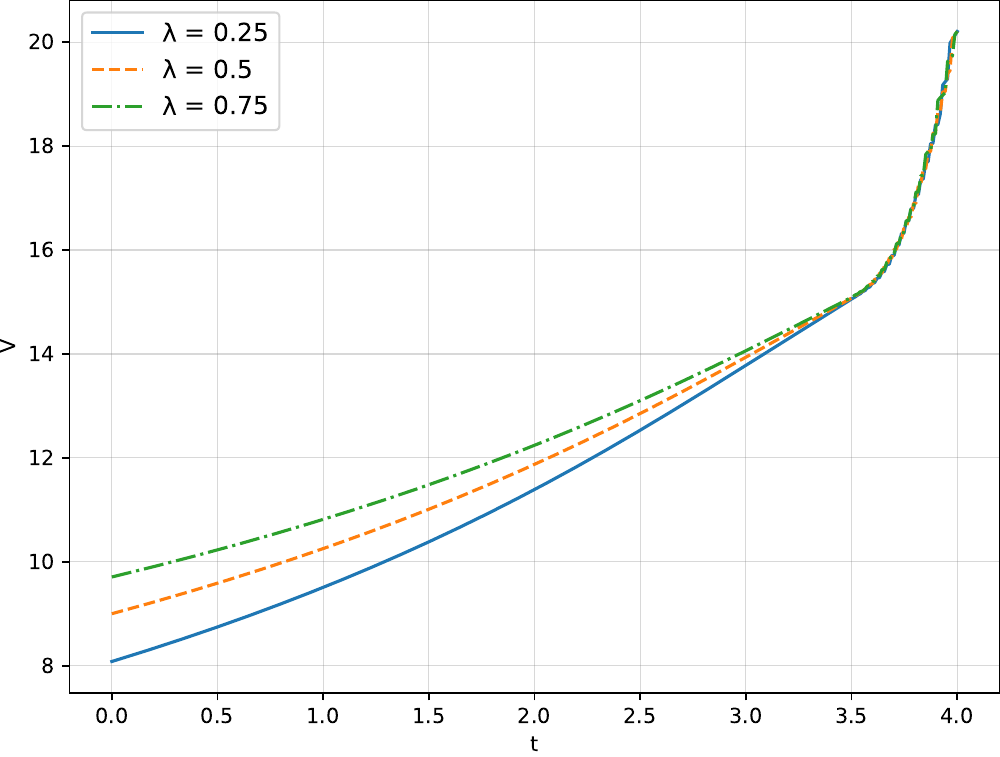}
    \end{subfigure}
    \hfill
    \begin{subfigure}[b]{0.31\textwidth}
        \centering
        \includegraphics[width=\textwidth]{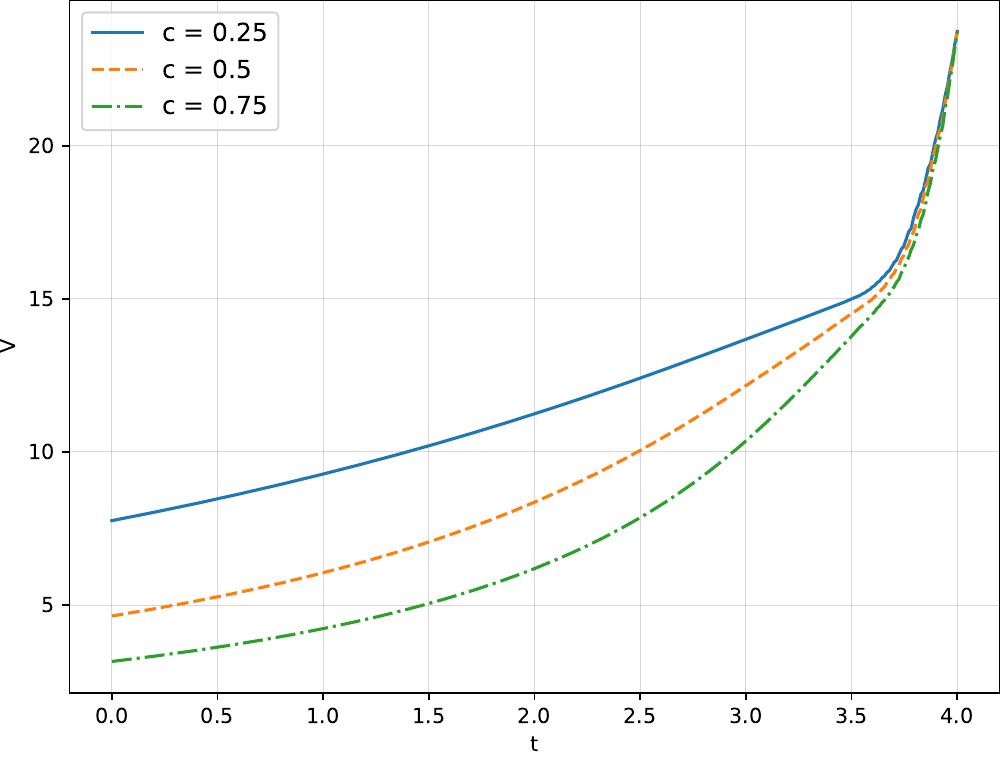}
    \end{subfigure}
    \hfill
    \begin{subfigure}[b]{0.316\textwidth}
        \centering
        \includegraphics[width=\textwidth]{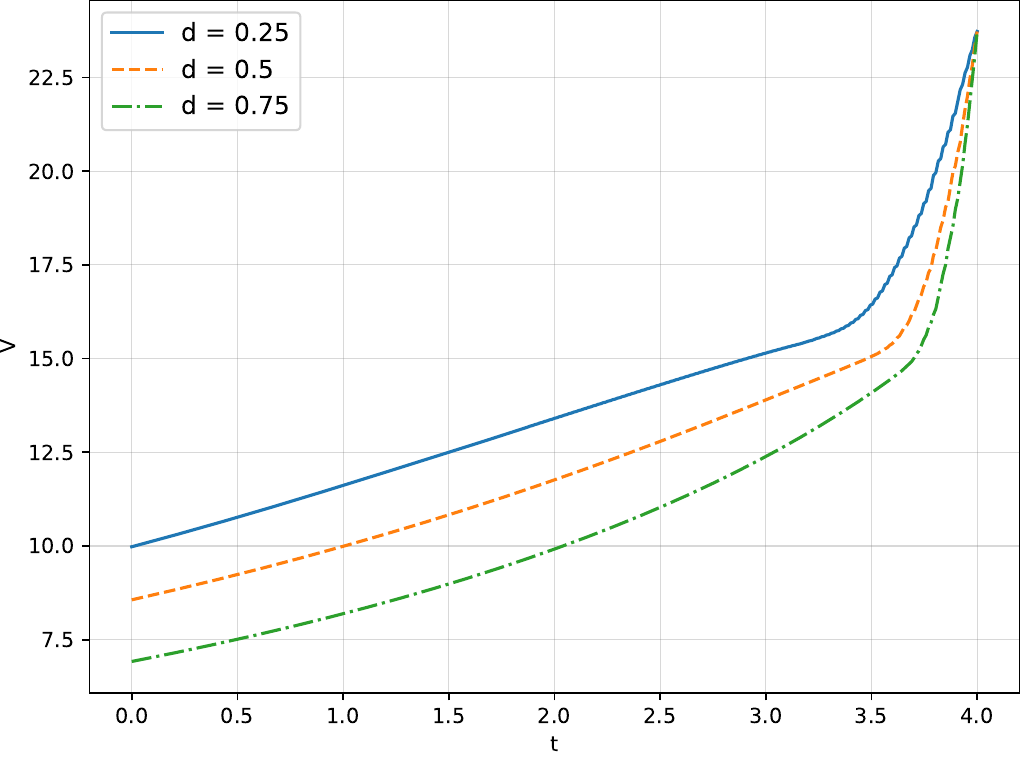}
    \end{subfigure}
    
    \caption{Variation of the value function over time under different market conditions with jump intensity on the left, resilience in the middle, and volatility on the right.}
    \label{fig:value_function_variation}
\end{figure}
  
Next, we illustrate the optimal execution strategy under the square root price impact model. 
\begin{figure}[H]
        \centering
\includegraphics[width=0.38\textwidth]{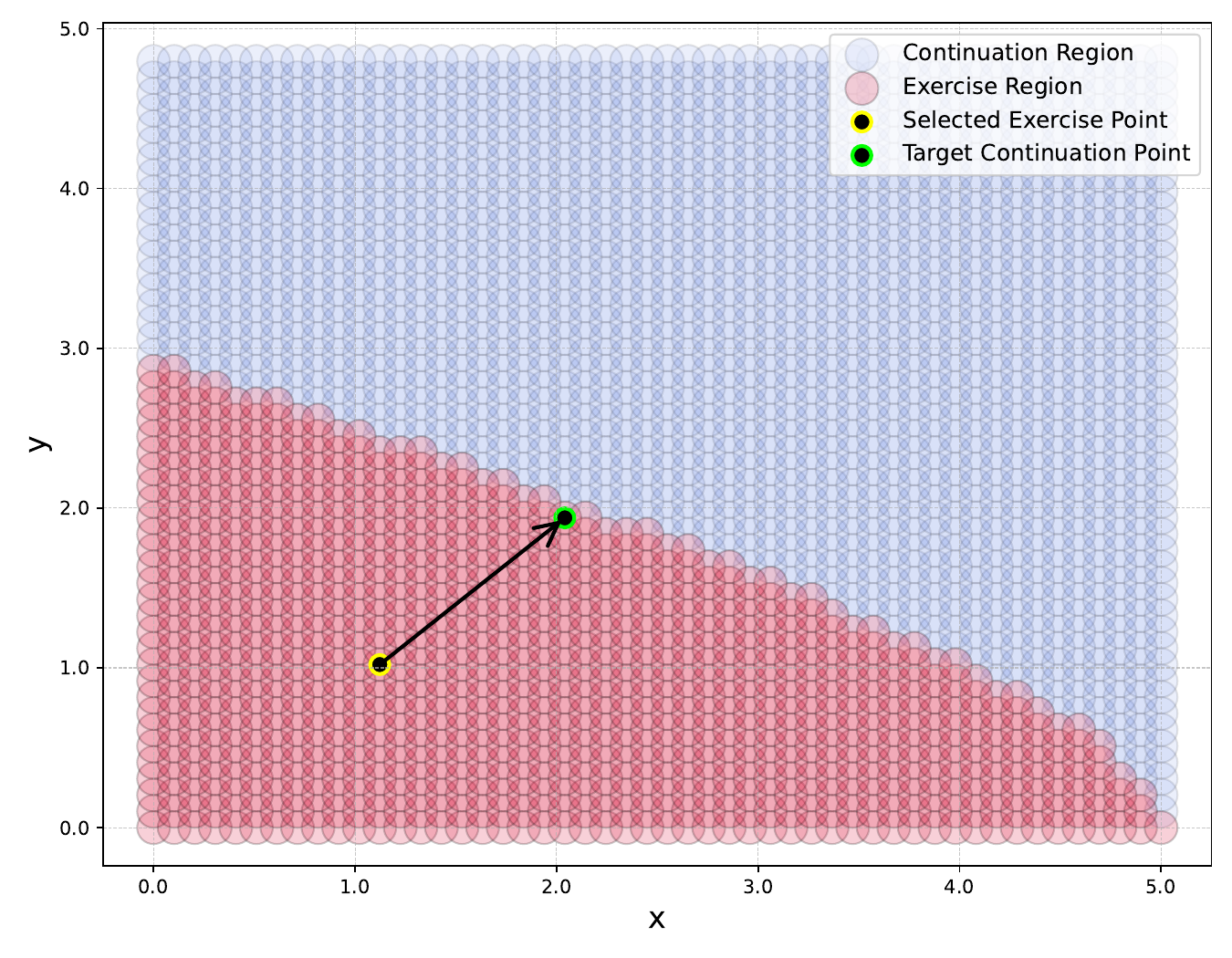}
       \caption{Illustration of continuation (blue) and exercise (red) regions, plotted against purchased quantity (x-axis) and volume effect (y-axis) for $\midbar X = 5$. The arrow shows an impulse shifting a state from exercise to continuation along $y=x$.}
    \label{fig:exercice_continuation}
\end{figure}
As the x-axis in Figure \ref{fig:exercice_continuation} represents the quantity already acquired, the continuation region expands with increasing 
$x$, indicating that larger holdings reduce the urgency to buy more immediately. The curved boundary between the regions reflects the trade-off between acquiring now or waiting, as price impact increases with trade size but at a decreasing rate. Upon reaching the exercise region, the optimal strategy is to purchase the smallest possible amount to reach the continuation region, thus reducing market impact in an optimal manner.

Subsequently, we verify in Figure \ref{fig:exercise_continuation_variation} the validity of the previous results by examining how variations affect the exercise and continuation regions. 
\begin{figure}[H]
    \centering
    \begin{subfigure}[b]{0.31\textwidth}
        \centering
        \includegraphics[width=\textwidth]{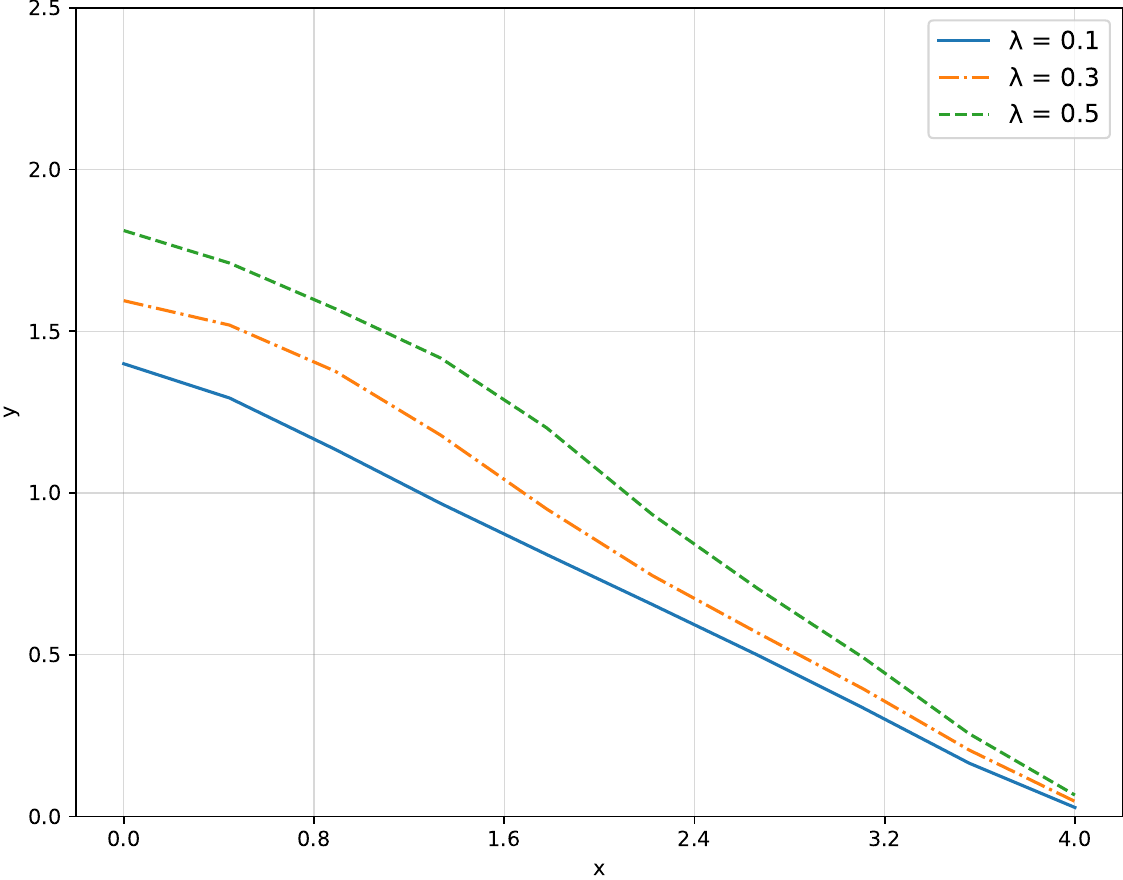}
    \end{subfigure}
    \hfill
    \begin{subfigure}[b]{0.31\textwidth}
        \centering
        \includegraphics[width=\textwidth]{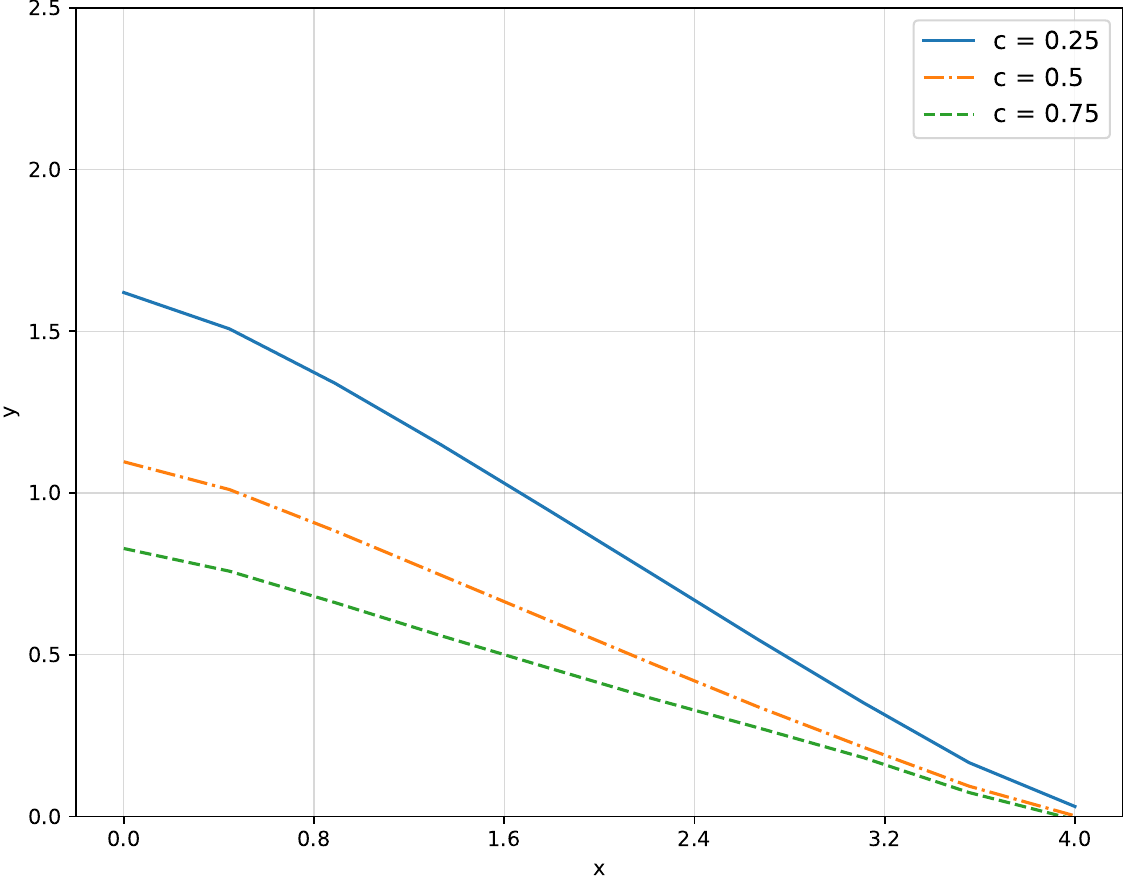}
        
    \end{subfigure}
    \hfill
    \begin{subfigure}[b]{0.31\textwidth}
        \centering
        \includegraphics[width=\textwidth]{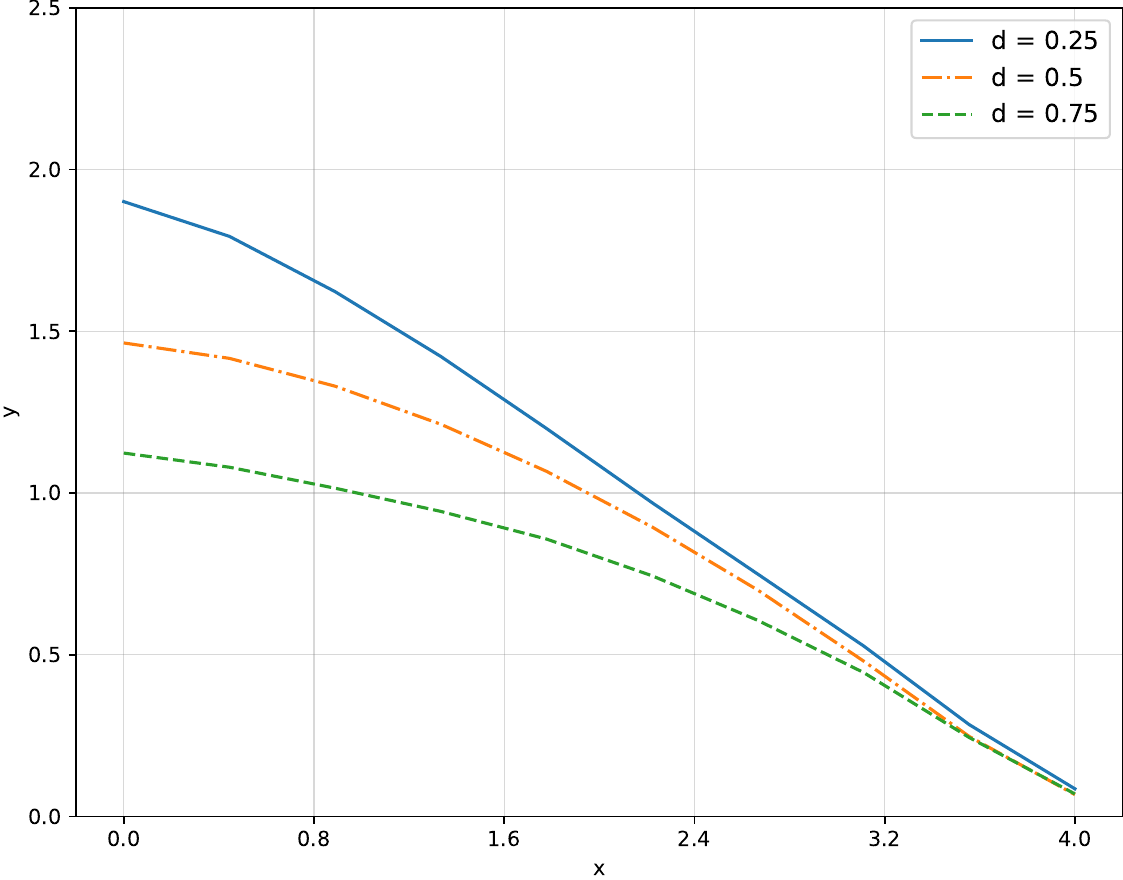}
    \end{subfigure}
    
    \caption{Variation of the exercise boundary $x\mapsto y^*(T/2,x)$ under different market conditions with jump intensity on the left, resilience in the middle, and volatility on the right.}
    \label{fig:exercise_continuation_variation}
\end{figure}
The exercise region expands in response to jumps. Indeed, when jump intensity is high, the optimal strategy involves anticipating increased future price impact by executing larger trades earlier. Conversely, when the volatility of the process governing the volume effect is high, the optimal policy is to acquire fewer shares. This reflects the expectation of more favorable market conditions ahead. The same holds when market resilience increases. This outcome aligns with the previous findings in Figure \ref{fig:value_function_variation} and those of \textcite{general_shape}.

\subsection{Regime-Switching case}
Next, we extend the analysis to the multi-regime case, where both sources of uncertainty interact. We assume a two-state homogeneous Markov chain, with the transition rate matrix $Q(t)$ at time $t \in [0,T]$ given by
$$
Q(t) = \begin{pmatrix}
-q_1 & q_1 \\
q_2 & -q_2
\end{pmatrix},$$
where $q_1$ and $q_2$ are positive constants. We begin by fixing $q_1 = q_2 = 0.2$ and $\gamma_1 = 0$. With these parameters set, we will then vary the price impact parameter $\gamma_2$ to analyze its effect on the continuation and exercise regions of the value functions $v_1$ and $v_2$.
\begin{figure}[H]
    \centering
    \begin{subfigure}[b]{0.33\textwidth}
        \centering
        \includegraphics[width=\textwidth]{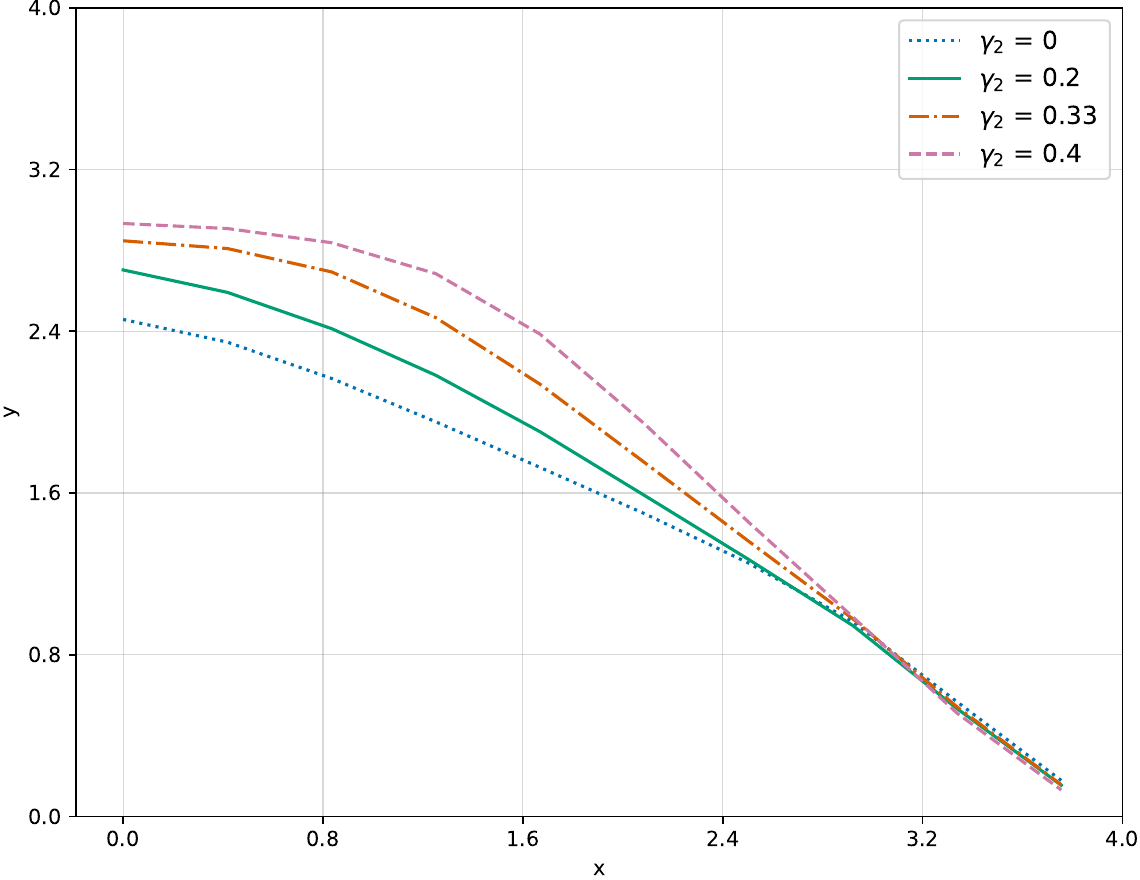}
    \end{subfigure}\quad\quad
    \begin{subfigure}[b]{0.33\textwidth}
        \centering
        \includegraphics[width=\textwidth]{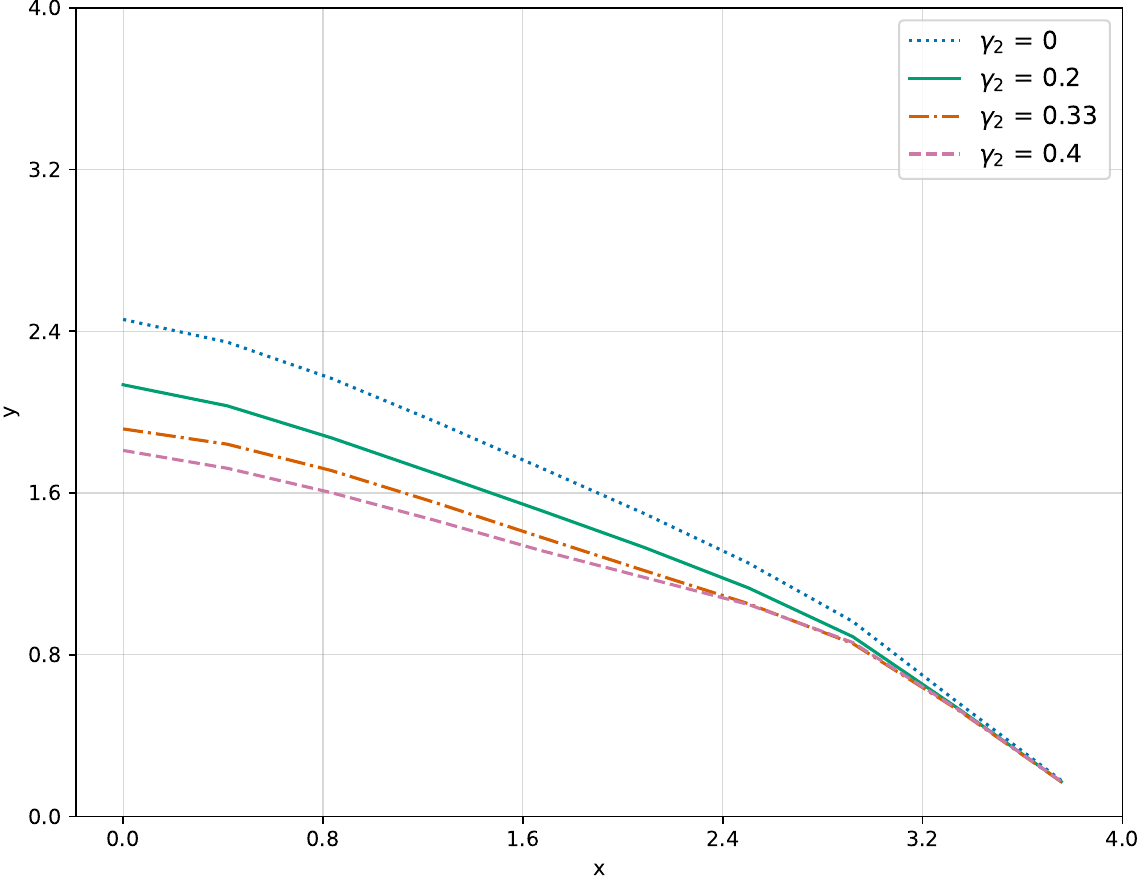}
        
    \end{subfigure}    
    \caption{Variation of the exercise boundary $x\mapsto y^*(T/2,x)$ of $v_1$ on the left and $v_2$ on the right under different price impacts.}
    \label{fig:execise_continuation_vary_gamma}
\end{figure}
Figure~\ref{fig:execise_continuation_vary_gamma} illustrates that when $\gamma_1 < \gamma_2$, the trader anticipates costlier execution in regime $2$ and adjusts their strategy accordingly. Indeed, in regime is $1$, where the price impact is lower than in regime $2$, an increase in the parameter $\gamma_2$ leads to an expansion of the exercise region (in regime $1$), as illustrated on the left-hand side of Figure \ref{fig:execise_continuation_vary_gamma}. Conversely, the right-hand side of Figure \ref{fig:execise_continuation_vary_gamma} shows that if we are in regime $2$, where the price impact is higher, the exercise region contracts as $\gamma_2$ increases. Consequently, the trader increases order sizes in regime $1$ to preempt higher execution costs associated with a potential switch to regime $2$, while in regime $2$, the trader decreases order sizes in anticipation of a favorable transition to regime $1$, where trading conditions are more advantageous.

Next, we study the case where the price impact parameter is set to $\gamma_1 = 0$ in regime $1$ and $\gamma_2 = 0.33$ in regime $2$. We consider $q_1 = q_2 = q_0$ and vary the values of $q_0$. This setup corresponds to a scenario where there is an equal probability of being in either regime, but the frequency of transitions between regimes changes. The results of this analysis are presented in Figure \ref{fig:execise_continuation_vary_switching_intensity}.
\begin{figure}[H]
    \centering
    \begin{subfigure}[b]{0.33\textwidth}
        \centering
        \includegraphics[width=\textwidth]{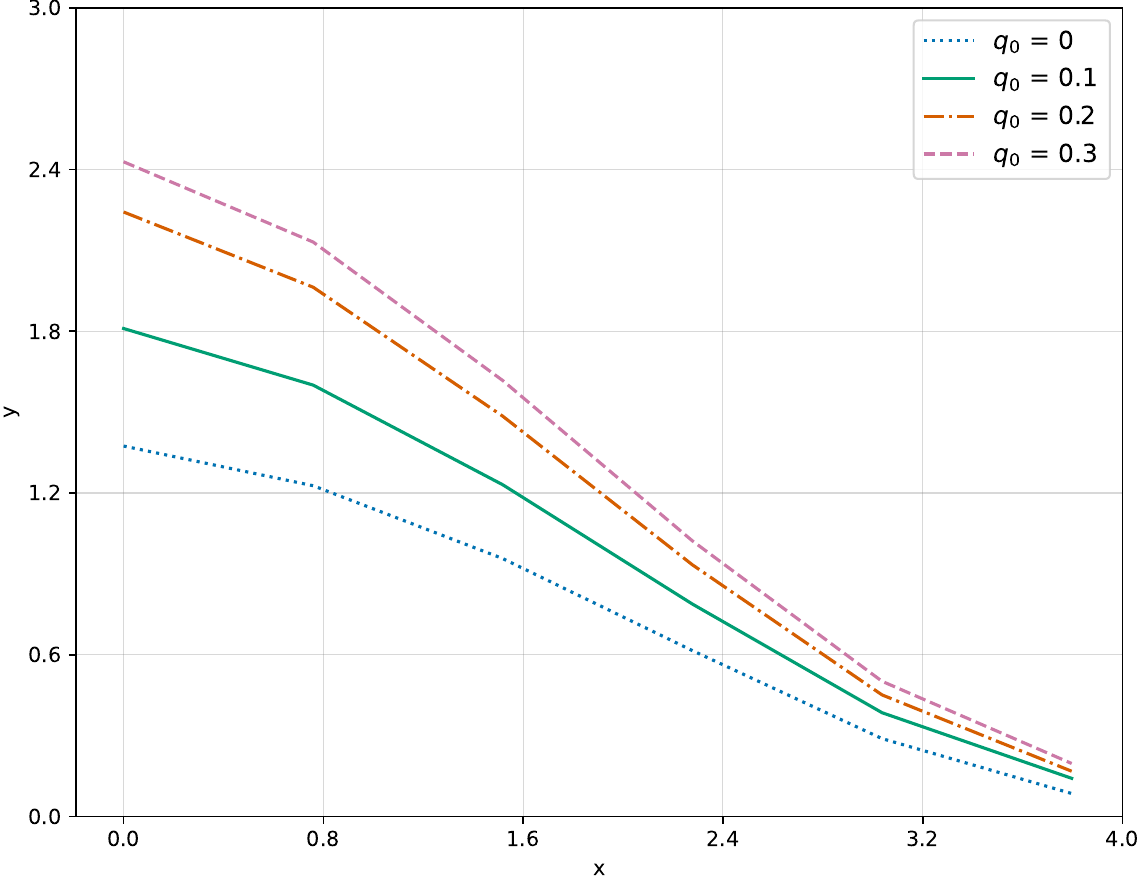}
    \end{subfigure}\quad\quad
    \begin{subfigure}[b]{0.33\textwidth}
        \centering
        \includegraphics[width=\textwidth]{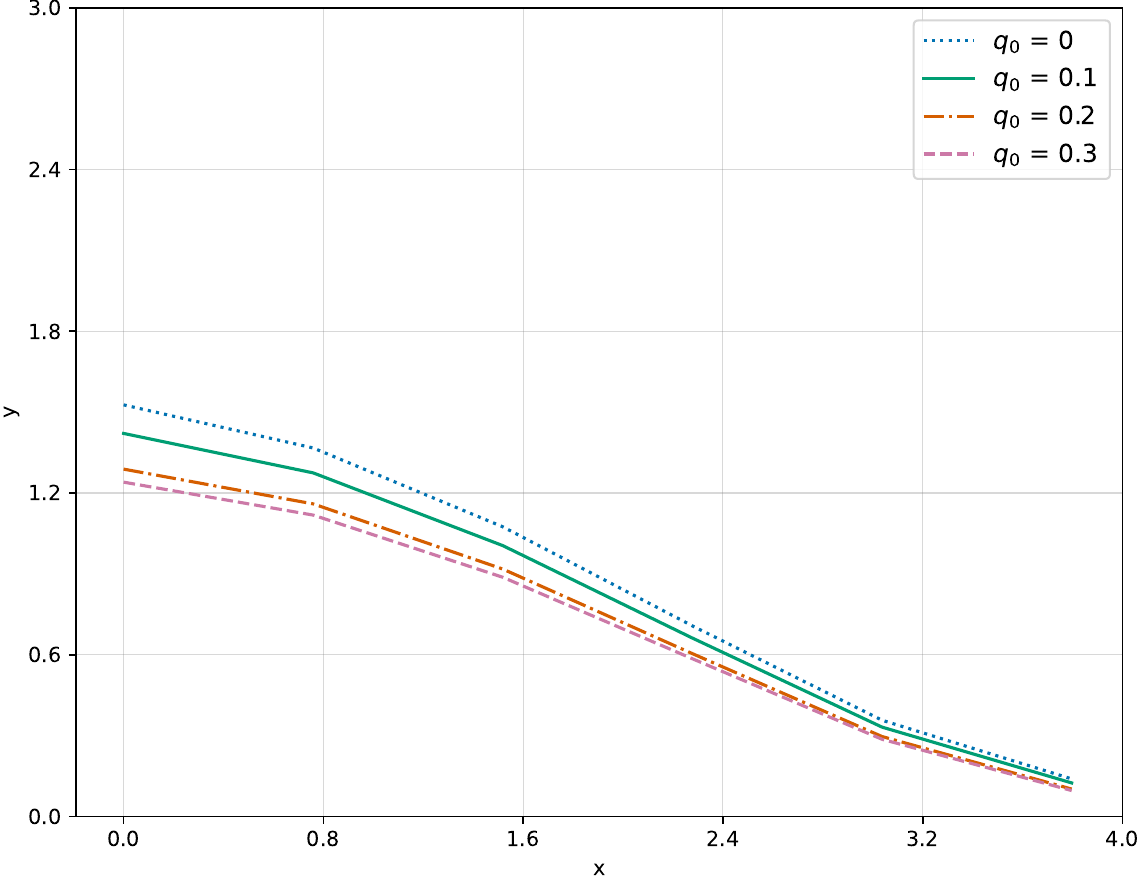}
        
    \end{subfigure}    
    \caption{Variation of the exercise boundary $x\mapsto y^*(T/2,x)$ of $v_1$ on the left and $v_2$ on the right under different regime switching intensities.}
    \label{fig:execise_continuation_vary_switching_intensity}
\end{figure}
 A higher switching intensity expands the exercise region in the low price impact regime (see the left-hand side of Figure \ref{fig:execise_continuation_vary_switching_intensity}) and contracts it in the high impact regime (see the right-hand side of Figure \ref{fig:execise_continuation_vary_switching_intensity}). In regime $1$, the increased likelihood of a regime shift prompts the trader to execute larger order sizes in anticipation of deteriorating execution conditions in the former. In the latter (regime $2$), the same anticipation leads to smaller order sizes, as the trader expects a transition to a more favorable regime. This behavior is consistent with the earlier observation that traders adjust their execution strategy based on the expected direction of the regime shift and its impact on future trading costs.

Finally, we analyze the case where the price impact parameter is set to $\gamma_1 = 0$ in regime $1$ and $\gamma_2 = 0.33$ in regime $2$. We fix $q_1 = 0.2$ and vary $q_2$, allowing us to study how the asymmetry in regime-switching probabilities affects execution. 
\begin{figure}[H]
    \centering
    \begin{subfigure}[b]{0.33\textwidth}
        \centering
        \includegraphics[width=\textwidth]{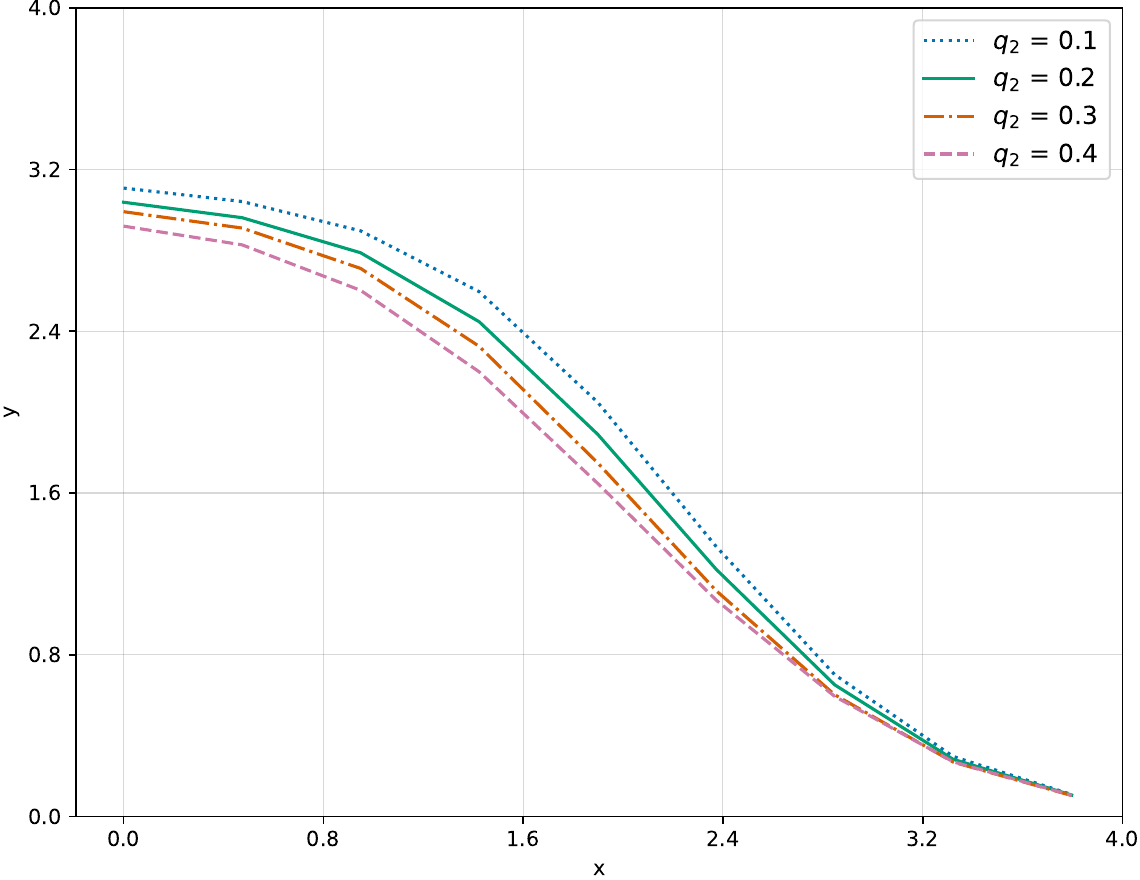}
    \end{subfigure}\quad\quad
    \begin{subfigure}[b]{0.33\textwidth}
        \centering
        \includegraphics[width=\textwidth]{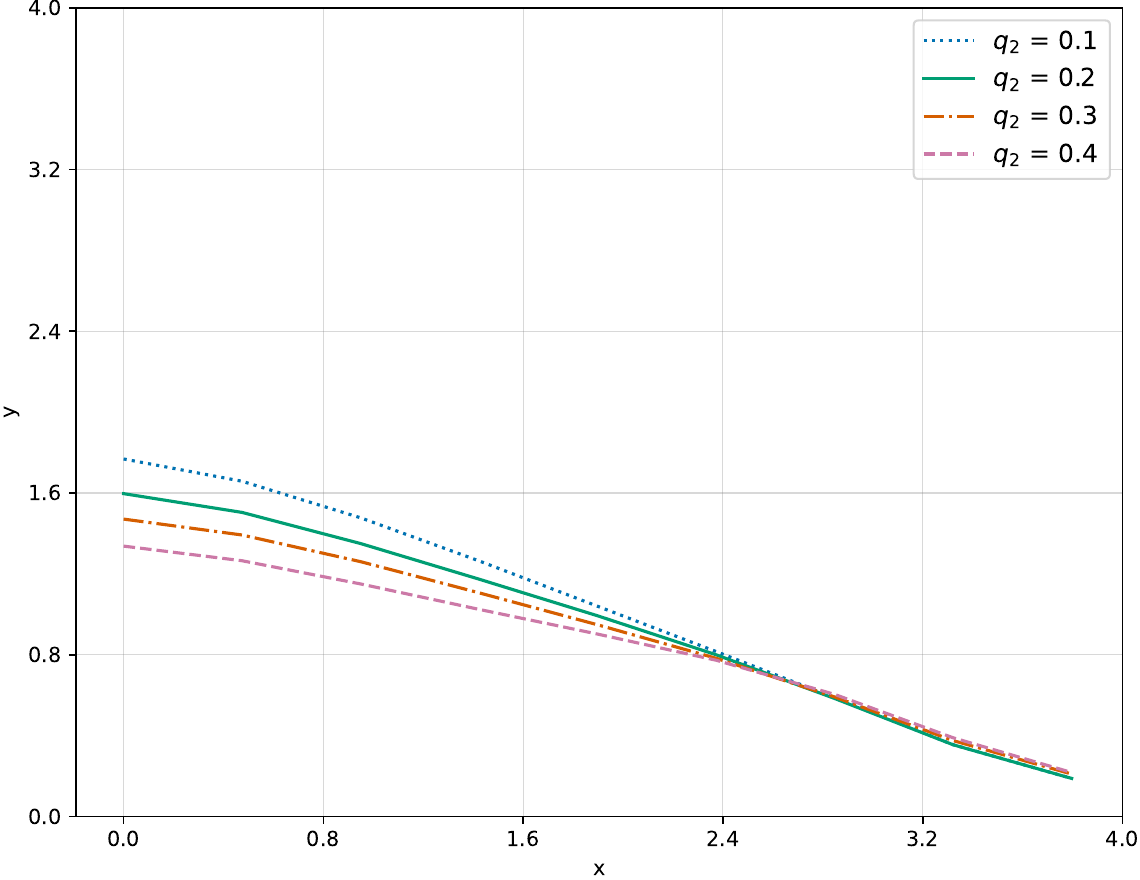}
        
    \end{subfigure}    
    \caption{Variation of the exercise boundary $x\mapsto y^*(T/2,x)$ of $v_1$ on the left and $v_2$ on the right under different asymmetric regime switching intensities.}
\label{fig:execise_continuation_vary_asymmetric_switching_intensity}
\end{figure}
Figure \ref{fig:execise_continuation_vary_asymmetric_switching_intensity} confirms our previous findings and highlights the effect of asymmetry in the switching dynamics. The exercise region expands in the low-impact regime when the transition rate matrix implies a higher long-term probability of remaining in the high-impact regime $2$. Conversely, as regime $2$ becomes more persistent, the exercise region contracts in regime~$2$ and expands in regime $1$, reflecting adjustments in order sizes based on anticipated execution costs. 

{\raggedright 
\printbibliography}

\appendix

\section{Proofs of the Results in Section \ref{Characterization of the Value Function}}
\label{appendix_caracterization}
\subsection{Continuity Modulus}
\begin{proof}[Proof of Lemma \ref{unifcont_phi}]
Let $0\leq t\leq s\leq T$, $x\in[0,\midbar X] $, $y\geq 0$ and $X\in\mathcal{A}_t(x)$. Suppose $\xi$ is a nonnegative random variable with finite moments. We have 
\begin{equation*}
    \begin{split}
\Delta & :=  
\E\Big[\Phi_{I_s}\left(Y^{t,y,X}_s+\xi\right)-\Phi_{I_s}\left(Y^{t,y,X}_s\right)\Big]  =   \E\Big[\int_0^\xi\psi_{I_s}\left(Y^{t,y,X}_s+\zeta\right)\ud\, \zeta\Big].
\end{split}
\end{equation*}
Note that under Assumption \ref{carnetinfini}, we have that $\psi_{I_s}(y)\leq b^{\frac{1}{\beta}}y^{\frac{1}{\beta}},$ for all $y\geq F_{I_s}(a).$ As $y\mapsto\psi_{I_s}(y)$ is non-decreasing on $\RR_+$, we have
\begin{align*}
\Delta & :=  \E\Big[\int_0^{\xi}\psi_{I_s}\left(Y^{t,y,X}_s+\zeta\right)\ud\, \zeta \Big]\\
 & \leq \E\Big[\int_0^{\xi}b^{\frac{1}{\beta}}\left(F_{I_s}(a)+\zeta\right)^{\frac{1}{\beta}}\1_{\{Y^{t,y,X}_s\leq F_{I_s}(a)\}}\ud\, \zeta +\int_0^{\xi}b^{\frac{1}{\beta}}\left(Y^{t,y,X}_s+\zeta\right)^{\frac{1}{\beta}}\1_{\{Y^{t,y,X}_s> F_{I_s}(a)\}}\ud\, \zeta\Big]\\ & = \sum_{i=1}^m\frac{\beta}{1+\beta}b^{\frac{1}{\beta}}\E\Big[\1_{\{I_s = i\}}\left(F_i(a)\right)^{\frac{\beta+1}{\beta}}\1_{\{Y^{t,y,X}_s\leq F_i(a)\}}\Big(\big(1+\frac{\xi}{F_i(a)}\big)^{\frac{\beta+1}{\beta}}-1\Big)\Big]\\ 
  &\quad\quad +\sum_{i=1}^m\frac{\beta}{1+\beta}b^{\frac{1}{\beta}}\E\Big[\1_{\{I_s = i\}}\left(Y^{t,y,X}_s\right)^{\frac{\beta+1}{\beta}}\1_{\{Y^{t,y,X}_s> F_i(a)\}}\Big(\big(1+\frac{\xi}{Y^{t,y,X}_s}\big)^{\frac{\beta+1}{\beta}}-1\Big)\Big]\\
   & \leq  \sum_{i=1}^m \frac{\beta}{1+\beta}b^{\frac{1}{\beta}}\E\Big[\Big(F_i(a)^{\frac{\beta+1}{\beta}}+\underset{0\leq u\leq T}{\sup}|Y^{t,y,X}_u|^{\frac{\beta+1}{\beta}}\Big)\Big(\big(1+\frac{\xi}{F_i(a)}\big)^{\frac{\beta+1}{\beta}}-1\Big)\Big].
\end{align*}
Using Cauchy-Schwarz inequality,
\begin{equation*}
\begin{split}
\Delta & \leq  \sum_{i=1}^m \frac{\beta}{1+\beta}b^{\frac{1}{\beta}}\E\Big[\Big(F_i(a)^{\frac{\beta+1}{\beta}}+\underset{0\leq u\leq T}{\sup}|Y^{t,y,X}_u|^{\frac{\beta+1}{\beta}}\Big)^2\Big]^{\frac{1}{2}}\E\Big[\Big(\big(1+\frac{\xi}{F_i(a)}\big)^{\frac{\beta+1}{\beta}}-1\Big)^2\Big]^{\frac{1}{2}}.
\end{split}
\end{equation*}
Finally, based on Proposition \ref{existence_uniqueness_SDE}, we get that
\begin{equation*}
\begin{split}
\Delta &\leq  \sum_{i=1}^m C^y_i\E\Big[\Big(\big(1+\frac{\xi}{F_i(a)}\big)^{\frac{\beta+1}{\beta}}-1\Big)^2\Big]^{\frac{1}{2}}\\&=\rho_y(\xi),
\end{split}
\end{equation*}
where $C^y_i:= \frac{2\beta}{1+\beta}b^{\frac{1}{\beta}}\left(F_i(a)^{\frac{2(\beta+1)}{\beta}}+ C_T\big(1+|y|^{\frac{2(\beta+1)}{\beta}}\big)\right)^{\frac{1}{2}}$. This completes the proof.
\end{proof}
\section{Proofs of the Results in Section \ref{viscosity_section}}
\label{appendix_viscosity}
\subsection{Viscosity Subsolution Property}
\begin{proof}[Proof of Theorem \ref{viscosity_sub}]
\label{proof_subsolution}
     Let $z_0:=(t_0,x_0,y_0) \in\mathcal{S}$ and $i\in \I_m$. Consider a test function $\varphi \in
C^{1,2}(\mathcal{S})$ such that $z_0$ achieves a local maximum of $v_i-\varphi$ and $v_i(z_0) = \varphi(z_0)$. In other words, there exists $r_0>0$ such that 
$$v_i(z) \leq \varphi(z), \quad \forall z \in \midbar B_{r_0}(z_0).
$$
Define, for all $(j,z) \in \I_m\times\cS$,
\begin{equation*}
\Psi(z,j)=\left\{\begin{array}{ll}
\varphi(z), & \text { if } j=i, \\
v_j(z), & \text { if } j \neq i.
\end{array} \right.\end{equation*}
Observe that $v_j(t,x,y) \leq \Psi(z,j)$ holds within the neighborhood $\midbar B_{r_0}(z_0)$. Define $Z^{z_0,X}_u$ as the tuple $\big(u, X_u, Y^{t_0,y_0,X}_u\big)$ for all $u \in [0,T]$. Next, define $\tau_{r_0}$ as the first exit time of $Z^{z_0,X}$ from $\midbar B_{r_0}(z_0)$, i.e.,  
$$\tau_{r_0} = \inf \big\{ u \geq t_0 : Z^{z_0,X}_u \notin \midbar B_{r_0}(z_0) \big\}.$$ Similarly, let $\tau^I_1$ denote the first jump time of the $I$ after time $t_0$, given by  
$$\tau^I_1 = \inf \big\{ u \geq t_0 : I_u \neq i \big\}.$$ 
Furthermore, we introduce the $\cF$-stopping time $\tau_0$ such that $\tau_0 \in\big]t_0, (\tau^I_1\wedge\tau_{r_0}) \wedge T\big[$. Note that $\tau^I_1>t_0$ and that $I_{u} = i$, for all $u\in[t_0,\tau^I_1[$. Consider an admissible control $X \in \mathcal{A}_{t_0}(x_0)$ in which the agent buys $0 \leq \eta < (\bar{X} - x_0)\wedge r_0$ assets at time $t_0$ and then remains inactive until time $(t_0 + \varepsilon) \wedge \tau_0$, where $\varepsilon>0$. Using the DPP introduced in Proposition \ref{prog_dyn}, we have
\begin{equation} \begin{split}
  \Psi(z_0,i) &= \varphi(z_0)\\ &\leq
  \E
  \Big[\Phi_{I_{(t_0+\eps)\wedge \tau_0}}(y_0+\eta)-\Phi_{I_{(t_0+\eps)\wedge \tau_0}}(y_0)+ v((t_0+\eps)\wedge \tau_0, x_0+\eta,{Y}^{t_0,y_0,X}_{(t_0+\eps)\wedge \tau_0},I_{(t_0+\eps)\wedge \tau_0}) \Big]\\
  &\leq \E
\Big[\Phi_{i}(y_0+\eta)-\Phi_{i}(y_0)+ \Psi((t_0+\eps)\wedge \tau_0, x_0+\eta, {Y}^{t_0,y_0,X}_{(t_0+\eps)\wedge \tau_0}, i) \Big].
\label{proofvis2}
\end{split} \end{equation}
 Applying It\^o's formula and taking the expectation over the interval 
$[t_0, (t_0+\varepsilon) \wedge \tau_0]$, while noting that 
$\mathbb{E} \big[\Psi(Z^{z_0,X}_{t_0}, I_{t_0})\big] = \mathbb{E} \big[\varphi(Z^{z_0,X}_{t_0})\big]$, we obtain
 \begin{equation} \begin{split}
  \E \Big[\Psi(Z^{z_0,X}_{(t_0+\eps)\wedge \tau_0}, I_{(t_0+\eps)\wedge \tau_0})\Big] &=
  \varphi(t_0,x_0+\eta,y_0+\eta) \\
  &\quad + \E \Big[\int_{t_0}^{(t_0+\eps)\wedge \tau_0} \big(
  \frac{\partial \varphi}{\partial t}+\Lc
  \varphi\big)(u,x_0+\eta,Y_{u}^{t_0,y_0,X}) \ud u \Big]\\
  &\quad + \E \Big[\int_{t_0}^{(t_0+\eps)\wedge \tau_0} \sum_{j\neq i}Q_{ij}(u)\big(v_j-\varphi\big)(u, x_0+\eta,Y_{u}^{t_0,y_0,X}) \ud u \Big].
  \label{proofvis3}
\end{split} \end{equation}
In the last equality, the martingale expectation vanishes since 
$\varphi(Z^{z_0,X}_u)$ and $D\varphi(Z^{z_0,X}_u)$ are bounded, given that $Z^{z_0,X}_u \in \midbar B_{r_0}(z_0)$ for all $u \in [t_0, (t_0+\varepsilon) \wedge \tau_0]$. For all $u \in [t_0, (t_0+\varepsilon) \wedge \tau_0]$, since $Z^{z_0,X}_u$ remains in $\midbar B_{r_0}(z_0)$, we obtain that 
$v_i(Z^{z_0,X}_u) \leq \varphi(Z^{z_0,X}_u)$. Combining this with relations \eqref{proofvis2} and
\eqref{proofvis3} yields
\begin{equation} \begin{split}
  \E \Big[\int_{t_0}^{(t_0+\eps)\wedge \tau_0}
  \big(\frac{\partial \varphi}{\partial t}+\Lc
  \varphi\big)(u,x_0+&\eta,Y_{u^-}^{t_0,y_0,X})+\sum_{j\neq i}Q_{ij}(u)\big(v_j-\varphi\big)(u,x_0+\eta,Y_{u^-}^{t_0,y_0,X}) \ud u \Big] \\& \geq
  \varphi(z_0)-\varphi(t_0,x_0+\eta,y_0+\eta) - \Phi_i(y_0+\eta)+\Phi_i(y_0).  \label{proofvis4}
\end{split} \end{equation}
Assume that the agent decides to do nothing, meaning $\eta = 0$. We know from Theorem \ref{continuous_v} that $v_j$ is continuous on $\midbar \cS$, for all $j\in \I_m$. Additionally, $Q$ is continuous on $\R_+$. Dividing inequality \eqref{proofvis4} by $\eps$ and letting $\eps$ going to $0$ and recalling that $\varphi(z_0) = v_i(z_0)$, we get that
\begin{equation*}
 \big(\frac{\partial \varphi}{\partial t}+\Lc \varphi+\sum_{j\neq i}Q_{ij}(v_j-v_i)\big) (z_0)\geq 0.
\end{equation*}
Assume now that the agent executes a part $\eta > 0$ of the position at time $t_0$. By sending $\varepsilon$ to $0$, we get
$$
0\geq-\varphi(t_0,x_0+\eta,y_0+\eta)+\varphi(t_0,x_0,y_0)- \Phi_i(y_0+\eta)+\Phi_i(y_0).$$
It follows from equation \eqref{definitionPhi} that
\begin{equation*}
 \int_0^\eta \frac{\partial \varphi}{\partial x}(t_0,x_0+v,y_0+\eta)\ud v+\int_0^\eta\frac{\partial
\varphi}{\partial y}(t_0,x_0,y_0+v)\, \ud v+\int_{y_0}^{y_0+\eta}\psi_i(\xi)\,
\ud \xi\geq 0.
\end{equation*}
The function $\psi_i$ is left-continuous and non-decreasing on $\R_+$. Dividing by $\eta$ and letting $\eta
\rightarrow 0$, we obtain
\begin{equation*}
\big(\frac{\partial \varphi}{\partial x}+\frac{\partial \varphi}{\partial
y}+\psi_i\big)(z_0)\geq 0.
\end{equation*}
This proves the required subsolution property
\begin{equation*}
 \max\Big(-\big(\frac{\partial
  \varphi}{\partial t}+\mathcal{L}\varphi+\sum_{j\neq i}Q_{ij}(v_j-v_i)\big)(z_0) ,\ -\big(\frac{\partial
  \varphi}{\partial x}+\frac{\partial \varphi}{\partial y}+\psi_i\big)(z_0)\Big)\leq 0.
\end{equation*}
\end{proof}
\subsection{Viscosity Supersolution Property}
\begin{proof}[Proof of Theorem \ref{viscosity_super}] 
\label{proof_supersolution}
Suppose that $v$ does not satisfy the supersolution property. Then, there exists $i_0\in \I_m$, $z_0:=(t_0,x_0,y_0) \in\mathcal{S}$, a test function $\bar\varphi \in
C^{1,2}(\mathcal{S})$ such that $z_0$ achieves a local minimum of $v_{i_0}-\bar\varphi$ and $v_{i_0}(z_0) = \bar\varphi(z_0)$, and $\eta>0$ such that
\begin{equation*}
-\big(\frac{\partial \bar\varphi}{\partial t}+\mathcal{L} \bar\varphi+\sum_{j\neq {i_0}}(v_j-v_{i_0})Q_{i_0j}\big)(z_0) < -\eta,~~\text{and}~~
-\big(\frac{\partial \bar\varphi}{\partial x}+\frac{\partial \bar\varphi}{\partial
y}+\psi_{i_0}\big)(z_0)<-\eta.
\end{equation*}
Since $z_0$ achieves a local minimum of $v_{i_0}-\bar\varphi$, there exists $r_0>0$ such that 
$$v_{i_0}(z) \geq \bar\varphi(z), \quad \forall z \in\midbar B_{r_0}(z_0).
$$
By Theorem \ref{continuous_v}, the value functions $(v_j)_{j\in \I_m}$ are continuous on $\bar{\cS}$. Moreover, the test function $\bar\varphi$ belongs to $C^{1,2}(\mathcal{S})$, and each $\psi_{i_0}$ is non-decreasing and left-continuous. Consequently, there exists $r_1>0$ such that $t_0+r_1<T$ and for all $z:=(t,x,y)\in\midbar{B}_{r_1}(z_0)$, 
\begin{equation}
  \label{proofvis5} - \big(\frac{\partial \bar\varphi}{\partial t}+\Lc\bar\varphi+\sum_{j\neq i_0}(v_j-v_{i_0})Q_{i_0j}\big)(z) < -\eta,~~\text{and}~~
-\big(\frac{\partial \bar\varphi}{\partial x}+\frac{\partial \bar\varphi}{\partial
y}+\psi_{i_0}\big)(z) < -\eta,
\end{equation}
where
$\midbar B_{r_1}(z_0):= \{ z \in \cS : \|z - z_0\| \leq r_1 \}$. Define, for all $(j,z) \in \I_m\times \cS$,
\begin{equation*}
\Psi(z,j)=\left\{\begin{array}{ll}
\bar\varphi(z), & \text { if } j=i_0, \\
v_j(z), & \text { if } j \neq i_0.
\end{array} \right.\end{equation*}
Let $X$ and admissible strategy in $\cA_{t_0}(x_0)$ and $Z^{z,X}_u$ represent the tuple $\big(u, X_u, Y^{t,y,X}_u\big)$ for all $(u,z)\in[0,T]\times \cS$. Additionally, set the radius $\eps_0$ as $\eps_0:=r_0\wedge r_1$. Define $\tau_B$ as the first exit time of $Z^{z_0,X}$ from $\midbar{B}_{\eps_0}(z_0)$, i.e.,  
$$\tau_{B} = \inf \big\{u \geq t_0:
Z^{z,X}_u \notin \midbar{B}_{r_0}(z_0)\cap \midbar B_{r_1}(z_0)\big\}$$ for any admissible control $X\in\mathcal{A}_{t_0}(x_0)$. Let $\tau^I_1$ be the first jump time after time $t_0$ of the Markov chain $I$ such that $$\tau^I_1=\inf \big\{u \geq t_0: I_u \neq i_0\big\}.$$ Next, we introduce the $\cF$-stopping time $\tau_1$ such that $\tau_1:=\tau^I_1\wedge\tau_B$. Note that $\tau_B<T$, $t_0<\tau^I_1$ and that $I_{u} = i_0$, for all $u\in[0,\tau^I_1[$. Applying  It{\^o}'s formula between $t_0$ and
$\tau_1^-$, we obtain \begin{equation}
\begin{split}
\E\Big[\Psi(Z^{z_0,X}_{\tau_1^-},I_{\tau_1^-}) \Big] &=
\bar\varphi(z_0)
+ \E \Big[\int_0^{\tau_1^-}\big(\frac{\partial
  \bar\varphi}{\partial x}+\frac{\partial \bar\varphi}{\partial
  y}\big)(Z^{z_0,X}_{u^-})\, \ud X^c_u \Big] \\
    & \quad+ \E \Big[\int_{t_0}^{\tau_1^-} \big(\frac{\partial \bar\varphi}{\partial t} +\Lc
\bar\varphi+\sum_{j\neq i}(v_j-\bar\varphi)Q_{i_0j}\big) (Z^{z_0,X}_{u^-}) \ud u \Big]\\
  &\quad + \E \Big[ \sum_{t_0\leq u < \tau_1}
  \bar\varphi(u,X_u,Y_u^{t_0,y_0,X}) -\bar\varphi(u,X_{u^-},\check{Y}_{u^-}^{t_0,y_0,X})\Big]. \label{proofvis8}
  \end{split}
  \end{equation}  
In the last equality, the martingale expectation vanishes since 
$\varphi(Z^{z_0,X}_u)$ and $D\varphi(Z^{z_0,X}_u)$ are bounded, given that $Z^{z_0,X}_u \in \midbar B_{\eps_0}(z_0)$ for all $u \in [t_0,\tau_1[$. From \eqref{proofvis5} the fact that 
$\Delta X_u := X_u - X_{u^-} = Y_{u}^{t_0,y_0,X} - \check{Y}_{u^-}^{t_0,y_0,X}$ 
for all $t_0 \leq u < \tau_1$, we deduce that
\begin{equation}
\label{proofvis55}
\begin{split}
\bar\varphi(u,X_u,Y_u^{t_0,y_0,X}) -
\bar\varphi(u,X_{u^-},\check{Y}_{u^-}^{t_0,y_0,X}) & =  \int_0^{\Delta X_u} \big(\frac{\partial \bar\varphi}{\partial x}+\frac{\partial
\bar\varphi}{\partial y}\big)(u,X_{u^-}+v,\check{Y}_{u^-}^{t_0,y_0,X}+v)\, \ud v\\
& \geq \eta\Delta X_u-\int_0^{\Delta X_u}\psi_{i_0}(\check{Y}_{u^-}^{t_0,y_0,X}+\zeta)\,
\ud\zeta\\ & = \eta\Delta
X_u-\Phi_{i_0}(Y^{t_0,y_0,X}_u)+\Phi_{i_0}(\check{Y}_{u^-}^{t_0,y_0,X}).
    \end{split}
  \end{equation}
Note that $\E\Big[\Psi(Z^{z_0,X}_{\tau_1^-},I_{\tau_1^-}) \Big] = \E\Big[\bar\varphi(Z_{\tau_1^-}^{z_0,X}) \Big]$. Plugging \eqref{proofvis5} and \eqref{proofvis55} in equation
\eqref{proofvis8}, we obtain
\begin{equation*}
\begin{split}
\E\Big[\bar\varphi(Z_{\tau_1^-}^{z_0,X}) \Big]
& \geq  \bar\varphi(z_0)
  + \eta\E \Big[\tau_1-t_0 +  \int_{t_0}^{\tau_1^-}\ud X_u\Big] -\E\Big[\int_{t_0}^{\tau_1^-}\psi_{i_0}(\check{Y}_{u^-}^{t_0,y_0,X})\ud X^c_u\Big]\\ & \quad-\E \Big[\sum_{t_0\leq u < \tau_1}
  \Phi_{i_0}(Y_u^{t_0,y_0,X}) -
  \Phi_{i_0}(\check{Y}_{u^-}^{t_0,y_0,X}) \Big].
\end{split}
  \end{equation*}
Therefore, we obtain
\begin{equation}
\begin{split}
  v_{i_0}(z_0) = \bar\varphi(z_0) &\leq \E\Big[\int_{t_0}^{\tau_1^-} \psi_{i_0}(\check{Y}_{u^-}^{t_0,y_0,X}) \ud X^c_u
+\bar\varphi(Z_{\tau_1^-}^{z_0,X}) \Big]
 \\& \quad+\E\Big[ \sum_{t_0\leq u < \tau_1}
  \Phi_{i_0}(Y_u^{t_0,y_0,X}) -
  \Phi_{i_0}(\check{Y}_{u^-}^{t_0,y_0,X}) - \eta \big( \int_{t_0}^{\tau_1^-}\ud u + \int_{t_0}^{\tau_1^-}\ud X_u\big)\Big].
    \label{proofvis11}
 \end{split}
  \end{equation}
Define the set $$\cD^{\eps_0} := \Big\{\omega\in \Omega:\tau_B(\omega)<\tau^I_1(\omega)~\text{and}~\check{Y}_{\tau_1^-}^{t_0,y_0,X}(\omega)\in[y_0-\varepsilon_0,
y_0+\varepsilon_0] \Big\},$$ 
and the $\cF_{\tau_1}$-random variable,
$$Z^{(\gamma)}_{\tau_1} :=
\big(\tau_1,\check{Y}_{\tau_1^-}^{t_0,y_0,X}+\gamma\Delta
X_{\tau_1},X_{\tau_1^-}+\gamma\Delta
X_{\tau_1}\big).$$
We know that if $\tau_1=\tau_B$ and $\check{Y}_{\tau_1^-}^{t_0,y_0,X} \in [y_0-\varepsilon_0,y_0+\varepsilon_0]$, then $Z^{z_0,X}_{\tau_1}$ either lies on the boundary $\partial\bar{B}_{\eps_0}(z_0)$ or has exited $\bar{B}_{\eps_0}(z_0)$ at time $\tau_1$. Since $\check{Y}_{\tau_1^-}^{t_0,y_0,X}$ accounts for the jumps of $M$, only the jumps in the inventory process $X$ can cause $Z^{z_0,X}$ to cross the boundary at $\tau_1$. Consequently, there exists an $\cF_{\tau_1}$-measurable random variable $\gamma \in [0,1]$ such that
$$
\check{Y}_{\tau_1^-}^{t_0,y_0,X} +\gamma\Delta
X_{\tau_1}\in\big\{y_0-\varepsilon_0,y_0+\varepsilon_0\big\},~~\text{and}~~Z^{(\gamma)}_{\tau_1} \in\midbar{B}_{\eps_0}(z_0),$$
on the set $\cD^{\eps_0}$. It follows from the DPP (see Lemma \ref{prog_dyn}), that
$$
v_{i_0}({Z}^{(\gamma)}_{\tau_1}(\omega)) \leq v_{i_0}(Z^{z_0,X}_{\tau^-_1}(\omega))
+\Phi_{i_0}({Y}^{t_0,y_0,X}_{\tau_1}(\omega))-\Phi_{i_0}(\check{Y}^{t_0,y_0,X}_{\tau_1^-}(\omega)+\gamma \Delta
X_{\tau_1}(\omega)),\quad \forall \omega\in\cD^{\eps_0}.$$
Additionally, given that $Z^{(\gamma)}_{\tau_1}\in\midbar{B}_{\eps_0}(z_0)$, it follows that, for all $\omega\in\cD^{\epsilon_0}$,
\begin{equation*}
\resizebox{\textwidth}{!}{$
\begin{aligned}
\bar\varphi(Z^{z_0,X}_{\tau_1^-}(\omega)) & =
\bar\varphi(Z^{(\gamma)}_{\tau_1}(\omega))-\int_0^{\gamma(\omega)\Delta X_{\tau_1}(\omega)}\big(\frac{\partial
  \bar\varphi}{\partial x}+\frac{\partial \bar\varphi}{\partial
  y}\big)(t,X_{\tau_1^-}(\omega) + v ,\check{Y}_{\tau_1^-}^{t_0,y_0,X}(\omega)+v)\, \ud v\\
& \leq \bar\varphi(Z^{(\gamma)}_{\tau_1}(\omega))-\eta\gamma(\omega)\Delta
X_{\tau_1}(\omega)+\Phi_{i_0}(\check{Y}^{t_0,y_0,X}_{\tau_1^-}(\omega)+\gamma(\omega) \Delta
X_{\tau_1}(\omega))-\Phi_{i_0}(\check{Y}^{t_0,y_0,X}_{\tau_1^-}(\omega))\\
& \leq v_{i_0}(Z^{z_0,X}_{\tau_1^-}(\omega))-\eta\gamma(\omega)\Delta
X_{\tau_1}(\omega)+\Phi_{i_0}({Y}^{t_0,y_0,X}_{\tau_1}(\omega))-\Phi_{i_0}(\check{Y}^{t_0,y_0,X}_{\tau_1^-}(\omega)).
\end{aligned}
$}
\end{equation*}
Plugging the last inequality in \eqref{proofvis11}, we obtain
\begin{equation*} \begin{split}
v_{i_0}(z_0) & \leq \E \Big[ \sum_{t_0\leq u \leq \tau_1}
  \Phi_{i_0}(Y_u^{t_0,y_0,X}) -
  \Phi_{i_0}(Y_{u^-}^{t_0,y_0,X}) +\int_{t_0}^{\tau_1 } \psi_{i_0}(\check{Y}_{u^-}^{t_0,y_0,X})
\ud X^c_u +
   v_{i_0}(Z^{z_0,X}_{\tau_1}) \Big]  \\
& \quad - \eta \E \Big[ \int_{t_0}^{\tau_1 ^-}\ud u + \int_{t_0}^{\tau_1
^-}\ud X_u+\gamma\Delta X_{\tau_1}\1_{\{\check{Y}_{\tau_1^-}^{t_0,y_0,X}\in[ 
y_0-\varepsilon_0, y_0+\varepsilon_0] \}}\1_{\{\tau_B<\tau^I_1\}}\Big].
 \end{split} \end{equation*}
From the DPP, it follows that
$$
0 \leq -\inf_{X\in\mathcal{A}_{t_0}(x_0)}\E \Big[\tau_1 -t_0 + X_{\tau_1
^-} - x_0+\gamma\Delta
X_{\tau_1}\1_{\{\check{Y}_{\tau_1^-}^{t_0,y_0,X}\in[  y_0-\varepsilon_0,
y_0+\varepsilon_0] \}}\1_{\{\tau_B<\tau^I_1\}}\Big].$$
Using the sequential characterization of the supremum, we can find a sequence of admissible controls $X^n = (X^n)_{n\geq 0}$ such that
$$\E \Big[ \tau^n_1 -t_0+ X^n_{\tau^{n^-}_1}-x_0+\gamma\Delta
X^n_{\tau^n_1}\1_{\{\check{Y}_{{\tau_1^n}^-}^{t_0,y_0,X^n}\in[ 
y_0-\varepsilon, y_0+\varepsilon] \}}\1_{\{\tau^n_B<\tau^I_1\}}\Big]\leq \frac{1}{n},$$
where $\tau^n_B:=\inf \big\{t \geq t_0:
(t,X^n_t,Y_t^{t_0,y_0,X^n}) \notin \midbar{B}_{\eps_0}(z_0) \big\}$ and $\tau_1^n:=\tau_B^n\wedge \tau^I_1$. As $\tau^n_1-t_0\geq
0$ and $X^n_{\tau^{n^-}_1}-x_0\geq 0$, we have that
$$\max\Big\{\E\big[ \tau^n_1 -t_0\big],\ \E \big[X^n_{\tau^{n^-}_1}-x_0\big],\ \E \Big[\gamma\Delta
X^n_{\tau^n_1}\1_{\{\check{Y}_{{\tau_1^n}^-}^{t_0,y_0,X^n}\in[ 
y_0-\varepsilon, y_0+\varepsilon] \}}\1_{\{\tau^n_B<\tau^I_1\}}\Big]\Big\}\leq \frac{1}{n}.$$
Moreover, $\gamma\Delta X^n_{\tau^n_1} = y_0 + \eps_0 - \check{Y}_{{\tau_1^n}^-}^{t_0,y_0,X^n} \geq y_0 + \eps_0-y_0 = \eps_0$. Therefore,
\begin{equation}
\label{proofsuper1}
\begin{split}
\frac{1}{n}&\geq\eps_0\P\big(y_0-\varepsilon_0\leq \check{Y}_{{\tau_1^n}^-}^{t_0,y_0,X^n}\leq y_0+\eps_0\ ,\ \tau^n_B<\tau^I_1\big)\\&\geq \eps_0\Big(\P\big(y_0-\varepsilon_0\leq \check{Y}_{{\tau_1^n}^-}^{t_0,y_0,X^n}\leq y_0+\eps_0\big) +\P\big(\tau^n_B<\tau^I_1\big)-1\Big).
\end{split}
\end{equation}
On the other hand, we have
$$\lim_{n\to +\infty}\P(\tau^{n}_1=t_0) = 1,~~\textrm{and}~~\lim_{n\to
+\infty}\P(X^n_{\tau^{n^-}_1}=x_0) = 1.$$
Therefore, \begin{equation*}
\resizebox{\textwidth}{!}{$
\begin{aligned}\lim_{n\to +\infty}\P(\tau^{n}_B=t_0) = 1, ~~\lim_{n\to +\infty}\P\big(\check{Y}_{{\tau_1^n}^-}^{t_0,y_0,X^n}=y_0\big)=1~~ \text{and}~~
\lim_{n\to +\infty}\P\big(y_0-\varepsilon_0\leq\check{Y}_{{\tau_1^n}^-}^{t_0,y_0,X^n}\leq y_0+\eps_0\big)=1.\end{aligned}
$}
\end{equation*} Since $\lim_{n\to +\infty}\P\big(\tau_B^n<\tau^I_1\big) = \P\big(t_0<\tau^I_1\big) = \exp \big(-\int_{t_0}^{t_0} Q_{ii}(u) \, \ud u \big) = 1,$ we get from inequality \eqref{proofsuper1} when $n$ goes to $+\infty$ that
\begin{equation*}
\begin{split}
0 \geq \eps_0\Big(1+\lim_{n\rightarrow\+\infty}\P\big(\tau_B^n<\tau^I_1\big)-1\Big)=\eps_0>0.
\end{split}
\end{equation*}
This leads to a contradiction. Therefore, we obtain the required
viscosity supersolution property.
\end{proof}
\section{Multi-Asset Extension}
\label{sec:multi_asset_extension}

In this section, we briefly describe a multi-asset extension of the framework introduced above. We do not pursue here the full analytical study of this extension, but we present the corresponding state dynamics, optimization problem, and the multidimensional quasi-variational inequalities extending \eqref{hjb}. This formulation is particularly relevant for high-dimensional numerical methods, including recent reinforcement-learning approaches to singular stochastic control problems in finance.

Let $d\geq 1$ be the number of assets and let
$$
\midbar X=(\midbar X^1,\dots,\midbar X^d)\in \RR_+^d
$$
denote the target inventory vector. For each asset $k\in\{1,\dots,d\}$, let $A^k=(A_t^k)_{t\geq 0}$ be the unaffected price process, assumed to be a continuous $(\PP,\cF)$-martingale. We write
$$
A_t=(A_t^1,\dots,A_t^d)\in\RR^d, \quad \forall t\geq 0.
$$
The liquidity regime process $I=(I_t)_{t\geq 0}$ is the same finite-state Markov chain with values in $\I_m$ and generator $Q$ as in Section \ref{sec:model}.

For each regime $i\in\I_m$, we now consider a multidimensional impact map
$$
\Psi_i=(\Psi_i^1,\dots,\Psi_i^d):\RR_+^d\to \RR_+^d,
$$
where $\Psi_i^k(y)$ represents the price deviation of asset $k$ when the vector of volume effects is $y=(y_1,\dots,y_d)\in\RR_+^d$. This allows for cross-impact effects, since the deviation of asset $k$ may depend on the liquidity state of all assets. The impacted ask-price vector is then given by
$$
P_t=A_t+D_t,
~~
\text{and}~~
D_t=\Psi_{I_t}(Y_t),
$$
where $t\geq 0$ and $Y=(Y^1,\dots,Y^d)$ denotes the multidimensional volume-effect process.

Let $X=(X^1,\dots,X^d)$ be the cumulative purchase process. For $(t,x)\in[0,T]\times\Pi_{k=1}^d [0,\midbar X^k]$, we define $\cA_t(x)$ as the set of $\RR_+^d$-valued càdlàg, $\cF$-adapted processes $X=(X_u)_{t\leq u\leq T}$ such that each component $X^k$ is non-decreasing and satisfies
$$
X_{t^-}=x,~~
\text{and}~~
X_T=\midbar X.
$$
We denote by $X^c$ the continuous part of $X$ and by $\Delta X_u=X_u-X_{u^-}$ its jump at time $u$. 

The controlled multidimensional volume-effect process is defined by
\begin{equation}
\label{multi_controlled_volume_effect}
\left\{
\begin{array}{ll}
\mathrm{d}Y_u=\mathrm{d}X_u-h(Y_{u^-})\mathrm{d}u+\Sigma(Y_{u^-})\mathrm{d}W_u+\displaystyle\int_E q(Y_{u^-},z)M(\mathrm{d}u,\mathrm{d}z),\\
Y_{t^-}=y\in\RR_+^d,
\end{array}
\right.
\end{equation}
where
$$
h:\RR_+^d\to\RR^d,\quad
\Sigma:\RR_+^d\to\RR^{d\times r},~~
\text{and}~~
q:\RR_+^d\times E\to\RR^d,
$$
$W$ is an $r$-dimensional Brownian motion, and $M$ is a Poisson random measure on $\RR_+\times E$ with compensator $\lambda_t\nu(\mathrm{d}z)\mathrm{d}t$. As in the one-dimensional case, we denote by
$$
\check Y_{u^-}:=Y_{u^-}+\Delta_M Y_u
$$
the post-jump value of $Y$ due to the jumps of $M$ alone. For each regime $i\in\I_m$ and each asset $k\in\{1,\dots,d\}$, let $F_i^k$ denote the corresponding LOB shape and let $\psi_i^k$ be its generalized inverse (see \eqref{def_psi}). For $y=(y_1,\dots,y_d)\in\RR_+^d$, we define the multidimensional impact-cost function $\Phi_i:\RR_+^d\to\RR_+$, which extends the one-dimensional quantity introduced in \eqref{definitionPhi},
$$
\Phi_i(y):=\sum_{k=1}^d \Phi_i^k(y_k),
~~
\text{and}~~
\Phi_i^k(y_k):=\int_0^{\psi_i^k(y_k)} \xi\,\mathrm dF_i^k(\xi)
=\int_0^{y_k}\psi_i^k(\zeta)\,\mathrm d\zeta.
$$
Accordingly, if the current volume-effect state is $y\in\RR_+^d$ and a block trade of size $a\in\RR_+^d$ is submitted, then the resulting impact cost in regime $i$ is
$$
\Phi_i(y+a)-\Phi_i(y).
$$
Accordingly, for $(i,t,x,y)\in\I_m\times[0,T]\times[0,\midbar X]\times\RR_+^d$, we define the cost functional
\begin{equation}
\label{multi_cost_functional}
\begin{split}
C_i(t,x,y,X):=\EE\Bigg[
\int_t^T \sum_{k=1}^d \Psi_{I_u}^k(\check Y_{u^-}^{t,y,X})\,\mathrm{d}(X_u^k)^c
+\sum_{t\leq u\leq T}\Big(
\Phi_{I_u}(Y_u^{t,y,X})-\Phi_{I_u}(\check Y_{u^-}^{t,y,X})
\Big)
\Bigg].
\end{split}
\end{equation}
The associated value function is then
\begin{equation}
\label{multi_value_function}
v_i(t,x,y):=\inf_{X\in\cA_t(x)} C_i(t,x,y,X).
\end{equation}
The terminal and boundary conditions take the form
\begin{equation}
\label{multi_boundary_conditions}
v_i(T,x,y)=\Phi_i(y+\midbar X-x)-\Phi_i(y),
~~
\text{and}~~
v_i(t,\midbar X,y)=0,
\end{equation}
where all vector operations are understood componentwise.

Let $\varphi=\varphi(t,x,y)$ be a smooth test function defined on $[0,T]\times\RR_+^d\times\RR_+^d $. The infinitesimal generator of the uncontrolled multidimensional volume-effect process is
\begin{equation}
\label{multi_generator}
\begin{split}
\cL \varphi (t,x,y)
:=
-h(y)\cdot D_y\varphi(t,x,y)
+\frac{1}{2}\mathrm{Tr}\big(\Sigma(y)\Sigma(y)^T D^2_{yy}\varphi(t,x,y)\big)\\
+\lambda_t\int_E\Big(\varphi(t,x,y+q(y,z))-\varphi(t,x,y)\Big)\nu(\mathrm{d}z).
\end{split}
\end{equation}
The multidimensional HJBQVI extending \eqref{hjb} is therefore
\begin{equation}
\label{multi_hjb}
\max\Big(
-\frac{\partial v_i}{\partial t}
-\cL v_i
-\sum_{j\neq i}(v_j-v_i)Q_{ij},
\,
\max_{1\leq k\leq d}
\big\{
-\frac{\partial v_i}{\partial x_k}
-\frac{\partial v_i}{\partial y_k}
-\Psi_i^k
\big\}
\Big)=0,
\end{equation}
for all $i\in\I_m$ and  on $[0,T)\times \prod_{k=1}^d [0,\midbar X^k) \times \RR_+^d.$
\\
As in the one-dimensional case, the first term in \eqref{multi_hjb} corresponds to the continuation region, while the gradient constraints characterize the intervention region. More precisely, an infinitesimal purchase in asset $k$ becomes optimal when
$$
-\frac{\partial v_i}{\partial x_k}
-\frac{\partial v_i}{\partial y_k}
-\Psi_i^k=0.
$$
Hence, in the multi-asset setting, the exercise region is determined by the set of points where at least one of these constraints is binding.

\end{document}